\documentclass[a4paper,11pt]{amsart}
\usepackage{mathrsfs}
 \usepackage[all]{xy}
\usepackage{amsmath,amssymb,amscd,bbm,amsthm,mathrsfs,dsfont}
\newtheorem{thm}{Theorem}[section]

\newtheorem{cor}{Corollary}[section]
\newtheorem{prop}{Proposition}[section]
\newtheorem{rem}{Remark}[section]

\begin{document}
\numberwithin{equation}{section}

 \title[On   $3$-gauge   transformations, $3$-curvatures and $\mathbf{Gray}$-categories]{
On   $3$-gauge   transformations, $3$-curvatures and $\mathbf{Gray}$-categories}
\author {Wei Wang}
\begin{abstract} In the $3$-gauge theory, a   $3$-connection  is given by  a $1$-form $A$ valued in the Lie algebra $ \mathfrak g$, a $2$-form $B$ valued in the Lie algebra $\mathfrak h $ and  a $3$-form $C$ valued in the Lie algebra $ \mathfrak l $, where $(\mathfrak g,\mathfrak h, \mathfrak l)$ constitutes a differential $ 2$-crossed module. We give the $3$-gauge transformations from a $3$-connection     to another, and show the transformation formulae of   the $1$-curvature $2$-form,  the $2$-curvature $3$-form and the $3$-curvature $4$-form.  The gauge configurations can be interpreted as smooth $\mathbf{Gray}$-functors    between two $\mathbf{Gray}$ $3$-groupoids: the path $3$-groupoid $\mathcal{P}_3(X)$ and the  $3$-gauge group  $ \mathcal{G}^{\mathscr L}$  associated to the $ 2$-crossed module $\mathscr L$, whose differential is $(\mathfrak g,\mathfrak h, \mathfrak l)$.  The derivatives of $\mathbf{Gray}$-functors are $3$-connections, and the derivatives  of
lax-natural transformations between two such $\mathbf{Gray}$-functors are $3$-gauge transformations.    We   give the $3$-dimensional  holonomy, the lattice version of the $3$-curvature, whose derivative
gives the $3$-curvature $4$-form. The covariance   of   $3$-curvatures  easily follows from this construction. This $\mathbf{ Gray}$-categorical construction explains why $3$-gauge transformations and $3$-curvatures have the given forms. The interchanging $3$-arrows
are responsible for the appearance of terms concerning the Peiffer commutator $\{,\}$.
\end{abstract}

\thanks{Supported by National Nature Science
Foundation
  in China (No. 11171298)}
  \thanks{Department of Mathematics,
Zhejiang University, Zhejiang 310027,
 P. R. China, Email:   wwang@zju.edu.cn}
 \maketitle

\section{Introduction}

String theory and $M$-theory involve various higher gauge fields, such as the $B$-field  in string theory and the $C$-field  in $11$-dimensional $M$-theory.   They are locally given
by differential form fields of higher degree and   are globally modeled by higher
bundles with connections (higher gerbes with connections, higher differential characters)
(cf. \cite{SSS09} \cite{SSS} and references therein). In general, the   extended $n$-dimensional relativistic objects appearing in string theory
are usually coupled to background fields, which can naturally be
 the $n$-categorical version of fiber bundles with connections.

For nonabelian bundle gerbes  \cite{ACJ} or more generally principal $2$-bundles \cite{W}, there exists a
 framework of differential geometry:  $2$-connections and $2$-curvatures (cf. \cite{ACJ} \cite{BS} \cite{BH11}
  and references therein). There are lattice and differential formulations of the $2$-gauge theory \cite{GP04} \cite {P03}, which can be applied to some $M$-brane models \cite{PS},  BF theory \cite{GPP} \cite{MM11} and non-Abelian
self-dual tensor field theories \cite{SW12}, etc.. The next step is to develop the $3$-gauge theory, $3$-connections and $3$-curvatures for $3$-bundles or bundles $2$-gerbes
 \cite{Br94}  \cite{Br10} \cite{St} \cite{Ju}. The $3$-form
gauge potentials have already appeared in physics (cf. \cite{FSS} \cite{SW13} and references therein).

Recall that the lattice gauge theory can be formulated in   language of categories. Let $(V,E)$ be a directed graph, given by  a set $V$ of vertices and a set $E$ of
  edges.  Let $\mathcal{C}^{V ;E }$ be the associated category: the vertices as objects and the edges as arrows. Then   configurations of  lattice gauge theory are the  functors from the category $\mathcal{C}^{V ;E }$
to the   gauge group $ \mathcal{G}^{G }$,  the groupoid associated to the Lie group $G$  with one object. A gauge transformation  is a natural transformation  from one functor to another.

This construction was generalized to the  $2$-lattice gauge theory by Pfeiffer in \cite{P03}.  Consider a {\it  simplicial $2$-complex} $(V ; E; F )$   consisting of sets $V$, $E$ as above
 and a set $F$ of faces.
 There is an associated small $2$-category $\mathcal{C}^{V ;E;F }$: the vertices as objects, the edges as arrows and the faces as
    $2$-arrows. The Lie group is replaced by a crossed module $ \mathscr H= (\alpha: H \rightarrow G,\rhd)$, from which we can construct  a strict Lie $2$-group $\mathcal{G}^{\mathscr H }$. Then
  the configurations of the   $2$-lattice gauge theory are   $2$-functors from the generalized lattice $\mathcal{C}^{V ;E;F }$
to the $2$-gauge group $ \mathcal{G}^{\mathscr H}$, i.e., the edges are coloured by group elements of $G$  and the   faces are
coloured by elements of $G\times H$. A gauge transformation  is a pseudonatural transformation  from one $2$-functor to another. If we take the length of lattice tending to zero, we get the differential $2$-gauge theory \cite{GP04} .

To define the $3$-lattice gauge theory, we need to replace  a crossed module by a
  $2$-crossed module $\mathscr L$, which  is given by a complex of Lie groups:
\begin{equation}\label{eq:2-crossed-module}
     L \xrightarrow{\delta}H \xrightarrow{\alpha} G,
\end{equation}
with smooth left actions $\rhd$ of $G$  on $L$ and $H$ by automorphisms
 and a $G$-equivariant
smooth function (the Peiffer lifting) $\{,\}: H \times H \rightarrow L$,
and to construct the associated  $\mathbf{Gray}$ $3$-groupoid  $\mathcal{G}^{\mathscr L }$.
A {\it  simplicial $3$-complex} $(V ; E; F;T)$   consists of sets $V$, $E$, $F$  as above and a set $T$ of tetrahedrons.
 There is an associated small tricategory $\mathcal{C}^{V ;E;F;T}$:   objects,  $1$-arrows and $2$-arrows  as above,  and
 the tetrahedrons  as
    $3$-arrows.
  The configurations of our   $3$-lattice gauge theory will be    functors from the generalized lattice $\mathcal{C}^{V ;E;F;T}$
to the $3$-gauge group $ \mathcal{G}^{\mathscr L}$. Namely,  the edges and faces are
coloured as before and the  tetrahedrons are
coloured by elements of $G\times H\times L$. A gauge transformation  is a  lax-natural transformation  from one functor to another.

There is another similar, but more mathematical approach to this construction. There exists a bijection between connections and functors (play the role of holonomies)  \cite{SW08}:
\begin{equation*}\label{eq:bijection}
     \Lambda^1(X,\mathfrak g)\cong \left\{{\rm smooth \hskip 1mm functors\hskip 1mm} \mathcal{P}_1(X)\rightarrow \mathcal{G}^{G  }\right\},
\end{equation*}
where $\mathcal{P}_1(X)$ is the path groupoid of the manifold $X$, and $\Lambda^k(X,\mathfrak g)$ is
the set of $\mathfrak g$-valued $k$-forms on $X$. This is generalized to the $2$-gauge theory by Schreiber  and Waldorf \cite{SW11}: there exists a bijection  between $2$-connections and $2$-functors (play the role of $2$-dimensional holonomies):
\begin{equation*}\label{eq:bijection-2}
   \left \{{\rm smooth} \hskip 1mm  2{\rm -functors\hskip 1mm} \mathcal{P}_2(X)\rightarrow \mathcal{G}^{\mathscr H  }\right\}\cong\{ A\in   \Lambda^1(X,\mathfrak g), B\in\Lambda^1(X,\mathfrak h);dA+A\wedge A=\alpha(B) \} ,
\end{equation*}
where $\mathcal{P}_2(X)$ is the path $2$-groupoid  of a  manifold $X$ by adding $2$-arrows to the path groupoid $\mathcal{P}_1(X)$. See also \cite{MP02} \cite{MP10} \cite{MP11-2} for $1$- and $2$-dimensional   holonomies.

To construct the $3$-gauge theory, we will consider the path $3$-groupoid $\mathcal{P}_3(X)$   by adding $3$-arrows to the path $2$-groupoid $\mathcal{P}_2(X)$, and smooth  $\mathbf{Gray}$-functors (play the role of $3$-dimensional holonomies) from $\mathcal{P}_3(X)$  to $ \mathcal{G}^{\mathscr L}$ \cite{MP11}. Locally, a {\it $3$-connection} on an open set $U$ of $\mathbb{R}^n$   is a triple $(A,B,C)$ with  $A\in   \Lambda^1(U,\mathfrak g), B\in\Lambda^2(U,\mathfrak h)$ and  $C\in   \Lambda^3(U,\mathfrak l)$. A $3$-gauge transformation from a $3$-connection   $(A,B,C)$ to another one $(A',B',C')$ is given by
     \begin{equation}\label{eq:gauge-transformations}
     \begin{split}
       &          {A} '=Ad_{g^{-1} } A+g^{-1}dg+\alpha(\varphi )
      \\&
          {B}' = g^{-1} \rhd B+d\varphi+A'\wedge^{\rhd}\varphi-\varphi\wedge\varphi-\delta (\psi)
      \\&
         {C }'= g^{-1}  \rhd C - d\psi-A'\wedge^{\rhd }\psi +\varphi\wedge^{\rhd'}\psi- {B}'\wedge^{\{,\}}\varphi-\varphi\wedge^{\{,\}} (g^{-1} \rhd{B}),
     \end{split}\end{equation}
     for some $g\in \Lambda^0(U,G)$, $\varphi\in   \Lambda^1(U,\mathfrak h), \psi\in\Lambda^2(U,\mathfrak l)$.
Here $\mathfrak l \xrightarrow{\delta}\mathfrak h \xrightarrow{\alpha}\mathfrak g$ is  a   differential $ 2$-crossed module with smooth left actions $\rhd$ of $\mathfrak g$ on $\mathfrak h $  and    $\mathfrak l$  by automorphisms   and a left action $  \rhd'  $ of $\mathfrak h $  on $\mathfrak l$ (cf. \S 2  for notations).

The {\it $1$-curvature $2$-form,  $2$-curvature $3$-form}  and {\it $3$-curvature $4$-form} are defined as
\begin{equation}\label{eq:3-curvature}\begin{split}
   & \Omega_1:=dA+A\wedge A,\\&
    \Omega_2:=dB+A\wedge^{\rhd} B, \\&
     \Omega_3:=dC+A\wedge^{\rhd} C+B\wedge^{\{\cdot\}} B,
\end{split}\end{equation} respectively.
Under the $3$-gauge transformation (\ref{eq:gauge-transformations}), these curvatures transform as follows:
\begin{equation}\label{eq:gauge-transformations-curvature}\begin{split}
&\Omega_1'   =g^{-1} \rhd\Omega_1+  \alpha(  {B}')-\alpha(g^{-1} \rhd B),
  \\ &\Omega_2'   =g^{-1} \rhd\Omega_2+ [\Omega_1' - \alpha(     B')] \wedge^{\rhd}\varphi+ \delta(  {C}')-\delta(g^{-1} \rhd C),\\&
   \Omega_3'=g^{-1} \rhd
  \Omega_3   -[\Omega_2'-\delta( C')]\wedge^{\{,\}}\varphi+\varphi\wedge^{\{,\}}[g^{-1} \rhd(\Omega_2-\delta(   C ))]-[\Omega_1' - \alpha(    B')] \wedge^{\rhd }\psi.
\end{split}\end{equation}

 We define the {\it fake $1$-curvature} to be $\mathcal{F}_1=\Omega_1 - \alpha( B)$ and the {\it fake $2$-curvature} to be $\mathcal{F}_2=\Omega_2 - \delta( C)$. Then the $3$-curvature $4$-form is covariant under the gauge transformations (\ref{eq:gauge-transformations}) if the fake $1$-  and fake $2$-curvatures vanish.

In section 3, we give an elementary proof of the transformation formulae
(\ref{eq:gauge-transformations-curvature}) of curvatures. This proof, having nothing to do with $\mathbf{Gray}$-categories,
is based on some properties of actions of $\rhd$ and $\{,\}$ on Lie algebra valued differential forms. which are established in section 2.

$\mathbf{Gray}$-categories are
  semi-strict tricategories.
In a $\mathbf{Gray}$-category  there are   two possible
ways of composing two $2$-arrows horizontally
\begin{equation}\label{eq:interchanging-def}\xy
(-8,0)*+{C}="4";
(8,0)*+{C'}="6";
(24,0)*+{C}="8";
{\ar@{->}|-{f'} "4";"6"};
{\ar@/^1.85pc/^{f } "4";"6"};
{\ar@{->}|-{g } "6";"8"};
{\ar@/_1.85pc/_{g' } "6";"8"};{\ar@{=>}^<<{\scriptstyle \gamma} (0,6)*{};(0,1)*{}} ;
{\ar@{=>}^<<{\scriptstyle \delta} (16,-1)*{};(16,-6)*{}} ;
\endxy\Longrightarrow
\xy
(-8,0)*+{C}="4";
(8,0)*+{C'}="6";
(24,0)*+{C}="8";
{\ar@{->}|-{f} "4";"6"};
{\ar@/_1.85pc/_{f' } "4";"6"};
{\ar@{->}|-{g' } "6";"8"};
{\ar@/^1.85pc/^{g } "6";"8"};{\ar@{=>}^<<{\scriptstyle \gamma} (0,-1)*{};(0,-6)*{}} ;
{\ar@{=>}^<<{\scriptstyle \delta} (16,6)*{};(16,1)*{}} ;
\endxy
\end{equation}($ (\gamma \#_0 g) \#_1( f'\#_0 \delta)$ and $(f \#_0\delta) \#_1 (\gamma\#_0  g')$, cf. \S \ref{sub:Gray-3-category}),
which are the $2$-source and the $2$-target of an interchanging $3$-arrow, while in a $2$-category two such horizontal compositions are
identical   (cf. \S 2.13 in \cite{KV} and references therein for the pasting theorem for $2$-categories). This is an essential difference between $2$-categories and $\mathbf{Gray}$-categories. For the $2$-crossed module $\mathscr L$ in (\ref{eq:2-crossed-module}), $H \xrightarrow{\alpha} G$ is no longer a crossed module in general. The Peiffer lifting $\{,\}:H\times H\rightarrow L$ measures its failure to be a crossed module.
The interchanging $3$-arrows in $ \mathcal{G}^{\mathscr L}$ are given by $\{,\}$, and
are responsible for the appearance of terms $\{,\}$ in our formulae of gauge transformations and curvatures.

In section 4, we recall the definitions of  a $\mathbf{ Gray}$-category, the Gray $3$-groupoid constructed from the $ 2$-crossed module $\mathscr L$ and
lax-natural transformations between two $\mathbf{Gray}$-functors.

In section 5, for a given   lax-natural transformation between two $\mathbf{Gray}$-functors $F$ and $\widetilde{F}$ from the path $3$-groupoid $\mathcal{P}_3(X)$  to  the $3$-gauge group   $ \mathcal{G}^{\mathscr L}$,    the naturality of the lax-natural transformation gives us  an equation with $3$-parameters. We write down explicitly each side of the equation
as the composition of several $3$-arrows in the $\mathbf{Gray}$ $3$-groupoid $ \mathcal{G}^{\mathscr L}$. Then take the derivatives with respect to the parameters at the origin to get the   gauge  transformation formula for the $C$ field in (\ref{eq:gauge-transformations}). The  same is done for the $A$ and $B$ fields.
In this construction, we must have  $\mathcal{F}_1=\mathcal{F}_2=0$.

Similarly in section 6, we consider a $4$-path   $ {\Theta}$, whose boundary $\partial {\Theta}$ can be viewed
as the composition  of several $3$-paths. For a $\mathbf{Gray}$-functor  $F$ from   the path $3$-groupoid $\mathcal{P}_3(X)$  to  the $3$-gauge group   $ \mathcal{G}^{\mathscr L}$, $F(\partial {\Theta})$ is its $3$-dimensional  holonomy, the lattice version of the $3$-curvature. We write it explicitly
as a composition of several $3$-arrows in the  $3$-gauge group   $ \mathcal{G}^{\mathscr L}$, and take the derivatives with respect to the parameters at the origin to get the expression of the $3$-curvature in (\ref{eq:3-curvature}). The covariance   of the $3$-holonomies under   $3$-gauge transformations is easily found from this construction.

The correct definition of the $3$-curvature on a manifold $X$ was first appeared in \cite{MP11} as the $2$-curvature of a $2$-connection on the loop space of $ X $. The authors of \cite{MP11} used it to define the $3$-dimensional  holonomy, but did not discuss the $3$-gauge transformations.
When this paper was almost finished, we found the preprint \cite{SW13}, where the authors studied the Penrose-Ward transformation between solutions to self-dual $3$-gauge fields on the flat $6$-dimensional space $M$ and $M$-trivial holomorphic principal $3$-bundles over the
twistor space $\mathbb{CP}^6$ (they consider the supersymmetric version). By writing down  gauge  transformations of relatively flat  $3$-gauge fields on the correspondence space, they found the  gauge  transformations of $3$-gauge fields and  transformation formulae of curvatures on the $6$-dimensional space-time $M$. Then they claimed the  gauge  transformations of $3$-gauge fields and the     transformation formulae of curvatures in any dimension   (cf. (5.17), (5.22) in \cite{SW13}). The authors informed me that their method, based on the equivalence of the
Cech and Dolbeault pictures and solving  the Riemann-Hilbert problem, actually works  in any dimension  and   the proofs of (5.17) and (5.22) in \cite{SW13} is similar to the $6$-dimensional case.  In this paper we give a detailed proof of the gauge  transformations of curvatures. Moreover, with the help of   the Gray $3$-groupoid constructed from a $ 2$-crossed module, we see why   gauge  transformations   and     $3$-curvatures are given by (\ref{eq:gauge-transformations}) and (\ref{eq:3-curvature}), respectively.

In this paper, we only consider the local $3$-gauge theory. See \cite{SW13} for a discussion of $3$-connections on principal $3$-bundles and  their transformations under coordinate   transformations.

I would like to thank   anonymous referees for their careful reading and many valuable suggestions.
\section{$2$-crossed modules and Lie algebra valued differential  forms}
\subsection{$2$-crossed modules and differential $2$-crossed modules}
A {\it pre-crossed module} $ \mathcal{G} = (\alpha:H \rightarrow G,\rhd)$  of Lie groups  is given by a homomorphism
$\alpha : H  \rightarrow G $ between Lie groups, together with a smooth left action $\rhd$ of $G$ on $H $ by automorphisms such that
  $\alpha(g \rhd e) = g\alpha(e)g^{-1}$, for each $ g \in G$ and $e \in H $.
The {\it  Peiffer commutators} in a pre-crossed module are defined as
\begin{equation*}
     \langle e,f\rangle=ef e^{-1}\left(\alpha(e)\rhd f^{-1}\right),
\end{equation*}
 for any $e, f \in H $. A pre-crossed module is said to be a {\it crossed module} if all of its Peiffer commutators are trivial, i.e.
\begin{equation*}
\alpha(e)\rhd f  =ef e^{-1}.
\end{equation*}

A {\it $2$-crossed module}  of Lie groups  is given by a complex (\ref{eq:2-crossed-module}) of Lie groups
together with smooth left actions $\rhd$ of $G$  on $L$ and $H$  by automorphisms  (and on $G$ by conjugations), i.e.,
\begin{equation}\label{eq:automorphisms}
g\rhd(e_1e_2)=g\rhd e_1\cdot  g\rhd e_2 ,\qquad (g_1g_2)\rhd e= g_1\rhd (g_2 \rhd e),
\end{equation}
for any $g,g_1,g_2\in G, e, e_1,e_2\in H$ or $ L$,
 and a $G$-equivariant
smooth function $\{,\}: H \times H \rightarrow L$, the {\it Peiffer lifting}, such that
\begin{equation}\label{eq:G-Peiffer}
     g\rhd \{e, f \} = \{g \rhd e,g \rhd f \},
\end{equation}
for any $g \in G$ and $e, f \in H$.
They    satisfy:

1. $L \xrightarrow{\delta}H \xrightarrow{\alpha} G $ is a complex of $G$-modules, namely, $\delta$ and $\alpha$ are $G$-equivariant and $\alpha\circ \delta$ maps $L$ to $1_G$, the identity of $G$;

2. For each $e, f \in H$, we have $\delta\{e, f \} =\langle e, f \rangle$;

3. For each $l,k \in L$, we have  $ [l,k] = \{\delta (l), \delta(k)\}$, where $[l,k] =lkl^{-1}k^{-1}$;

4. For each $e, f , g\in H$, we have $ \{ef , g\} = \{e, f gf^{-1}\}\alpha(e) \rhd \{ f , g\}$;

5. For $e, f , g\in H$, we have $\{e, f g\} = \{e, f \}\{e, g\}\{\langle e, g\rangle^{-1}, \alpha(e) \rhd f \}$;

6. For each $  e\in H$  and $ l \in L$, we have  $ \{\delta(l), e\}\{e, \delta(l)\} =l(\alpha(e)\rhd l^{-1})$.

Define
\begin{equation*}
     e\rhd' l =l
\{
\delta(l)^{-1}, e\},
\end{equation*}
  where $l \in L $ and $e \in H$.
It is known   that $\rhd'$ is a left action of $H$ on $L$ by automorphisms.  This  together with the homomorphism $\delta : L\rightarrow H$  defines a crossed module \cite{MP11}. In particular, for any $h\in H$,
  \begin{equation}\label{eq:brackt-1}
  h\rhd'  1_L= \{1_H , h \}= \{h,1_H\} =1_L
  \end{equation} (Lemma 1.4 in \cite{MP11}),
  where $1_H$ and $1_L$ are the identity of $H$ and $L$, respectively.

A {\it differential  crossed module} is given by a homomorphism of Lie algebras
\begin{equation*}
     \mathfrak e \xrightarrow{\partial}\mathfrak g,
\end{equation*}
together with the smooth left action  $\rhd$ of $\mathfrak g$ on $\mathfrak e $   by automorphisms   (and on $\mathfrak g$ by the adjoint representation),   such that

1. For each $x \in   \mathfrak g $ and each $u,v \in   \mathfrak e $, we have $x\rhd [u ,v]=[x\rhd u ,v]+[u ,x\rhd v]$;

2.  For each $x ,y\in   \mathfrak g $ and each $u \in   \mathfrak e $, we have $[x,y]\rhd u  = x \rhd (y \rhd u )- y \rhd (x \rhd u ) $;

3. For each $x \in   \mathfrak g $ and each $u \in   \mathfrak e $, we have $\partial (x\rhd u )= [x ,\partial( u )]$;

4. For each     $u ,v\in   \mathfrak e $, we have $\partial (u)\rhd v= [u ,v]$.

A {\it differential $ 2$-crossed module} is given by a complex of Lie algebras
\begin{equation*}
    \mathfrak l \xrightarrow{\delta}\mathfrak h \xrightarrow{\alpha}\mathfrak g,
\end{equation*}
together with smooth left actions $\rhd$ of $\mathfrak g$ on $\mathfrak h $  and    $\mathfrak l$  by automorphisms   (and on $\mathfrak g$ by the adjoint representation), and a $\mathfrak g$-equivariant
smooth function $\{,\}: \mathfrak h \times\mathfrak h\rightarrow \mathfrak l$, the {\it Peiffer lifting},  such that

1. $\mathfrak l \xrightarrow{\delta}\mathfrak h \xrightarrow{\alpha}\mathfrak g$ is a complex of $\mathfrak g$-modules;

2. For each $u, v \in \mathfrak h$, we have $\delta\{u, v \} =\langle u, v \rangle$, where $\langle u, v \rangle=[u,v]-\alpha(u)\rhd v$;

3. For each $x,y\in   \mathfrak l $, we have $ [x,y] = \{\delta (x), \delta(y)\}$;

4. For each $u,v,w  \in \mathfrak h$, we have $ \{[u,v] , w\} = \alpha(u)\rhd\{ v  , w\} + \{ u,[v , w]\}-\alpha(v)\rhd\{ u  , w\} - \{ v,[u , w]\}$;

5. For each $u,v,w  \in \mathfrak h$, we have $\{ u,[v , w]\}=\{\delta\{u, v \}, w \}- \{\delta\{u, w \},v \}$;

6. For each $  x\in \mathfrak l$  and $ v \in \mathfrak h$, we have $ \{\delta(x), v\}+\{v,\delta(x) \}=-\alpha(v) \rhd  x $.

 \begin{prop}\label{prop:rhd'}{\rm(cf., e.g., Lemma 1.9 in \cite{MP11})}
     $v \rhd' x =-\{\delta(x), v\}$, for $  v \in \mathfrak h$ and $ x\in \mathfrak l$, defines a left action of $\mathfrak h $  on $\mathfrak l$,
  which together with the homomorphism  $ \delta : \mathfrak l\rightarrow \mathfrak h $ defines a differential crossed module.
 \end{prop}

\subsection{Lie algebra valued differential  forms}

Given a Lie algebra $\mathfrak k$, we denote by $\Lambda^k(U, \mathfrak k)$ the vector space of $\mathfrak k$-valued differential $k$-forms on $U$. For $K \in \Lambda^{r }(U, \mathfrak k) $, we can write $K=\sum_a K^aX_a $ for some scalar differential $k$-forms $K^a $  and elements $ X_a\in \mathfrak k$ (here $\{X_a\}$ need not to be a basis). We will choose $\mathfrak k$ to be $\mathfrak g $ or $\mathfrak h$ or $\mathfrak l$. Here we assume $\mathfrak k$  to be a   matrix Lie algebra. Thus we have $[X,X']=XX'-X'X$ for $ X,X'\in \mathfrak k$.

For $K=\sum_a K^aX_a\in \Lambda^{r }(U, \mathfrak k),M=\sum_b M^bX_b\in \Lambda^{t }(U, \mathfrak k)$, define
\begin{equation*}
     K\wedge M: =\sum_{a,b} K^a\wedge M^b X_a X_b,\qquad
     K\wedge^{[,]} M: =\sum_{a,b} K^a\wedge M^b [X_a, X_b],
\end{equation*}and define
\begin{equation*}
     dK=\sum dK^a X_a.
\end{equation*}
   For   $\Psi_j =\sum_b \Psi_j^b Y_b \in \Lambda^{k_j} (U, \mathfrak h)$, $j=1,2,$ where $ Y_b\in \mathfrak h$, define
\begin{equation}\label{eq:rhd}
     K \wedge^\rhd \Psi_j: =\sum_{a,b} K^a\wedge \Psi_j^b X_a\rhd Y_b,
\qquad
     \Psi_1\wedge^{\{,\}} \Psi_2:=\sum_{a,b} \Psi_1^a \wedge\Psi_2^b {\{Y_a,Y_b\}},
\end{equation}where $K$ is valued in $\mathfrak g$.
In a similar way   define $\Psi_1\wedge^{\langle,\rangle} \Psi_2,  K \wedge^\rhd \psi$ and $ \Psi_j  \wedge^{\rhd '}\psi$ for $\psi \in \Lambda^{* }(U, \mathfrak l)$. It is easy to see that the definitions above are independent of the choice of the expressions $K=\sum K^aX_a, $ etc., by linearity.
Here we use the notations in \cite{MP11}. In \cite{SW13}, $ \Psi_1\wedge^{\{,\}} \Psi_2$ is written as $ \{\Psi_1 ,\Psi_2\}  $.

The following proposition gives   properties for Lie algebra valued differential  forms corresponding to identities in the definition of a
differential $ 2$-crossed module.
    \begin{prop}\label{prop:2cross-form}  For  $\Psi\in \Lambda^k(U,\mathfrak h)$, $ {\Psi}'\in \Lambda^{  k'}(U,\mathfrak h)$, $\Phi\in \Lambda^1(U,\mathfrak h)$ and $\psi\in \Lambda^s(U,\mathfrak l)$, we have

    (1)
     $
        \alpha (  {\Psi}')\wedge^{\rhd}  \Psi
        =  {\Psi}'\wedge^{[,]}\Psi - {\Psi}' \wedge^{\langle,\rangle}  \Psi.
   $

(2)
       $-\alpha (\Psi)\wedge^{\rhd }\psi=
         \Psi \wedge^{\{\cdot\}} \delta (\psi) +(-1)^{ks}\delta (\psi) \wedge^{\{\cdot\}}\Psi.
     $

(3)
        $
         \delta ( \psi)\wedge^{\{,\}}\Psi=-(-1)^{ks} \Psi\wedge^{\rhd' }\psi.
     $

(4)
  $
     \Psi\wedge^{\{,\}}(\Phi\wedge\Phi)  = (\Psi \wedge^{ \langle,\rangle}   \Phi)\wedge^{\{,\}}\Phi.
  $

(5)
  $
 (\Phi\wedge\Phi)  \wedge^{\{,\}}  \Psi =\Phi\wedge^{\{,\}} ( \Phi\wedge^{[,]}  \Psi)+\alpha (\Phi)\wedge^\rhd( \Phi \wedge^{\{,\}}\Psi).
  $    \end{prop}
    \begin{proof}
(1) Write  $\Psi =\sum_a \Psi^a Y_a, {\Psi}' =\sum_b  (\Psi')^b Y_b  \in \Lambda^k (U, \mathfrak h)$. Then we have
        \begin{equation*}\begin{split}
        \alpha ( \Psi')\wedge^{\rhd}  \Psi &=\sum_{a,b} (\Psi')^b\wedge\Psi^a  \alpha(   Y_b) \rhd  Y_a =\sum_{a,b} (\Psi')^b\wedge\Psi^a   ([Y_b, Y_a ]-\langle Y_b, Y_a \rangle)\\&
        =  {\Psi}'\wedge^{[,]}\Psi - {\Psi}' \wedge^{\langle,\rangle}  \Psi
  \end{split}\end{equation*}
  by using 2. in the definition of a differential $ 2$-crossed module.

(2) Write $   \psi  =\sum_c \psi^c Z_c \in \Lambda^k (U, \mathfrak l)$. Then,
\begin{equation*}\begin{split}
 -\alpha (\Psi)\wedge^{\rhd}\psi&=-\sum_{a,c} \Psi^a\wedge\psi^c  \alpha (Y_a)\rhd Z_c
  =\sum_{a,c} \Psi^a\wedge\psi^c (\{Y_a, \delta(Z_c)\}+\{\delta(Z_c), Y_a\})
  \\&=       \Psi \wedge^{\{\cdot\}} \delta (\psi) +(-1)^{ks}\delta (\psi) \wedge^{\{\cdot\}}\Psi
   \end{split}  \end{equation*}
by 6. in the definition of a differential $ 2$-crossed module.

(3)
   Let $\Psi$ and $\psi$ be as above. \begin{equation*}\begin{split}
        \delta( \psi)\wedge^{\{,\}}\Psi&=\sum_{a,c}  \psi^c\wedge\Psi^a\{\delta(Z_c), Y_a\}
        =-\sum_{a,c} \psi^c\wedge\Psi^a Y_a\rhd' Z_c\\&=-(-1)^{ks}  \sum_{a,c}  \Psi^a\wedge\psi^c Y_a\rhd' Z_c=-(-1)^{ks} \Psi\wedge^{\rhd' }\psi
     \end{split} \end{equation*}
    by  the definition of $\rhd'$ in Proposition \ref{prop:rhd'}.

(4)  Write $   \Phi =\sum_b\Phi^b Y_b \in \Lambda^k (U, \mathfrak h)$. Then,
 \begin{equation*}\begin{split}
      \Psi\wedge^{\{,\}}(\Phi\wedge\Phi)&=\sum_{a,b,c } \Psi^a \wedge\Phi^b\wedge\Phi^c  {\{Y_a,Y_bY_c\}} = \sum_{a, b<c} \Psi^a \wedge \Phi^b\wedge\Phi^c  {\{Y_a,[Y_b,Y_c]\}}\\&
      =\sum_{a,b<c} \Psi^a\wedge \Phi^b\wedge\Phi^c  ( \{\delta\{Y_a, Y_b\},Y_c \} - \{\delta\{Y_a, Y_c\},Y_b \} )\\&
      =\sum_{a,b,c } \Psi^a\wedge \Phi^b\wedge\Phi^c   \{\langle Y_a, Y_b\rangle,Y_c \} =(\Psi \wedge^{ \langle,\rangle}   \Phi)\wedge^{\{,\}}\Phi,
   \end{split}\end{equation*}
by 5. in the definition of a differential $ 2$-crossed module. Here $\Phi^b\wedge\Phi^c =-\Phi^c\wedge \Phi^b$ since $\Phi^*$'s are $1$-forms.

(5) Let $\Psi$ and $\Phi$ be as above. Then,
\begin{equation*}\begin{split}
      (\Phi&\wedge\Phi)\wedge^{\{,\}}\Psi=\sum_{a,b,c } \Phi^b\wedge\Phi^c\wedge \Psi^a {\{Y_bY_c ,Y_a\}} = \sum_{a  ,b<c} \Phi^b\wedge\Phi^c\wedge \Psi^a {\{[Y_b,Y_c],Y_a\}}\\&
      =\sum_{a ,b<c} \Phi^b\wedge\Phi^c\wedge \Psi^a \left( \{Y_b ,[ Y_c,Y_a] \}-\{Y_c ,[ Y_b,Y_a] \}+\alpha (Y_b) \rhd  \{  Y_c,Y_a  \}-\alpha (Y_c) \rhd  \{  Y_b,Y_a  \}\right) \\&
      =\sum_{a,b,c } \Phi^b\wedge\Phi^c\wedge \Psi^a ( \{Y_b ,[ Y_c,Y_a] \} +\alpha (Y_b) \rhd  \{  Y_c,Y_a  \} ) \\&=\Phi\wedge^{\{,\}} ( \Phi\wedge^{[,]}  \Psi)+\alpha (\Phi)\wedge^\rhd( \Phi \wedge^{\{,\}}\Psi),
   \end{split}\end{equation*}
by 4. in the definition of a differential $ 2$-crossed module.
    \end{proof}

 \begin{prop}\label{prop:2cross-form-2} (1) For $  K  \in \Lambda^{t}(U,\mathfrak g)$, $  M  \in \Lambda^{s}(U,\mathfrak g)$  and $  \psi  \in \Lambda^{*}(U,\mathfrak l)$,
\begin{equation} \label{eq:wedge-rhd}\begin{split}
    \delta (K \wedge^{\rhd}\psi)&= K \wedge^{\rhd}  \delta (\psi),
\\
     K\wedge^{\rhd} M&=K\wedge M-(-1)^{ts} M \wedge K.\end{split}
\end{equation}
   (2)   For $  K  \in \Lambda^{t}(U,\mathfrak g)$, $  \Psi_j\in \Lambda^{k_j}(U,\mathfrak h)$, $j=1,2$, and $ \gamma\in \Lambda^{*}(U,\mathfrak k)$ ($\mathfrak k=\mathfrak g$, $\mathfrak h$, $\mathfrak l$),
      \begin{equation*}\begin{split}
     d ( K\wedge^{\rhd}  \gamma) &=   d K \wedge^{\rhd}  \gamma+(-1)^{t} K\wedge^{\rhd}   d \gamma ,\\ d ( \Psi_1\wedge^{\{\cdot\}}  \Psi_2)&=  (d \Psi_1)\wedge^{\{\cdot\}}  \Psi_2+(-1)^{k_1} \Psi_1\wedge^{\{\cdot\}}   d \Psi_2 ,\\
    K \wedge^{\rhd }( \Psi_1\wedge^{\{\cdot\}}  \Psi_2)&=  ( K \wedge^{\rhd }\Psi_1)\wedge^{\{\cdot\}}  \Psi_2+(-1)^{tk_1} \Psi_1\wedge^{\{\cdot\}} ( K\wedge^{\rhd } \Psi_2),\\
    K \wedge^{\rhd }( \Psi_1\wedge   \Psi_2)&=  ( K \wedge^{\rhd }\Psi_1)\wedge   \Psi_2+(-1)^{tk_1} \Psi_1\wedge  ( K\wedge^{\rhd } \Psi_2).
   \end{split}\end{equation*}
(3)    For $K,M\in   \Lambda^1(U,\mathfrak g)$, $\psi\in \Lambda^2(U,\mathfrak l)$,
      \begin{equation*}
         K\wedge^{\rhd }(M \wedge^{\rhd }\psi)+ M\wedge^{\rhd }( K\wedge^{\rhd }\psi)= (K\wedge M + M \wedge K )\wedge^{\rhd }\psi.
      \end{equation*}
(4)
For $  \psi \in \Lambda^2(U,\mathfrak l)$ and $  \Psi  \in \Lambda^{1}(U,\mathfrak h)$,
\begin{equation}\label{eq:vanish}
       \delta (\psi) \wedge^{\{\cdot\}} \delta (\psi)=  0,\qquad
     (\Psi\wedge\Psi )\wedge^{ [,] } \Psi=  0.
     \end{equation}
    (5) (Equivariance) For $K \in   \Lambda^*(U,\mathfrak g)$, $\Psi \in  \Lambda^*(U,\mathfrak h)$ and $\Upsilon\in \Lambda^*(U,\mathfrak k)$ ($\mathfrak k=\mathfrak g,\mathfrak h,\mathfrak l$),
      \begin{equation}\label{eq:equivariance1} \begin{split}
       ( g  \rhd \Psi) \wedge^{\{\cdot\}}( g \rhd \Psi)& = g \rhd \left(\Psi \wedge^{\{\cdot\}}   \Psi\right),
      \\
     Ad_g K \wedge^\rhd (g \rhd \Upsilon)& =   g \rhd (K \wedge^\rhd \Upsilon).
      \end{split}\end{equation}
    \end{prop}
  \begin{proof} (1) Write $   K  =\sum_a K^a X_a$, $   M  =\sum_b M^b X_b $ and $   \psi  =\sum_c \psi^c Z_c$. By $\delta$ being a $\mathfrak g$-homomorphism,
   \begin{equation*}   K \wedge^{\rhd}  \delta (\psi)=\sum_{a,b} K^a\wedge \psi^b  X_a \rhd  \delta ( Z_b)=\sum_{a,b} K^a\wedge \psi^b \delta ( X_a \rhd Z_b)=\delta (K \wedge^{\rhd}\psi).
   \end{equation*}
 Since    $X_a\rhd X_b =[X_a,X_b ]$ ($\rhd$ acting on $\mathfrak g$ by the adjoint action), we have
\begin{equation*}
     K \wedge^{\rhd}  M=\sum_{a,b} K^a\wedge M^b  X_a \rhd X_b  =\sum_{a,b} K^a\wedge M^b  (X_a X_b - X_b  X_a)=K\wedge M-(-1)^{ts} M \wedge K.
     \end{equation*}

(2) Write $   \Psi_j =\sum_b\Psi_j^b Y_b \in \Lambda^k (U, \mathfrak h)$. Then
\begin{equation*}\begin{split}
         d ( \Psi_1\wedge^{\{\cdot\}}  \Psi_2)&=d \sum_{a,b}\Psi_1^a\wedge\Psi_2^b  \{ Y_a, Y_b   \}
         =\sum_{a,b}d\Psi_1^a\wedge\Psi_2^b  \{ Y_a, Y_b   \}+(-1)^{k_1} \sum_{a,b}\Psi_1^a\wedge d\Psi_2^b  \{ Y_a, Y_b   \}\\&=   d \Psi_1 \wedge^{\{\cdot\}}  \Psi_2+(-1)^{k_1} \Psi_1\wedge^{\{\cdot\}}  d \Psi_2 ,
      \end{split} \end{equation*}and
\begin{equation*}\begin{split}
         K \wedge^{\rhd }\left( \Psi_1\wedge^{\{\cdot\}}  \Psi_2\right)&=\sum_{a,b,c}K^c\wedge \Psi_1^a\wedge\Psi_2^b  X_c \rhd\{ Y_a, Y_b   \}\\&=\sum_{a,b,c}K^c\wedge \Psi_1^a\wedge\Psi_2^b(\{X_c \rhd Y_a, Y_b   \}+\{ Y_a,X_c \rhd Y_b   \}),
                     \end{split} \end{equation*}
    by   $\mathfrak g$-equivariance of $\{,\}$ from (\ref{eq:G-Peiffer}). The proofs of the other identities are similar.

     (3) Since $\rhd$ is a left action of $\mathfrak g$ on $\mathfrak l$ from (\ref{eq:automorphisms}), we have
\begin{equation*}\begin{split}
K\wedge^{\rhd }(M \wedge^{\rhd }\psi)+ M\wedge^{\rhd }( K\wedge^{\rhd }\psi)&= \sum_{a,b,c}K^a\wedge M^b\wedge \psi^c
(X_a\rhd(X_b\rhd Z_c)-X_b\rhd(X_a\rhd Z_c))\\&
=\sum_{a,b,c}K^a\wedge M^b \wedge\psi^c
[X_a,X_b]\rhd Z_c.
      \end{split} \end{equation*}

(4) \begin{equation*}
       \delta (\psi) \wedge^{\{\cdot\}} \delta (\psi)=  \sum_{a,b} \psi^a\wedge \psi^b\{\delta(Z_a), \delta(Z_b)\}=\sum_{a,b} \psi^a\wedge \psi^b[Z_a ,  Z_b ]
     \end{equation*}
   by 3. in the definition of a differential $ 2$-crossed module. Here $\psi^a\wedge\psi^b = \psi^b\wedge \psi^a$ since they are $2$-forms. It must vanish. And
\begin{equation*}\begin{split}
        (\Psi\wedge\Psi ) \wedge^{ [,] } \Psi &=  \sum_{a,b,c} \Psi^a \wedge \Psi^b\wedge \Psi^c [  Y_a Y_b , Y_c] =  \sum_{a,b,c}  \Psi^a \wedge \Psi^b\wedge \Psi^c(   Y_a Y_bY_c  -Y_c Y_a  Y_b )=  0,
      \end{split} \end{equation*}
     by  $  \Psi^a \wedge \Psi^b\wedge \Psi^c = \Psi^c \wedge\Psi^a \wedge \Psi^b  $, since $\Psi^*$'s are $1$-forms.

 (5)
       \begin{equation*}  \begin{split}
       ( g \rhd \Psi) \wedge^{\{\cdot\}}( g \rhd \Psi)&=\sum_{a,b} \Psi^a\wedge\Psi^b \{g \rhd Y_a,g \rhd Y_b\} =\sum_{a,b}\Psi^a\wedge \Psi^b g \rhd\{   Y_a,   Y_b\}= g \rhd \left(\Psi \wedge^{\{\cdot\}}   \Psi\right)
      \end{split}  \end{equation*}
      by $G$-equivariance of $ {\{\cdot\}} $   in (\ref{eq:G-Peiffer}). And
\begin{equation*}  \begin{split}
     Ad_g K \wedge^\rhd (g \rhd  \Upsilon)&=\sum_{a,b}  K^a\wedge \Upsilon^b (  {g }X_ag^{-1})\rhd (g \rhd  Z_b)=\sum_{a,b} K^a\wedge \Upsilon^b  ( g  X_a)\rhd   Z_b \\&=\sum_{a,b} K^a\wedge \Upsilon^b   g  \rhd(X_a \rhd   Z_b )=   g \rhd (K \wedge^\rhd \Upsilon),
      \end{split}\end{equation*}
by $     T \rhd (g  \rhd S)=(  T g )\rhd   S $, which follows from
$
     (g_1g_2)\rhd S= g_1\rhd(g_2 \rhd S).
$
         \end{proof}
    \begin{cor}\label{cor:2cross-form} For $  \Psi \in \Lambda^1(U,\mathfrak h)$   and $ A  \in \Lambda^{1}(U,\mathfrak g)$, we have
\begin{equation*}\begin{split}
     (\alpha(\Psi)\wedge {A} + {A} \wedge\alpha(\Psi))\wedge^{\rhd }\Psi
     &=(  {A} \wedge^{\rhd }
     \Psi)\wedge\Psi-\Psi\wedge(  {A} \wedge^{\rhd }
     \Psi)-(  {A} \wedge^{\rhd }
     \Psi)\wedge^{\langle,\rangle}\Psi  .
\end{split} \end{equation*}
    \end{cor}
  \begin{proof}  Write $A =\sum_a   A^a X_a$ and $ {\Psi} =\sum_b   \Psi^b Y_b $. Then,
    \begin{equation*}\begin{split}
     (\alpha(\Psi)\wedge {A} & + {A} \wedge\alpha(\Psi))\wedge^{\rhd }\Psi
     =\sum_{a,b,c}  \Psi^b\wedge A^a \wedge \Psi^c ([\alpha(   Y_b),  X_a] \rhd Y_c)\\&=\sum_{a,b,c}  \Psi^b\wedge A^a \wedge \Psi^c ( \alpha(   Y_b) \rhd(  X_a  \rhd Y_c)-X_a  \rhd  ( \alpha(   Y_b) \rhd Y_c))\\&
     =\sum_{a,b,c}  \Psi^b\wedge A^a \wedge \Psi^c ( [  Y_b , X_a  {\rhd } Y_c)]-\langle Y_b , X_a  {\rhd } Y_c \rangle- X_a   {\rhd } [ Y_b ,Y_c]+X_a   {\rhd } \langle Y_b ,Y_c\rangle)
     \\&= \sum_{a,b,c}  \Psi^b\wedge A^a \wedge \Psi^c (-[ X_a {\rhd } Y_b , Y_c ]+\langle X_a  {\rhd }Y_b , Y_c \rangle) \\&
     =(  {A} \wedge^{\rhd }
     \Psi)\wedge\Psi-\Psi\wedge(  {A} \wedge^{\rhd }
     \Psi)-(  {A} \wedge^{\rhd }
     \Psi)\wedge^{\langle,\rangle}\Psi.
\end{split}\end{equation*}
Here  $X_a   {\rhd } [ Y_b ,Y_c]= [X_a   {\rhd } Y_b ,Y_c] +[ Y_b ,X_a   {\rhd }Y_c]$ by $\rhd$ acting on $\mathfrak h$   as automorphisms, and
\begin{equation*}\begin{split}
X_a  {\rhd }\langle Y_b , Y_c \rangle&=X_a  {\rhd }\delta\{ Y_b , Y_c \}=\delta (X_a  {\rhd }\{ Y_b , Y_c \})\\&=\delta (\{ X_a  {\rhd }Y_b , Y_c \} +\{ Y_b ,    X_a{\rhd } Y_c \})=\langle X_a  {\rhd }Y_b , Y_c \rangle + \langle Y_b ,  X_a{\rhd } Y_c \rangle .
\end{split}\end{equation*}
The corollary is proved.
\end{proof}

\section{Covariance of curvatures under the  $3$-gauge transformations}
\subsection{ Three kinds of  gauge transformations}

There are three kinds of $3$-gauge transformations. The  {\it  $3$-gauge transformation  of the first kind}:
  \begin{equation} \label{eq:trans1} \begin{split}
        & A'=Ad_{g^{-1}} A+g^{-1}dg, \\&
         B'= g^{-1} \rhd B , \\&
       C'= g^{-1} \rhd C  ,
      \end{split}\end{equation}
      the  {\it $3$-gauge transformation  of the  second kind}:
  \begin{equation} \label{eq:trans2} \begin{split}
        & A'=A  + \alpha(\varphi ), \\&
         B'=  B+d\varphi+A'\wedge^{\rhd}\varphi-\varphi\wedge\varphi , \\&
       C'=   C  - {B}'\wedge^{\{,\}}\varphi-\varphi\wedge^{\{,\}} {B},
      \end{split}\end{equation}
and the {\it $3$-gauge transformation  of the  third kind}:
    \begin{equation}\label{eq:trans3} \begin{split}&  A'=A , \\&
B'=  {   {B}}  -\delta(\psi),
     \\&  C'=  { {C }}- d\psi-A'\wedge^{\rhd }\psi.
    \end{split} \end{equation}

If we write a  $3$-gauge transformation  of the  second kind as
 \begin{equation} \label{eq:trans2'} \begin{split}
        & A''=A'  + \alpha(\varphi ), \\&
         B''= B'+d\varphi+A''\wedge^{\rhd}\varphi-\varphi\wedge\varphi , \\&
       C''=   C'  - {B}''\wedge^{\{,\}}\varphi-\varphi\wedge^{\{,\}} {B}',
      \end{split}\end{equation}
and a $3$-gauge transformation  of the  third kind as
 \begin{equation}\label{eq:tilde'}\begin{split}&
     \widetilde{A}=A'' ,\\&
\widetilde{B}=  {   {B}}''  -\delta(\psi),\\&
        \widetilde{C}=  { {C }}''- d\psi-\widetilde{A} \wedge^{\rhd }\psi ,
\end{split}\end{equation}
then the composition of (\ref{eq:trans1}), (\ref{eq:trans2'}) and (\ref{eq:tilde'}) gives
\begin{equation*} \begin{split}
    & \widetilde{A}=A'  + \alpha(\varphi )=Ad_{g^{-1}} A+g^{-1}dg + \alpha(\varphi ),\\&
   \widetilde{B}=   {   {B}}''  -\delta(\psi)=g^{-1} \rhd B +d\varphi+\widetilde{A} \wedge^{\rhd}\varphi-\varphi\wedge\varphi  -\delta(\psi),
    \end{split}\end{equation*}
and
     \begin{equation*}\begin{split}
         \widetilde{C} &=  C'  - {B}''\wedge^{\{,\}}\varphi-\varphi\wedge^{\{,\}} {B}' - d\psi-\widetilde{A} \wedge^{\rhd }\psi  \\&= g^{-1} \rhd   C- \left(\widetilde{{B}}+\delta(\psi)\right) \wedge^{\{,\}}\varphi-\varphi\wedge^{\{,\}} (g^{-1} \rhd  {B}) - d\psi-\widetilde{A} \wedge^{\rhd }\psi\\&
                  =  g^{-1} \rhd  C - d\psi-\widetilde{A} \wedge^{\rhd }\psi -  \widetilde{{B}}  \wedge^{\{,\}}\varphi-\varphi\wedge^{\{,\}} (g^{-1} \rhd  {B}) +\varphi\wedge^{\rhd'}  \psi
        \end{split} \end{equation*}
        by using Proposition \ref{prop:2cross-form} (3).  Namely, the composition of (\ref{eq:trans1}), (\ref{eq:trans2'}) and (\ref{eq:tilde'}) gives the general $3$-gauge transformation  (\ref{eq:gauge-transformations}).

\subsection{Covariance of   $1$-curvature $2$-forms and   $2$-curvature $3$-forms}
For the $1$-curvature, under the $3$-gauge transformation (\ref{eq:gauge-transformations}),
\begin{equation}\label{eq:Omega1'}\begin{split}
     dA'+A'\wedge A'=& dg^{-1}\wedge A g+g^{-1}dA g-g^{-1}A \wedge dg+dg^{-1} \wedge d g+ \alpha(d\varphi)\\&
     +(g^{-1} A g+g^{-1} dg+  \alpha( \varphi))\wedge (g^{-1} A g+g^{-1} dg+  \alpha( \varphi))\\=&
g^{-1}( dA +A \wedge A)g+ \alpha(d\varphi)+A'\wedge \alpha( \varphi)+\alpha( \varphi)\wedge  A'- \alpha( \varphi\wedge \varphi)\\=&
 g^{-1} \rhd\Omega_1+  \alpha(  {B}')-\alpha(g^{-1} \rhd B),
\end{split} \end{equation}
by $\alpha( \delta(\psi))=0$. Here $A'\wedge \alpha( \varphi)+\alpha( \varphi)\wedge  A'= A'\wedge^{\rhd}\alpha( \varphi)$ by the second identity of (\ref{eq:wedge-rhd}) in Proposition \ref{prop:2cross-form-2}.

For the $2$-curvature, under the $3$-gauge transformation (\ref{eq:gauge-transformations}), we find that
\begin{equation}\label{eq:Omega2'}\begin{split}
     dB'+A'\wedge^{\rhd} B'=& dg^{-1} \wedge^{\rhd} B+g^{-1}\rhd dB+dA'  \wedge^{\rhd}  \varphi-A'  \wedge^{\rhd} d \varphi- d \varphi\wedge \varphi+\varphi\wedge d \varphi\\&-\delta(d\psi) +A'\wedge^{\rhd}  \left(g^{-1}\rhd B+d\varphi+A'\wedge^{\rhd}\varphi-\varphi\wedge\varphi-\delta(\psi)\right)
      \\=&g^{-1}\rhd (dB +A \wedge^{\rhd} B )+(dA' +A'\wedge A') \wedge^{\rhd}\varphi+\alpha( \varphi)\wedge^{\rhd} ( g^{-1}\rhd  B) \\&- d \varphi\wedge \varphi+\varphi\wedge d \varphi  -A'\wedge^{\rhd}   (\varphi\wedge\varphi) -\delta(d\psi +A'\wedge^{\rhd}   \psi )\\
      =&g^{-1}\rhd \Omega_2+ \Omega_1' \wedge^{\rhd}\varphi- d\varphi\wedge^{[,] }\varphi -(A'\wedge^{\rhd}   \varphi)\wedge\varphi+\varphi\wedge ( A'\wedge^{\rhd}\varphi)\\&+\alpha( \varphi)\wedge^{\rhd}(g^{-1}\rhd B)-\delta(d\psi +A'\wedge^{\rhd}   \psi)
 \end{split} \end{equation}
 by using Proposition \ref{prop:2cross-form-2} (1)-(3) and (5) in the second identity and using    Proposition \ref{prop:2cross-form-2} (2) in the third identity. Here $d\varphi\wedge^{[,] }\varphi=d \varphi\wedge \varphi-\varphi\wedge d \varphi $ by definition.
  $\Omega_1'$ in (\ref{eq:Omega1'}) can be written as
     \begin{equation}\label{eq:Omega-1'}\begin{split}
    \Omega_1'
     = g^{-1} \rhd\Omega_1 +\alpha(d\varphi)+\alpha(\varphi)\wedge {A}'+ {A}'\wedge\alpha(\varphi)-\alpha(\varphi\wedge\varphi).
\end{split}\end{equation}
Note that
\begin{equation}\label{eq:rhd-1} \begin{split}
 \alpha(\varphi\wedge\varphi)\wedge^{\rhd }\varphi&=(\varphi\wedge\varphi)\wedge^{[,]}\varphi- (\varphi\wedge\varphi)\wedge^{\langle,\rangle}\varphi   =- (\varphi\wedge\varphi)\wedge^{\langle,\rangle}\varphi  ,\end{split}\end{equation} by using Proposition \ref{prop:2cross-form} (1) and Proposition \ref{prop:2cross-form-2} (4), and
\begin{equation}\label{eq:rhd-2}\begin{split}
     (\alpha(\varphi)\wedge {A}'+ {A}'\wedge\alpha(\varphi))\wedge^{\rhd }\varphi
     =(  {A}'\wedge^{\rhd }
     \varphi)\wedge\varphi-\varphi\wedge(  {A}'\wedge^{\rhd }
     \varphi)-(  {A}'\wedge^{\rhd }
     \varphi)\wedge^{\langle,\rangle}\varphi
\end{split}\end{equation}
by using Corollary \ref{cor:2cross-form}.
Substitute (\ref{eq:Omega-1'})-(\ref{eq:rhd-2}) to $\Omega_1'\wedge^{\rhd }\varphi$ to get
\begin{equation}\label{eq:Omega-varphi} \begin{split}
     \Omega_1'\wedge^{\rhd }\varphi&=(g^{-1}\rhd\Omega_1)\wedge^{\rhd }\varphi+d\varphi\wedge^{[,] }\varphi -d\varphi\wedge^{\langle,\rangle}\varphi\\&+( {A}'\wedge^{\rhd }
     \varphi)\wedge\varphi-\varphi\wedge(  {A}'\wedge^{\rhd }
     \varphi)
    -( {A}'\wedge^{\rhd }
     \varphi)\wedge^{\langle,\rangle}\varphi+(\varphi\wedge\varphi)\wedge^{\langle,\rangle}\varphi .
\end{split}\end{equation}
 Here we have   applied  Proposition \ref{prop:2cross-form} (1) to $\alpha(d\varphi)\wedge^{\rhd }\varphi$.
Now substitute (\ref{eq:Omega-varphi}) into (\ref{eq:Omega2'}) to get
\begin{equation}\label{eq:Omega-2-half}\begin{split}
     \Omega_2'=&   g^{-1}\rhd \Omega_2+ (g^{-1}\rhd\Omega_1)  \wedge^{\rhd}\varphi -d\varphi\wedge^{\langle,\rangle}\varphi -( {A}'\wedge^{\rhd }
     \varphi)\wedge^{\langle,\rangle}\varphi +(\varphi\wedge\varphi)\wedge^{\langle,\rangle}\varphi\\&  +\alpha( \varphi)\wedge^{\rhd}(g^{-1}\rhd B) -\delta(d\psi +A'\wedge^{\rhd}   \psi).
 \end{split} \end{equation}
 Note that   \begin{equation}\label{eq:varphi-B}\begin{split}
        \alpha ( \varphi)\wedge^{\rhd} (g^{-1}\rhd  B) = & \varphi\wedge^{[,]}  (g^{-1}\rhd  B)
       -\varphi \wedge^{\langle,\rangle}   (g^{-1}\rhd  B)
       ,
  \end{split}\end{equation}
 by  Proposition \ref{prop:2cross-form} (1) again, and  similarly
   \begin{equation}\label{eq:B-varphi}\begin{split}
  \alpha(g^{-1}\rhd   B)\wedge^{\rhd}  \varphi&= (g^{-1}\rhd   B) \wedge^{[,]} \varphi - (g^{-1}\rhd B ) \wedge^{\langle,\rangle} \varphi\\&
  =- \varphi\wedge^{[,]}(g^{-1}\rhd B) - (g^{-1}\rhd  B) \wedge^{\langle,\rangle} \varphi.
  \end{split}\end{equation}
 Now substitute the summation of (\ref{eq:varphi-B}) and (\ref{eq:B-varphi})  into (\ref{eq:Omega-2-half}) to get
 \begin{equation*}\begin{split}
     \Omega_2'=&   g^{-1}\rhd \Omega_2+ ( g^{-1}\rhd\Omega_1  -\alpha(g^{-1}\rhd  B)) \wedge^{\rhd}\varphi - ( g^{-1}\rhd  B +d\varphi+  {A}'\wedge^{\rhd }
     \varphi - \varphi\wedge\varphi)\wedge^{\langle,\rangle}\varphi \\&-\varphi \wedge^{\langle,\rangle} (g^{-1}\rhd  B) -\delta(d\psi +A'\wedge^{\rhd}   \psi)\\= &
     g^{-1}\rhd \Omega_2+ ( \Omega_1' - \alpha(  B')) \wedge^{\rhd}\varphi -B'\wedge^{\langle,\rangle}\varphi- \delta( \psi)\wedge^{\langle,\rangle}\varphi \\& -\varphi \wedge^{\langle,\rangle}   (g^{-1}\rhd  B)-\delta(d\psi +A'\wedge^{\rhd}   \psi)\\ =&
g^{-1}\rhd \Omega_2+ (\Omega_1' -\alpha(  B')) \wedge^{\rhd}\varphi +\delta( C')-   \delta(g^{-1}\rhd C )
 \end{split} \end{equation*}
by $\delta( \psi)\wedge^{\langle,\rangle}\varphi=-\delta(\varphi\wedge^{\rhd'} \psi)  $ from applying $\delta$ to Proposition \ref{prop:2cross-form} (3) and
\begin{equation*}
        \delta( C')-g^{-1}\rhd\delta( C )=-\delta( d\psi +A'\wedge^{\rhd } \psi) +\delta(\varphi\wedge^{\rhd'}\psi)-B'\wedge^{\langle,\rangle}\varphi -\varphi \wedge^{\langle,\rangle}   (g^{-1}\rhd  B).
     \end{equation*}

\subsection{Covariance of   $3$-curvature $4$-forms }
(1)
Under  the $3$-gauge transformation (\ref{eq:trans1}) of the  first kind,
      \begin{equation}\label{eq:curvature-under1} \begin{split}
      \Omega_3'= &g^{-1}\rhd dC-  g^{-1} dg  g^{-1}\wedge^\rhd C+
     (Ad_{g^{-1}} A+g^{-1}dg)\wedge^\rhd (g^{-1}\rhd C)\\& + (g^{-1}\rhd B )\wedge^{\{\cdot\}}( g^{-1}\rhd B)=
  g^{-1}\rhd( dC+A \wedge^{\rhd} C +B \wedge^{\{\cdot\}} B )=g^{-1}\rhd \Omega_3,
 \end{split}\end{equation}
       by  equivariance (\ref{eq:equivariance1}) in Proposition \ref{prop:2cross-form-2}. We have already seen from the above subsection that under this transformation
\begin{equation}\label{eq:curvature-under1'}
      \Omega_2' =g^{-1}\rhd \Omega_2  ,\qquad  \Omega_1' = g^{-1}\rhd\Omega_1.
\end{equation}

\vskip 5mm
(2) Under the $3$-gauge transformation  (\ref{eq:trans2}) of the   second kind, we have
         \begin{equation}\label{eq:Omega-3-varphi} \begin{split}
     \Omega_3'=&dC'+A'\wedge^{\rhd} C'+B'\wedge^{\{\cdot\}} B'\\=&dC +A'\wedge^{\rhd} C
  -(d {B}'+A'\wedge^{\rhd}  {B}')  \wedge^{\{,\}}\varphi +\varphi\wedge^{\{,\}}(d {B}+A'\wedge^{\rhd}   {B})\\&- {B}'\wedge^{\{,\}}(d\varphi+A'\wedge^{\rhd}\varphi)-(d\varphi+A'\wedge^{\rhd}\varphi)\wedge^{\{,\}} {B}+B'\wedge^{\{\cdot\}} B'\\=&
  \Omega_3+\alpha(\varphi )\wedge^{\rhd} C -\Omega_2'\wedge^{\{,\}}\varphi+\varphi\wedge^{\{,\}}\Omega_2 +\varphi\wedge^{\{,\}}( \alpha (\varphi ) \wedge^{\rhd}   {B})\\ &-{B}'\wedge^{\{,\}}(\varphi\wedge\varphi)  -(\varphi\wedge\varphi)\wedge^{\{,\}} {B},
 \end{split}\end{equation}
by applying Proposition \ref{prop:2cross-form-2} (2) and using the following in the last identity, i.e.,
\begin{equation*}
      - {B}'\wedge^{\{,\}}(d\varphi+A'\wedge^{\rhd}\varphi)+B'\wedge^{\{\cdot\}} B'=B'\wedge^{\{\cdot\}} B -{B}'\wedge^{\{,\}}(\varphi\wedge\varphi), \end{equation*} and
      \begin{equation*}
        -(d\varphi+A'\wedge^{\rhd}\varphi)\wedge^{\{,\}} {B}+B'\wedge^{\{\cdot\}} B =B \wedge^{\{\cdot\}} B -(\varphi\wedge\varphi)\wedge^{\{,\}} {B}  .  \end{equation*}
But
\begin{equation}\label{eq:Bff}
  -   {B}'\wedge^{\{,\}}(\varphi\wedge\varphi)  = -({B}'\wedge^{ \langle,\rangle}   \varphi)\wedge^{\{,\}}\varphi
\end{equation}by applying Proposition \ref{prop:2cross-form} (4)
and
\begin{equation}\label{eq:fC}\begin{split}
     \alpha(\varphi )\wedge^{\rhd} C=&-\varphi  \wedge^{\{,\}}\delta( C)+\delta( C)\wedge^{\{,\}} \varphi
\\=&-\varphi  \wedge^{\{,\}}\delta( C)+\delta( C')\wedge^{\{,\}} \varphi +[ {B}'\wedge^{ \langle,\rangle}  \varphi+\varphi\wedge^{ \langle,\rangle}   {B}]\wedge^{\{,\}}\varphi
 \end{split}\end{equation}by applying  Proposition \ref{prop:2cross-form} (2).
Now substitute the summation of (\ref{eq:Bff}) and (\ref{eq:fC}) into (\ref{eq:Omega-3-varphi})
to get
\begin{equation}\label{eq:curvature-under2}
     \Omega_3'=
  \Omega_3   -[\Omega_2'-\delta( C')]\wedge^{\{,\}}\varphi+\varphi\wedge^{\{,\}}[\Omega_2-\delta( C )] +\mathcal{E}
\end{equation}
with  \begin{equation}\label{eq:E}
    \mathcal{E}= (\varphi \wedge^{ \langle,\rangle}   B) \wedge^{\{,\}} \varphi +\varphi\wedge^{\{,\}}( \alpha (\varphi ) \wedge^{\rhd} {B})-(\varphi\wedge\varphi) \wedge^{\{,\}} {B}=0.
 \end{equation}
This vanishing is because we have that \begin{equation}\label{eq:affB}
     \varphi\wedge^{\{,\}} [ \alpha(\varphi)\wedge^{\rhd} {B}]= \varphi\wedge^{\{,\}} ( \varphi \wedge^{[,]} {B})
     -\varphi\wedge^{\{,\}} ( \varphi \wedge^{\langle,\rangle} {B})
\end{equation}
by Proposition \ref{prop:2cross-form} (1), and
\begin{equation}\label{eq:ffB}\begin{split}
   (\varphi\wedge\varphi) \wedge^{\{,\}} {B}  =& \varphi\wedge^{\{,\}}(\varphi\wedge^{[,]}  {B})+
\alpha(\varphi)\wedge^{\rhd}( \varphi  \wedge^{\{,\}}  {B})
   \end{split}\end{equation}
by Proposition \ref{prop:2cross-form} (5), and by substituting (\ref{eq:affB})-(\ref{eq:ffB}) into (\ref{eq:E})  we have
\begin{equation}\label{eq:fBf}\begin{split}
     \mathcal{E}= & \delta( \varphi \wedge^{\{,\}}   B) \wedge^{\{,\}} \varphi -\varphi\wedge^{\{,\}} \delta( \varphi \wedge^{\{,\}} {B})-\alpha(\varphi)\wedge^{\rhd}( \varphi  \wedge^{\{,\}}  {B})=0
   \end{split}\end{equation}
by applying Proposition \ref{prop:2cross-form} (2) to $\psi=\varphi \wedge^{\{,\}}   B$.

\vskip 5mm
(3) Under the $3$-gauge transformation (\ref{eq:tilde'}) of the  third kind,
     \begin{equation*} \begin{split}
   \widetilde  \Omega_3 =&d\widetilde{C} + \widetilde A \wedge^{\rhd}\widetilde  C +\widetilde B \wedge^{\{\cdot\}}\widetilde  B  = d  { C}'' +A'' \wedge^{\rhd}   {C''} +  {B''} \wedge^{\{\cdot\}}   {B''} +\mathcal{R}
 \end{split}\end{equation*}
 with
 \begin{equation*}\begin{split}
     \mathcal{R}=&-dA'' \wedge^{\rhd }\psi  -A'' \wedge^{\rhd}(A ''\wedge^{\rhd }\psi)- {B''} \wedge^{\{\cdot\}} \delta(\psi) -\delta(\psi) \wedge^{\{\cdot\}} {B}'' +\delta(\psi) \wedge^{\{\cdot\}}\delta(\psi)\\=& -( dA''    +A'' \wedge A''-\alpha(B'') )\wedge^{\rhd}\psi=-(\Omega_1''-\alpha(B''))\wedge^{\rhd}\psi.
  \end{split}\end{equation*}
   since
    $\delta(\psi) \wedge^{\{\cdot\}} \delta(\psi)  =0$ by Proposition \ref{prop:2cross-form-2} (4),  and
          \begin{equation*}\begin{split}
    -  \delta(\psi) \wedge^{\{\cdot\}} B''-B'' \wedge^{\{\cdot\}}\delta(\psi)= \alpha(B'')\wedge^{\rhd}\psi ,
    \end{split}  \end{equation*}
    by Proposition \ref{prop:2cross-form} (2) and $A'' \wedge^{\rhd}(A'' \wedge^{\rhd }\psi)=(A'' \wedge A'' )\wedge^{\rhd }\psi $ by Proposition \ref{prop:2cross-form-2} (3).
In summary we     know that under this transformation,
\begin{equation}\label{eq:curvature-under3}\begin{split} &\widetilde \Omega_1 - \alpha( \widetilde B  )= \Omega_1'' - \alpha( B'' )
,\\& \widetilde \Omega_2 - \delta(\widetilde  C )= \Omega_2'' - \delta( C'' ),\\&\widetilde \Omega_3 = \Omega_3''-(\widetilde{\Omega}_1 -\alpha(\widetilde{B} ))\wedge^{\rhd}\psi.
\end{split}\end{equation}
The last two identities are already proved in  subsection 3.2.

Since any $3$-gauge transformation (\ref{eq:gauge-transformations})
 can be written as a composition of
 (\ref{eq:trans1}), (\ref{eq:trans2'}) and (\ref{eq:tilde'}), identities (\ref{eq:Omega1'}),
(\ref{eq:curvature-under1}), (\ref{eq:curvature-under1'}), (\ref{eq:curvature-under2}) and (\ref{eq:curvature-under3})
imply the transform formula (\ref{eq:gauge-transformations-curvature}) of the $3$-curvature for general $3$-gauge transformations.

\section{The $\mathbf{Gray}$-categories and the lax-natural transformations}

 \subsection{$\mathbf{Gray}$-categories}\label{sub:Gray-3-category}

$\mathbf{Gray}$ is a closed symmetric monoidal category with the Gray tensor product \cite{Gr}. The underlying category  is the category of  $2$-categories and   $2$-functors between them.  A $\mathbf{Gray}$-category is a category enriched over $\mathbf{Gray}$.
Any tricategory is triequivalent to
a  $\mathbf{Gray}$-category \cite{GPS} \cite{G06}. $\mathbf{Gray}$-categories  are the  semi-strictification of   tricategories. It is also well-known that the homotopy category of $\mathbf{Gray}$ $3$-categories is equivalent
to the homotopy category of $3$-types (cf., e.g.,   \cite{Ber}).

The unpacked version of this definition is as follows.
A $\mathbf{Gray}$-category  $  \mathcal{ C }$  (cf. \cite{C99} \cite{RS08} and references therein) consists of collections $ \mathcal C_0$ of objects, $ \mathcal C_1$ of $1$-arrows, $ \mathcal C_2$ of $2$-arrows and $ \mathcal C_3$ of $3$-arrows,
together with

$\bullet$ functions $s_n, t_n :  \mathcal C_i\rightarrow  \mathcal C_n$ for all $0 \leq n < i \leq 3$,  called
{\it $n$-source} and {\it $n$-target},

$\bullet$ functions $\#_n :  \mathcal C_{n+1}\, {}_{t_n}\times_{s_n}  \mathcal C_{n+1}\rightarrow  \mathcal C_{n+1}$ for all $0 \leq n < 3$, called {\it vertical composition},

$\bullet$ functions $\#_n :  \mathcal C_i\,  {}_{t_n}\times_{s_n}  \mathcal C_{n+1}\rightarrow  \mathcal C_i$ and $\#_n :  \mathcal C_{n+1}\, {}_{t_n}\times_{s_n} \mathcal C_i \rightarrow  \mathcal C_i$ for all $  n=0,1$,
$n + 1 <i \leq3$, called {\it whiskering},

$\bullet$ a function  $\#_0 :  \mathcal C_2 \, {}_{t_0}\times_{s_0}  \mathcal C_{2}\rightarrow  \mathcal C_3$, called the {\it interchanging $3$-arrow (horizontal composition)},

$\bullet$ a function  $ id_{*}   :  \mathcal C_i \rightarrow  \mathcal C_{i+1 }$ for all $ 0 \leq i \leq 2$, called {\it identity},

such that

(1) $ \mathcal{C }$  is a $ 3$-skeletal reflexive globular (cf. \cite{Bra})  set;

(2) for every $C,C'\in  \mathcal C_0$, the collection of elements of $ \mathcal{C }$ with $0$-source $C$ and $0$-target
$C'$ forms a $2$-category $\mathcal{C}(C,C')$, with $n$-composition in $\mathcal{C}(C,C')$ given by $\#_{n+1}$ and
identities given by $id_{*}$;

(3) for each  $g : C' \rightarrow C''$ in $ \mathcal C_1$ and every $C, C''' \in  \mathcal C_0$, $ g\#_0$ is a $2$-functor
$\mathcal{C}(C'',C''')\rightarrow \mathcal{C}(C',C''')$,  and  $\#_0 g$ is   a $2$-functor
$\mathcal{C}(C ,C')\rightarrow \mathcal{C}(C,C'')$;

(4) for     every $C,C', C''  \in  \mathcal C_0$, $id_{C'} \#_0$ is equal to the identity $2$-functor
$\mathcal{C}(C' ,C'' )\rightarrow \mathcal{C}(C',C'' )$,  and  $\#_0  id_{C'} $ is equal to the identity $2$-functor
$\mathcal{C}(C ,C')\rightarrow \mathcal{C}(C,C' )$;

(5) for $ \gamma: \xy
(-8,0)*+{C}="4";
(8,0)*+{C'}="6";
{\ar@/^ .85pc/^{f} "4";"6"};
{\ar@/_ .85pc/_{f'} "4";"6"};
{\ar@{=>}^<<<{ } (0,2)*{};(0,-2)*{}} ;
\endxy$ and $ \delta: \xy
(-8,0)*+{C'}="4";
(8,0)*+{C''}="6";
{\ar@/^ .85pc/^{g} "4";"6"};
{\ar@/_ .85pc/_{g'} "4";"6"};
{\ar@{=>}^<<<{ } (0,2)*{};(0,-2)*{}} ;
\endxy$ in $ \mathcal C_2$, we have the {\it interchanging $3$-arrow} $\gamma\#_0 \delta$ (\ref{eq:interchanging-def}) with
\begin{equation*}
            s_2(\gamma\#_0 \delta)=(\gamma \#_0 g) \#_1( f'\#_0 \delta) \qquad {\rm and} \qquad
        t_2(\gamma\#_0 \delta)=(f \#_0\delta) \#_1 (\gamma\#_0  g');
\end{equation*}

(6) for $\xy
(-8,0)*+{C}="4";
(8,0)*+{C'}="6";
{\ar@{->}|-{f'} "4";"6"};
{\ar@/^1.55pc/^{f'} "4";"6"};
{\ar@/_1.55pc/_{f''} "4";"6"};
{\ar@{=>}^<<{\scriptstyle \gamma} (0,6)*{};(0,1)*{}} ;
{\ar@{=>}^<<{\scriptstyle \gamma'} (0,-1)*{};(0,-6)*{}} ;
\endxy  $ and $ \delta: \xy
(-8,0)*+{C'}="4";
(8,0)*+{C''}="6";
{\ar@/^ .85pc/^{g} "4";"6"};
{\ar@/_ .85pc/_{g'} "4";"6"};
{\ar@{=>}^<<<{ } (0,2)*{};(0,-2)*{}} ;
\endxy$ in $ \mathcal C_2$,
\begin{equation*}
        (\gamma\#_1\gamma')\#_0\delta=[(\gamma \#_0 g) \#_1 (\gamma' \#_0 \delta)] \#_2  [(\gamma \#_0\delta) \#_1(\gamma'\#_0 g')],
\end{equation*}
and for $ \gamma: \xy
(-8,0)*+{C }="4";
(8,0)*+{C' }="6";
{\ar@/^ .85pc/^{f} "4";"6"};
{\ar@/_ .85pc/_{f'} "4";"6"};
{\ar@{=>}^<<<{ } (0,2)*{};(0,-2)*{}} ;
\endxy$  and  $\xy
(-8,0)*+{C'}="4";
(8,0)*+{C''}="6";
{\ar@{->}|-{g'} "4";"6"};
{\ar@/^1.55pc/^{g } "4";"6"};
{\ar@/_1.55pc/_{g''} "4";"6"};
{\ar@{=>}^<<{\scriptstyle \delta} (0,6)*{};(0,1)*{}} ;
{\ar@{=>}^<<{\scriptstyle\delta'} (0,-1)*{};(0,-6)*{}} ;
\endxy  $ in $ \mathcal C_2$,
\begin{equation*}
       \gamma\#_0(\delta\#_1\delta')=[(\gamma\#_0 \delta)\#_1 (  f ' \#_1 \delta')] \#_2    [(  f \#_1 \delta) \#_1  (\gamma\#_0 \delta')];
\end{equation*}

(7) for $\varphi: \xy 0;/r.22pc/:
(0,10)*{};
(0,-10)*{};
(-20,0)*+{C}="1";
(0,0)*+{C'}="2";
{\ar@/^1.33pc/^f "1";"2"};
{\ar@/_1.33pc/_{f'} "1";"2"};
(-10,4)*+{}="A";
(-10,-4)*+{}="B";
{\ar@{=>}@/_.5pc/_\gamma "A"+(-1.33,0) ; "B"+(-.66,-.55)};
{\ar@{=}@/_.5pc/ "A"+(-1.33,0) ; "B"+(-1.33,0)};
{\ar@{=>}@/^.5pc/^{\gamma'} "A"+(1.33,0) ; "B"+(.66,-.55)};
{\ar@{=}@/^.5pc/ "A"+(1.33,0) ; "B"+(1.33,0)};
{\ar@3{->} (-12,0)*{}; (-7,0)*{}};
\endxy $  in $ \mathcal C_3$ and $ \delta: \xy
(-8,0)*+{C'}="4";
(8,0)*+{C''}="6";
{\ar@/^ .85pc/^{g} "4";"6"};
{\ar@/_ .85pc/_{g'} "4";"6"};
{\ar@{=>}^<<<{ } (0,2)*{};(0,-2)*{}} ;
\endxy$ in $ \mathcal C_2$,
\begin{equation*}( \gamma \#_0  \delta)\#_2[(f\#_0\delta ) \#_1 ( \varphi\#_0 g  ')]
= [(\varphi \#_0 g)\#_1 ( f'\#_0\delta ) ] \#_2 (\gamma'\#_0\delta ).
\end{equation*}

(8) for $f: C\rightarrow C'$  in $ \mathcal C_1$ and $ \delta: \xy
(-8,0)*+{C'}="4";
(8,0)*+{C''}="6";
{\ar@/^ .85pc/^{g} "4";"6"};
{\ar@/_ .85pc/_{g'} "4";"6"};
{\ar@{=>}^<<<{ } (0,2)*{};(0,-2)*{}} ;
\endxy$ in $ \mathcal C_2$,
$
     {\rm id}_f\#_0 \delta={\rm id}_{f\#_0 \delta},
$
and for $ \gamma: \xy
(-8,0)*+{C }="4";
(8,0)*+{C' }="6";
{\ar@/^ .85pc/^{f} "4";"6"};
{\ar@/_ .85pc/_{f'} "4";"6"};
{\ar@{=>}^<<<{ } (0,2)*{};(0,-2)*{}} ;
\endxy$ in $ \mathcal C_2$ and $f: C'\rightarrow C''$  in $ \mathcal C_1$,
$
 \gamma  \#_0  {\rm id}_f={\rm id}_{\gamma\#_0 f}.
$

(9) For every $c\in\mathcal{C}(C,C')_p$, $c'\in\mathcal{C}(C',C'')_q $ and $c''\in\mathcal{C}(C'',C''')_r$ with $p+q+r\leq 2$,
\begin{equation*}
   c   \#_0 (c'\#_0c'')= (c   \#_0  c')\#_0c''.
\end{equation*}

 Here (5)-(7) are the definition of the interchanging $3$-arrow and its functoriality.
 Following \cite{MP11}, the definition of a $\mathbf{ Gray}$-category is a little bit different from the standard one. We write $\alpha\#_n\beta$ instead of $\beta\#_n\alpha$ when $t_n(\alpha)=s_n(\beta)$.
A (strict) {\it  $\mathbf{Gray}$-functor} $ F : \mathcal{C}\rightarrow\mathcal{C}'$ between $\mathbf{ Gray}$-categories $\mathcal{C}$ and $\mathcal{C}'$  is given by maps $F_i:  \mathcal{C}_i\rightarrow\mathcal{C}_i'$, $i = 0, \ldots , 3$,
preserving all compositions, identities, interchanges, sources and targets, strictly.
A {\it  $\mathbf{Gray}$ $3$-groupoid}
is a $\mathbf{ Gray}$-category whose $k$-arrows are all  equivalences, for all $k=1,2,3$.

 \subsection{The Gray $3$-groupoid $\mathcal{G}^{\mathscr L }$ constructed from the  $ 2$-crossed module $\mathscr L$}
Given a $2$-cross module $\mathscr L$, we can construct a $\mathbf{Gray}$ $3$-groupoid $ \mathcal{ G}^{\mathscr L}$ (cf. \S 1.2.5 in \cite{MP11}) with a single
object by putting $\mathcal{ G}_0^{\mathscr L} = \{ \bullet\}$, $\mathcal{ G}_1^{\mathscr L} = G$, $\mathcal{ G}_2^{\mathscr L} = G \times H$ and $\mathcal{ G}_3^{\mathscr L} = G \times H \times L$. This construction appeared in \cite{K} , with different
conventions, and also in a slightly different language in \cite{BG}   \cite{CCG}.

For a $2$-arrow   $(X, e) \in\mathcal{ G}_2^{\mathscr L} $,
\begin{equation*}
    s_1(X, e) = X ,\qquad t_1 (X, e) = \alpha(e)^{-1} X,
\end{equation*}
and for a $3$-arrow $(X, e,l) \in\mathcal{ G}_3^{\mathscr L} $,
\begin{equation*}\begin{split}
      &s_1 (X, e,l) =   X  ,\qquad\qquad t_1 (X, e,l) =\alpha(e)^{-1}  X ,\\
      & s_2(X, e,l) =  (X, e) ,\qquad t_2 (X, e,l) = (X, \delta(l)^{-1}e) .
\end{split}\end{equation*}

The vertical composition of two $2$-arrows is defined as
\begin{equation*}
  (X, e)  \#_1 (\alpha(e)^{-1} X, f)
= (X, ef ),
\end{equation*}
and the vertical composition of two $3$-arrows is defined as
\begin{equation}\label{eq:2-vertical}
     (X, e, l)\#_2( X, \delta(l)^{-1}e ,k)=(X, e  ,  lk).
\end{equation}
The  $\#_1$-composition of two $3$-arrows  is
\begin{equation}\label{eq:1-vertical}
     (X, e, l)\#_1(\alpha(e)^{-1} X, f ,k)=(X, ef , (e \rhd' k)l),
\end{equation}whose $2$-source is  $(X, ef)$  and the $2$-target is $(X,\delta((e \rhd' k)l)^{-1}ef)=(X,\delta( l)^{-1}e \cdot\delta(  k)^{-1} f)$.

 The interchanging $3$-arrow is defined as\begin{equation}\label{eq:interchanging}
     (X, e)\#_0(Y , f ) =
\left(
XY,  e (\alpha(e)^{-1}X)\rhd f , e\rhd' \{e^{-1},X \rhd f\}^{-1}  \right),
\end{equation}whose $2$-source  and $2$-target are
\begin{equation}\label{eq:interchanging-s-t}\begin{split}&  s_2( (X, e)\#_0(Y , f ))=
 (XY, e (\alpha(e)^{-1}X)\rhd f  )
, \\&
    t_2( (X, e)\#_0(Y , f ))= (XY, (X \rhd f ) e )
,
  \end{split}\end{equation}respectively.
This is because
$
   \delta  ( e\rhd'  \{e^{-1},X \rhd f \}^{-1}  )^{-1} \cdot e (\alpha(e)^{-1}X)\rhd f
   = e  e^{-1}(X \rhd f )  e (\alpha(e)^{-1}X)\rhd f^{-1}  e^{-1}
         e (\alpha(e)^{-1}X)\rhd f  =(X \rhd f ) e .
$

 Whiskering  by a $1$-arrow is defined as
\begin{equation}\label{eq:1-vertical0}\begin{split}
      X   \#_0 (Y, e)
&= (XY,  X \rhd e  ),\qquad\qquad \qquad  (Y, e)\#_0 X
= (Y X,  e  ),
\\  X   \#_0 (Y, e,l)
&= (XY,  X \rhd e  ,  X \rhd l),\qquad (Y, e,l)\#_0 X
= (Y X,  e ,l ).
\end{split}\end{equation}
The $3$-arrow whiskered  by a $ 2$-arrow {\it from above} is defined as
 \begin{equation}\label{eq:whiskering}
     (X, e )\#_1(\alpha(e)^{-1} X, f ,k)=(X, ef ,  e \rhd' k  )
\end{equation}
(i.e.,  the   composition (\ref{eq:1-vertical}) of a $3$-arrow  with a trivial $3$-arrow $(  X,  e ,1_L )$) and the one  {\it from below} is
\begin{equation}\label{eq:whiskering-r}
     (X, e, l)\#_1(\alpha(e)^{-1} X, f  )=(X, ef ,  l).
\end{equation}

\begin{rem}\label{rem:whiskering}
  By (\ref{eq:1-vertical0}) and (\ref{eq:whiskering-r}),  in the Gray $3$-groupoid $ \mathcal{ G}^{\mathscr L}$,    whiskering from right by a $ 1$-arrow or by a $ 2$-arrow from below  is trivial in the sense that the principal element ($e$ or $l$) is unchanged.
\end{rem}

 \subsection{The lax-natural transformations}
(cf. \S 5.1 in \cite{C99})
Let $F, \widetilde{F} : \mathcal{C}  \rightarrow\mathcal{ D}$ be   $\mathbf{Gray}$-functors between Gray  categories $\mathcal{C}$ and $\mathcal{D}$. A    {\it lax-natural transformation}
$\Psi: F \rightarrow {\widetilde{F}} $ consists of the following data:

$\bullet$ for every object $C$ of $\mathcal{C}$ a $1$-arrow $\Psi_C :{\widetilde{F}}(C)\rightarrow F(C) $ in $\mathcal{D}$,

$\bullet$ for every arrow $f : C \rightarrow C'$ in $\mathcal{C}$ a 2-arrow $\Psi_f$ in $ \mathcal{ D}$:
\begin{equation*}
       \xy
(0,0)*+{\scriptstyle F(C)}="1";
(20,0)*+{\scriptstyle F(C')}="2";
(0,20)*+{\scriptstyle {\widetilde{F}}(C)}="3";
(20,20)*+{\scriptstyle {\widetilde{F}}(C')}="4";
{\ar@{->}|-{\scriptscriptstyle F(f) } "1";"2"};
{\ar@{->}|-{ \scriptscriptstyle {\widetilde{F}} (f)} "3";"4"};
{\ar@{->}|-{\Psi_C  } "3";"1"};
{\ar@{->}|-{\Psi_{C'}  } "4";"2"};
{\ar@{=>}|-{ \scriptstyle  \Psi_f } (18,18)*{};(2, 2)*{}} ;
\endxy\end{equation*}

$\bullet$ for every $2$-arrow $ \gamma: \xy
(-8,0)*+{C}="4";
(8,0)*+{C'}="6";
{\ar@/^ .85pc/^{f} "4";"6"};
{\ar@/_ .85pc/_{f'} "4";"6"};
{\ar@{=>}^<<<{ } (0,2)*{};(0,-2)*{}} ;
\endxy$    in $\mathcal{C}$, we have  a $3$-arrow $\Psi_\gamma$  in $ \mathcal{ D}$:
\begin{equation*}
       \xy
(0,0)*+{\scriptstyle F(C)}="1";
(20,0)*+{\scriptstyle F(C')}="2";
(0,20)*+{\scriptstyle {\widetilde{F}}(C)}="3";
(20,20)*+{\scriptstyle {\widetilde{F}}(C')}="4";
{\ar@/^ 1.15pc/_{\scriptscriptstyle F(f) } "1";"2"};
{\ar@/^ 1.15pc/^{ \scriptscriptstyle {\widetilde{F}}(f)} "3";"4"};
{\ar@{->}|-{\Psi_C  } "3";"1"};
{\ar@{->}|-{\Psi_{C'}  } "4";"2"};
{\ar@{=>}|-{  \Psi_f } (17,17)*{};(5, 5)*{}} ;
{\ar@/_ 1.15pc/_{ \scriptscriptstyle F (f')} "1";"2"};
{\ar@{=>}^{ \scriptscriptstyle F(\gamma )} (10,1)*{};(10,-3)*{}} ;
\endxy
   \xy
  (0,0)*+{ }="1";
(10,0)*+{ }="2";
{\ar@3{->}^{\Psi_\gamma}  "1"+(0,10);"2"+(0,10)};
   \endxy
   \xy
(0,0)*+{\scriptstyle F(C)}="1";
(20,0)*+{\scriptstyle F(C')}="2";
(0,20)*+{\scriptstyle {\widetilde{F}}(C)}="3";
(20,20)*+{\scriptstyle {\widetilde{F}}(C')}="4";
{\ar@/_1.15pc/_{\scriptscriptstyle F(f') } "1";"2"};
{\ar@/^ 1.15pc/^{ \scriptscriptstyle {\widetilde{F}}(f)} "3";"4"};
{\ar@{->}|-{\Psi_C  } "3";"1"};
{\ar@{->}|-{\Psi_{C'}  } "4";"2"};
{\ar@{=>}|-{   \Psi_{f'} } (15,13)*{};(5, 0)*{}} ;
{\ar@/_ 1.15pc/_{ \scriptscriptstyle {\widetilde{F}} (f')} "3";"4"};
{\ar@{=>}^{ \scriptscriptstyle {\widetilde{F}}(\gamma) } (10,21)*{};(10,17)*{}} ;
\endxy
\end{equation*}

satisfying the following conditions:

(1) (naturality) for every $3$-arrow $\varphi: \xy 0;/r.22pc/:
(0,10)*{};
(0,-10)*{};
(-20,0)*+{C}="1";
(0,0)*+{C'}="2";
{\ar@/^1.33pc/^f "1";"2"};
{\ar@/_1.33pc/_{f'} "1";"2"};
(-10,4)*+{}="A";
(-10,-4)*+{}="B";
{\ar@{=>}@/_.5pc/_\gamma "A"+(-1.33,0) ; "B"+(-.66,-.55)};
{\ar@{=}@/_.5pc/ "A"+(-1.33,0) ; "B"+(-1.33,0)};
{\ar@{=>}@/^.5pc/^{\gamma'} "A"+(1.33,0) ; "B"+(.66,-.55)};
{\ar@{=}@/^.5pc/ "A"+(1.33,0) ; "B"+(1.33,0)};
{\ar@3{->} (-12,0)*{}; (-7,0)*{}};
\endxy $  in $\mathcal{C }$,
\begin{equation}\label{eq:naturality1} \Psi_{\gamma  }\#_2[
   (  {\widetilde{F}}( \varphi )\#_0\Psi_{C'})\#_1\Psi_{f'  }]=
   \left  [\Psi_f\#_1(\Psi_C\#_0 F(\varphi))\right]\#_2 \Psi_{\gamma'  }:
\end{equation}
\begin{equation}\label{eq:naturality0}
   \xy
(0,0)*+{\scriptstyle F(C)}="1";
(30,0)*+{\scriptstyle F(C')}="2";
(0,30)*+{\scriptstyle {\widetilde{F}}(C)}="3";
(30,30)*+{\scriptstyle {\widetilde{F}}(C')}="4";
{\ar@{-->}@/^ 1.15pc/^{\scriptscriptstyle F(f) } "1";"2"};
{\ar@/^ 1.15pc/^{ } "3";"4"};
 {\ar@/_1.15pc/_{ \scriptscriptstyle {\widetilde{F}} (f')} "3";"4"};
{\ar@{->}|-{\Psi_C  } "3";"1"};
{\ar@{->}|-{\Psi_{C'}  } "4";"2"};
{\ar@3{->}|-{  \Psi_{\gamma  } } (10,10)*{};(20, 20)*{}} ;
{\ar@/_ 1.15pc/_{ \scriptscriptstyle { {F}} (f')} "1";"2"};
{\ar@{=>}^{ \scriptscriptstyle F(\gamma  )} (15,1)*{};(20,-3)*{}} ;
{\ar@{=>}_{ \scriptscriptstyle {\widetilde{F}}(\gamma  )} (15,31)*{};(20,27)*{}} ;
 ( 7,35)*+{}="7";
(20,30)*+{}="8";
{\ar@{=>}@/^1.5pc/^{\scriptscriptstyle {\widetilde{F}}(\gamma')} "7" ; "8" };
{\ar@3{<-}|-{\scriptscriptstyle {\widetilde{F}}( \varphi )} (15,38)*{};(15, 31)*{}} ;
\endxy\xy
  (0,0)*+{ }="1";
(10,0)*+{ }="2";
{\ar@{=}   "1"+(0,15);"2"+(0,15)};
   \endxy
       \xy
(0,0)*+{\scriptstyle F(C)}="1";
(30,0)*+{\scriptstyle F(C')}="2";
(0,30)*+{\scriptstyle {\widetilde{F}}(C)}="3";
(30,30)*+{\scriptstyle {\widetilde{F}}(C')}="4";
{\ar@{-->}@/^ 1.15pc/^{\scriptscriptstyle F(f) } "1";"2"};
{\ar@/^ 1.15pc/^{ \scriptscriptstyle {\widetilde{F}}(f)} "3";"4"};
{\ar@/_1.15pc/_{ \scriptscriptstyle {\widetilde{F}} (f')} "3";"4"};
{\ar@{->}|-{\Psi_C  } "3";"1"};
{\ar@{->}|-{\Psi_{C'}  } "4";"2"};
{\ar@3{->}|-{  \Psi_{\gamma'} } (10,10)*{};(20, 20)*{}} ;
{\ar@/_ 1.15pc/_{  } "1";"2"};
{\ar@{=>}^{ \scriptscriptstyle F(\gamma' )} (15,1)*{};(20,-3)*{}} ;
{\ar@{=>}^{ \scriptscriptstyle {\widetilde{F}}(\gamma' )} (15,31)*{};(20,27)*{}} ;
 ( 7,-5)*+{}="7";
(20,-4)*+{}="8";
{\ar@{=>}@/_1.5pc/_{\scriptscriptstyle F(\gamma)} "7" ; "8" };
{\ar@3{->}^{\scriptscriptstyle F( \varphi) } (10,-5)*{};(16, -1)*{}} ;
\endxy
   \end{equation}

(2) (functoriality    with respect to $0$-composition of $1$-arrows) for every $C
\xrightarrow{f} C'
\xrightarrow{f'}C''
 $
in  $\mathcal{C}$, $
     \Psi_{f\#_0 f'}=[{\widetilde{F}} (f)\#_0\Psi_{f'}]\#_1[\Psi_{f }\#_0F (f')]$:
\begin{equation*}
       \xy
(0,0)*+{\scriptstyle F(C)}="1";
(20,0)*+{\scriptstyle F(C'')}="2";
(0,20)*+{\scriptstyle {\widetilde{F}}(C)}="3";
(20,20)*+{\scriptstyle {\widetilde{F}}(C'')}="4";
{\ar@{->}"1";"2"_{\scriptscriptstyle F(f\#_0 f') } };
{\ar@{->}  "3";"4"^{ \scriptscriptstyle {\widetilde{F}} (f\#_0 f')}};
{\ar@{->}|-{\Psi_C  } "3";"1"};
{\ar@{->}|-{\Psi_{C''}  } "4";"2"};
{\ar@{=>}|-{    \Psi_{f\#_0 f'} } (18,18)*{};(2, 2)*{}} ;
\endxy\xy
  (0,0)*+{ }="1";
(7,0)*+{ }="2";
{\ar@{=}   "1"+(0,10);"2"+(0,10)};
   \endxy\xy
(0,0)*+{\scriptstyle F(C)}="1";
(20,0)*+{\scriptstyle F(C')}="2";
(0,20)*+{\scriptstyle {\widetilde{F}}(C)}="3";
(20,20)*+{\scriptstyle {\widetilde{F}}(C')}="4";
(40, 0)*+{\scriptstyle F(C'')}="5";
(40,20)*+{\scriptstyle {\widetilde{F}}(C'')}="6";
{\ar@{->}|-{\scriptscriptstyle F(f) } "1";"2"};
{\ar@{->}|-{ \scriptscriptstyle {\widetilde{F}} (f)} "3";"4"};
{\ar@{->}|-{\Psi_C  } "3";"1"};
{\ar@{->}|-{\Psi_{C'}  } "4";"2"};
{\ar@{=>}|-{   \Psi_f } (18,18)*{};(2, 2)*{}} ;
{\ar@{->}|-{\scriptscriptstyle {\widetilde{F}}(f') } "4";"6"};
{\ar@{->}|-{ \scriptscriptstyle F (f')} "2";"5"};
{\ar@{->}|-{\Psi_{C''}  } "6";"5"};
{\ar@{=>}|-{  \Psi_{f'} } (38,18)*{};(22, 2)*{}} ;
\endxy\end{equation*}

(3) (functoriality   with respect to $1$-composition of $2$-arrows) for every  $\xy
(-8,0)*+{C}="4";
(8,0)*+{C'}="6";
{\ar@{->}|-{f'} "4";"6"};
{\ar@/^1.55pc/^{f } "4";"6"};
{\ar@/_1.55pc/_{f''} "4";"6"};
{\ar@{=>}^<<{\scriptstyle \gamma} (0,6)*{};(0,1)*{}} ;
{\ar@{=>}^<<{\scriptstyle \gamma'} (0,-1)*{};(0,-6)*{}} ;
\endxy  $   in $\mathcal{C }$, $
 \Psi_{\gamma\#_1\gamma'}   =\left[\Psi_\gamma\#_1(\Psi_C\#_0F(\gamma' ))\right]\#_2\left[\left({\widetilde{F}}(\gamma  )\#_0\Psi_{C'}\right)\#_1\Psi_{\gamma'}\right]:
$
\begin{equation*}
       \xy
(0,0)*+{\scriptstyle F(C)}="1";
(20,0)*+{\scriptstyle F(C')}="2";
(0,20)*+{\scriptstyle {\widetilde{F}}(C)}="3";
(20,20)*+{\scriptstyle {\widetilde{F}}(C')}="4";
{\ar@/^ 1.15pc/^{\scriptscriptstyle F(f) } "1";"2"};
{\ar@/^ 1.15pc/^{ \scriptscriptstyle {\widetilde{F}}(f)} "3";"4"};
{\ar@{->}|-{\Psi_C  } "3";"1"};
{\ar@{->}|-{\Psi_{C'}  } "4";"2"};
{\ar@{=>}|-{  \Psi_f } (17,17)*{};(5, 8)*{}} ;
{\ar@/_ 1.15pc/_{  } "1";"2"};
{\ar@{=>}^{ \scriptscriptstyle F(\gamma )} (10,1)*{};(10,-3)*{}} ;
{\ar@/_ 3.15pc/_{ \scriptscriptstyle F(f'')} "1";"2"};
{\ar@{=>}^{ \scriptscriptstyle F(\gamma' )} (10,-6)*{};(10,-11)*{}} ;
\endxy
   \xy
  (0,0)*+{ }="1";
(20,0)*+{ }="2";
{\ar@3{->}^{\scriptscriptstyle \Psi_\gamma\#_1[\Psi_C\#_0F(\gamma' )]}  "1"+(0,10);"2"+(0,10)};
   \endxy
   \xy
(0,0)*+{\scriptstyle F(C)}="1";
(20,0)*+{\scriptstyle F(C')}="2";
(0,20)*+{\scriptstyle {\widetilde{F}}(C)}="3";
(20,20)*+{\scriptstyle {\widetilde{F}}(C')}="4";
{\ar@/_1.15pc/^{\scriptscriptstyle F(f') } "1";"2"};
{\ar@/^ 1.15pc/^{ \scriptscriptstyle {\widetilde{F}}(f)} "3";"4"};
{\ar@{->}|-{\Psi_C  } "3";"1"};
{\ar@{->}|-{\Psi_{C'}  } "4";"2"};
{\ar@{=>}|-{   \Psi_{f'} } (15,12)*{};(4, -1)*{}} ;
{\ar@/_ 1.15pc/_{ \scriptscriptstyle {\widetilde{F}} (f')} "3";"4"};
{\ar@{=>}^{ \scriptscriptstyle {\widetilde{F}}(\gamma) } (10,23)*{};(10,17)*{}} ;
 {\ar@/_ 3.15pc/_{ \scriptscriptstyle F(f'')} "1";"2"};
{\ar@{=>}^{ \scriptscriptstyle F(\gamma' )} (10,-6)*{};(10,-11)*{}} ;
\endxy
    \xy
  (0,0)*+{ }="1";
(20,0)*+{ }="2";
{\ar@3{->}^{\scriptscriptstyle[  \widetilde F (\gamma  )\#_0 \Psi_{C'}]\#_1\Psi_{\gamma'}}  "1"+(0,10);"2"+(0,10)};
   \endxy
       \xy
(0,0)*+{\scriptstyle F(C)}="1";
(20,0)*+{\scriptstyle F(C')}="2";
(0,20)*+{\scriptstyle {\widetilde{F}}(C)}="3";
(20,20)*+{\scriptstyle {\widetilde{F}}(C')}="4";
{\ar@/^ 1.15pc/^{ \scriptscriptstyle {\widetilde{F}}(f)} "3";"4"};
{\ar@/_ 3.15pc/_{ \scriptscriptstyle {\widetilde{F}}(f'')} "3";"4"};
{\ar@{=>}^{ \scriptscriptstyle {\widetilde{F}}(\gamma' )} (10,14)*{};(10,9)*{}} ;
{\ar@{->}|-{\Psi_C  } "3";"1"};
{\ar@{->}|-{\Psi_{C'}  } "4";"2"};
{\ar@{=>}|-{   \Psi_{f''} } (15, 3)*{};(5, -7)*{}} ;
{\ar@/_ 1.15pc/_{  } "3";"4"};
{\ar@{=>}^{ \scriptscriptstyle {\widetilde{F}}(\gamma) } (10,21)*{};(10,17)*{}} ;
 {\ar@/_ 3.15pc/_{ \scriptscriptstyle F(f'')} "1";"2"};
\endxy
\end{equation*}

(4) (functoriality   with respect to $0$-composition of a $2$-arrow with a  $1$-arrow) for every

 $  \xy
(-8,0)*+{C}="4";
(8,0)*+{C'}="6";
{\ar@/^ .85pc/^{f} "4";"6"};
{\ar@/_ .85pc/_{f'} "4";"6"};
{\ar@{=>}^{\gamma } (0,2)*{};(0,-2)*{}} ;
\endxy   \xrightarrow{f''}C''$    in $\mathcal{C}$, \begin{equation*}
    \Psi_{\gamma\#_0 f''}=\left[\left(({\widetilde{F}}(f)\#_0 \Psi_{f''}\right)\#_1\left(\Psi_{\gamma } \#_0 F(f'')\right)\right]\#_2\left[\left({\widetilde{F}}(\gamma)\#_0 \Psi_{f''}\right)^{-1}\#_1(\Psi_{f' }\#_0F(f''))\right]:
\end{equation*}
\begin{equation*}
       \xy
(0,0)*+{\scriptstyle F(C)}="1";
(20,0)*+{\scriptstyle F(C')}="2";
(0,20)*+{\scriptstyle {\widetilde{F}}(C)}="3";
(20,20)*+{\scriptstyle {\widetilde{F}}(C')}="4";
(40,0)*+{\scriptstyle {{F}}(C'')}="5";
(40,20)*+{\scriptstyle\widetilde F(C'')}="6";
(30,10)*+{(1)}="11";
(10,15)*+{(2)}="12";(10,0)*+{(3)}="13";
{\ar@/^ 1.15pc/^{\scriptscriptstyle F(f) } "1";"2"};
{\ar@/^ 1.15pc/^{ \scriptscriptstyle {\widetilde{F}}(f)} "3";"4"};
{\ar@{->}|-{\Psi_C  } "3";"1"};
{\ar@{->}|-{\Psi_{C'}  } "4";"2"};
{\ar@/_ 1.15pc/_{ \scriptscriptstyle F (f')} "1";"2"};
{\ar@{->} "2";"5"_{\scriptscriptstyle F (f'')  }};
{\ar@{->} "4";"6"^{\scriptscriptstyle {\widetilde{F}} (f'')  }};
{\ar@{->}|-{\Psi_{C''}  } "6";"5"};
\endxy
   \xy
  (0,0)*+{ }="1";
(8,0)*+{ }="2";
{\ar@3{->}^{ }  "1"+(0,10);"2"+(0,10)};
   \endxy
   \xy
(0,0)*+{\scriptstyle F(C)}="1";
(20,0)*+{\scriptstyle F(C')}="2";
(0,20)*+{\scriptstyle {\widetilde{F}}(C)}="3";
(20,20)*+{\scriptstyle {\widetilde{F}}(C')}="4";
(40,0)*+{\scriptstyle {{F}}(C'')}="5";
(40,20)*+{\scriptstyle\widetilde F(C'')}="6";
(30,10)*+{(2)}="11";
(10,5)*+{(3)}="12";(10,20)*+{(1)}="13";
{\ar@/_1.15pc/_{\scriptscriptstyle F(f') } "1";"2"};
{\ar@/^ 1.15pc/^{ \scriptscriptstyle {\widetilde{F}}(f)} "3";"4"};
{\ar@{->}|-{\Psi_C  } "3";"1"};
{\ar@{->}|-{\Psi_{C'}  } "4";"2"};
{\ar@/_ 1.15pc/_{ \scriptscriptstyle {\widetilde{F}} (f')} "3";"4"};
{\ar@{->} "2";"5"_{\scriptscriptstyle F (f'')  }};
{\ar@{->} "4";"6"^{\scriptscriptstyle {\widetilde{F}} (f'')  }};
{\ar@{->}|-{\Psi_{C''}  } "6";"5"};
\endxy
\end{equation*}

(5) (functoriality   with respect to $0$-composition of a  $1$-arrow with a $2$-arrow) for every  $ C
  \xrightarrow{f }\xy
(-8,0)*+{C'}="4";
(8,0)*+{C''}="6";
{\ar@/^ .85pc/^{f'} "4";"6"};
{\ar@/_ .85pc/_{f''} "4";"6"};
{\ar@{=>}^{\gamma' } (0,2)*{};(0,-2)*{}} ;
\endxy$    in $\mathcal{C}$,\begin{equation*}
    \Psi_{f\#_0\gamma '}=\left[\left({\widetilde{F}}(f )\#_0 \Psi_{f' }\right)\#_1(\Psi_{f }\#_0 F(\gamma '))\right]\#_2 \left  [\left( {\widetilde{F}}(f)\#_0 \Psi_{\gamma' }\right)\#_1(\Psi_{f } \#_0 F(f''))\right]
 :\end{equation*}
\begin{equation*}
       \xy
(0,0)*+{\scriptstyle F(C')}="1";
(20,0)*+{\scriptstyle F(C'')}="2";
(0,20)*+{\scriptstyle {\widetilde{F}}(C')}="3";
(20,20)*+{\scriptstyle {\widetilde{F}}(C'')}="4";
(-20,0)*+{\scriptstyle {{F}}(C )}="5";
(-20,20)*+{\scriptstyle\widetilde F(C )}="6";
(10,15)*+{(1)}="11";
(10,0)*+{(3)}="12";(-10,10)*+{(2)}="13";
{\ar@/^1.15pc/^{\scriptscriptstyle F(f') } "1";"2"};
{\ar@/^ 1.15pc/^{ \scriptscriptstyle {\widetilde{F}}(f')} "3";"4"};
{\ar@{->}|-{\Psi_{C'}  } "3";"1"};
{\ar@{->}|-{\Psi_{C''}  } "4";"2"};
{\ar@/_ 1.15pc/_{ \scriptscriptstyle F (f'')} "1";"2"};
{\ar@{->}|-{\scriptscriptstyle F (f )  } "5";"1"};
{\ar@{->}|-{\scriptscriptstyle {\widetilde{F}} (f )  } "6";"3"};
{\ar@{->}|-{\Psi_{C  }  } "6";"5"};
\endxy
   \xy
  (0,0)*+{ }="1";
(8,0)*+{ }="2";
{\ar@3{->}^{ }  "1"+(0,10);"2"+(0,10)};
   \endxy
   \xy
(0,0)*+{\scriptstyle F(C')}="1";
(20,0)*+{\scriptstyle F(C'')}="2";
(0,20)*+{\scriptstyle {\widetilde{F}}(C')}="3";
(20,20)*+{\scriptstyle {\widetilde{F}}(C'')}="4";
(-20,0)*+{\scriptstyle {{F}}(C )}="5";
(-20,20)*+{\scriptstyle\widetilde F(C )}="6";
(-10,10)*+{(3)}="11";
(10,5)*+{(2)}="12";(10,20)*+{(1)}="13";
{\ar@/_1.15pc/_{\scriptscriptstyle F(f'') } "1";"2"};
{\ar@/^ 1.15pc/^{ \scriptscriptstyle {\widetilde{F}}(f')} "3";"4"};
{\ar@{->}|-{\Psi_{C'}  } "3";"1"};
{\ar@{->}|-{\Psi_{C'' }  } "4";"2"};
{\ar@/_ 1.15pc/_{ \scriptscriptstyle {\widetilde{F}} (f'')} "3";"4"};
{\ar@{->}|-{\scriptscriptstyle F (f )  } "5";"1"};
{\ar@{->}|-{\scriptscriptstyle {\widetilde{F}} (f )  } "6";"3"};
{\ar@{->}|-{\Psi_{C }  } "6";"5"};
\endxy
\end{equation*}

(6) (functoriality   with respect to identities) for every $C$ in $ \mathcal{C}$, $\Psi_{id_C} = id_{\Psi_C}$, and for every
$f : C \rightarrow C'$ in $ \mathcal{C}$, $\Psi_{id_f} = id_{\Psi_f }$.

\vskip 4mm
In the definition of $\Psi_{\gamma\#_0 f''}$ in (4) and $\Psi_{f\#_0\gamma' }$ in (5), the interchanging $3$-arrows  are   used to interchange the order of $2$-arrows.

\section{The $3$-connections and the $3$-gauge transformations}

 \subsection{$1$-path, $2$-path and $3$-path groupoids}

For  a positive integer $n$,  an  {\it  $n$-path}  is   a smooth map $\alpha : [0,1]^n = [0,1] \times [0,1]^{n-1}\rightarrow X$,  for which
there exists an $\epsilon > 0 $ such that $\alpha(t_1, \ldots , t_n) = \alpha(0, t_2, \ldots , t_n)$ for $t_1\leq \epsilon$, and analogously for any other face of $[0,1]^n$, of any
dimension. We will abbreviate this property as saying that $\alpha$ has a  {\it  product structure} close to the boundary of the $n$-cube.
We also require that $ \alpha(\{0\}\times [0,1]^{n-1})$ and $\alpha(\{1\}\times [0,1]^{n-1})$  both consist of just a single point.

 Given an $n$-path $\alpha$ and an $i\in\{1,  \ldots , n\}$,  we   define $(n-1)$-paths  $\partial_i^- (\alpha)$ and  $\partial_i^+ (\alpha)$  by restricting it to $[0,1]^{i-1}\times\{0\}\times [0,1]^{n-i}$
and $[0,1]^{i-1}\times\{1\}\times [0,1]^{n-i}$ . By definition,  $\partial_1^\pm (\alpha)$ must be constant $(n-1)$-paths. Given two $n$-paths $\alpha$ and $\beta$ with $\partial_i^+ (\alpha)=\partial_i^- (\beta)$,
 we have obviously the composition $\alpha\#_i\beta$ (we only consider $n=1,2,3$ in this paper). The product structure   guarantees $\alpha\#_i\beta$ also to be a
$n$-path. This is why the condition of product structure is imposed to an $n$-path.

For example,
a $1$-path is a smooth map
$\gamma : [0,1] \rightarrow X$ with sitting instants, i.e.,
there exists $ 0 <\epsilon < \frac 12$
 such that $\gamma
(t) =
\gamma(0)$ for $ 0 \leq t < \epsilon$ and
$\gamma (t) =\gamma
(1)$ for $1 -\epsilon< t \leq 1$.
Two $1$-paths $\gamma_1, \gamma_2: [0; 1] \rightarrow X$
  are called {\it rank-$1$ homotopic} (cf. \cite{CP})
  if there exists a $2$-path $\Gamma$ such that
(1) $\partial_2^- (\Gamma)= \gamma_1$, $\partial_2^+ (\Gamma)= \gamma_2$;
(2) the differential of $\Gamma$ at each point of $[0,1]^2$ has at most rank $1$.
The quotient of  the set of $1$-paths of $X$, by the relation of the rank-$1$ homotopy,  is denoted by $\mathcal{S}_1
(X)$. We call the elements of $\mathcal{S}_1
(X)$ $1$-tracks.
The category with objects $X$ and arrows    $\mathcal{S}_1
(X)$ is
  a groupoid,  called the {\it (thin) path groupoid} $\mathcal{P}_1(X)$
of $X$.

The quotient  of the set of $2$-paths of $X$, by the relation of the laminated rank-$2$ homotopy,  is denoted by $S^l_2
(X)$. We call the elements of $\mathcal S^l_2
(X)$ laminated $2$-tracks.
A $3$-path $(t_1, t_2, t_3)
\rightarrow J (t_1, t_2, t_3)$ is called {\it good} if the restrictions $\partial_2^\pm ( J )$ each are independent of $t_3$.
Denote by $\mathcal{S}_3(X)$  the set of all good $3$-paths up to the  rank-$3$ homotopy (with the laminated boundary).  (cf. \cite{MP11}  for the laminated rank-$2$ homotopy and rank-$3$ homotopy  with the laminated boundaries. We will not use these concepts precisely).
Vertical and horizontal compositions of laminated $2$-tracks, whiskering $2$- and $3$-tracks by $1$-tracks, the interchange $3$-tracks,   vertical   compositions of $3$-tracks, etc., are all well defined. Boundaries  $\partial_1^\pm  $ of good $3$-paths are   $0$-sources and   $0$-targets $\mathcal{S}_3(X)\rightarrow X$, boundaries  $\partial_2^\pm  $ of good $3$-paths are   $1$-sources and   $1$-targets $\mathcal{S}_3(X)\rightarrow \mathcal{S}_1(X)$,
and boundaries  $\partial_3^\pm  $ of good $3$-paths are   $2$-sources and  $2$-targets $\mathcal{S}_3(X)\rightarrow \mathcal{S}_2^l(X)$.

\begin{thm} {\it (Theorem 2.4 in \cite{MP11} )}
     Let $X$ be a smooth manifold. The sets of $1$-tracks, laminated $2$-tracks and $3$-tracks can be arranged into a Gray $3$-groupoid
$\mathcal{P}_3(X) =(X,\mathcal{S}_1(X),\mathcal{S}^l_2
(X)$, $\mathcal{S}_3(X))
$. In particular, $\mathcal{P}_2(X) =(X,\mathcal{S}_1(X),\mathcal{S}^l_2
(X) )
$ is automatically a $2$-groupoid.
 \end{thm}

  \subsection{Vanishing of  fake $1$-  and    $2$-curvatures} Given  a $3$-connection   $(A,B,C)$, we can
  construct $1$-, $2$- and $3$-dimensional holonomies, which constitute a smooth $\mathbf{Gray}$-functor from
 the $3$-groupoid $\mathcal{P}_3(X)$ to the $\mathbf{Gray}$ $3$-groupoid
$ \mathcal{G}^{\mathscr L}$ (cf. \cite{MP11}).  Conversely, let us   derive  a $3$-connection   $(A,B,C)$   as derivatives of a smooth $\mathbf{Gray}$-functor from the
 $3$-groupoid $\mathcal{P}_3(X)$ to the $\mathbf{Gray}$ $3$-groupoid
$ \mathcal{G}^{\mathscr L}$. Its fake $1$-  and    $2$-curvatures vanish in this case.

For $(x_1, x_2)\in \mathbb{R}^2$, choose
a  $2$-path $\dot{\Sigma}_{x_1, x_2}$ in $ \mathcal{S}^l_2(\mathbb{R}^2)$  to be
   \begin{equation}\label{eq:Sigma}\xy
  (0,0)*+{  }="1";
(10,0)*+{ }="2";(0,20)*+{[0,1]^2  }="3";(7,20)*+{ }="5";
(25,20)*+{ }="4";
{\ar@{->} "5";"4"^{\dot{\Sigma}_{x_1, x_2}  } };
   \endxy
  \xy 0;/r.17pc/:
( 15,15)*+{\scriptstyle   (0, x_2)}="3";
(55,15)*+{\scriptstyle ( x_1,x_2 )}="4";
(  15,55)*+{\scriptstyle(0,0)}="7";
(55,55)*+{\scriptstyle( x_1,0)}="8";
{\ar@{-->}|-{  } "7";"3" };
{\ar@{-->}|-{  } "3";"4" };
{\ar@{~>}|-{  } "8";"4" };
{\ar@{~>}|-{ } "7";"8" };
{\ar@{=>}|-{ \scriptscriptstyle \dot{\Sigma}_{x_1,x_2 }  } (50,50)*{};(20,20)*{}} ;
 \endxy \xy
  (0,0)*+{  }="1";
(10,0)*+{ }="2";(0,20)*+{   }="3";
(20,20)*+{ }="4";
{\ar@{->} "3";"4"^{\Gamma } };
   \endxy
  \xy 0;/r.17pc/:
( 15,15)*+{\bullet}="3";
(55,15)*+{\bullet}="4";
(  15,55)*+{\bullet}="7";
(55,55)*+{\bullet}="8";
{\ar@{-->}|-{\gamma^{x_2 } } "7";"3" };
{\ar@{-->}|-{\gamma^{x_1;x_2 } } "3";"4" };
{\ar@{~>}|-{ \gamma^{x_2;x_1 } } "8";"4" };
{\ar@{~>}|-{ \gamma^{x_1 }} "7";"8" };
{\ar@{=>}|-{ \Sigma_{x_1,x_2 } } (50,50)*{};(20,20)*{}} ;
 \endxy
      \end{equation}In the $2$-path $\dot{\Sigma}_{x_1, x_2}$,
  the wavy line is the $1$-source and  the dotted  line is the $1$-target.
      $\dot{\Sigma}_{x_1, x_2}$ can be constructed by dilation from one fixed $2$-path $\dot{\Sigma}_{1, 1}$. So it is a smooth family of $2$-paths $\dot{\Sigma}$ in $ \mathcal{S}^l_2(\mathbb{R}^2)$. It is important to see that the wavy and  the dotted lines are smooth $1$-path by the product structure, although their images in $\mathbb{R}^2$ are not smooth.

For fixed $x\in X$ and tangential vectors $  v_1, v_2\in T_xX$, choose a smooth mapping $\Gamma:\mathbb{R}^2\rightarrow   X$ such that $\Gamma(0)=x$ and
\begin{equation}\label{eq:Gamma}
     v_j= \frac {\partial\Gamma}{\partial x_j}(0,0) ,\qquad j=1,2.
\end{equation}
Then $\Sigma_{x_1, x_2}:=\Gamma\circ\dot{\Sigma}_{x_1, x_2}$ is a $2$-path in $\mathcal{S}^l_2(X)$.   We use notations $\gamma^{x_i;x_j }$ for the $1$-path of $ \Sigma_{x_1, x_2}$ corresponding to line    $  [0,x_i]\times \{x_j\}$ in   $\dot{\Sigma}_{x_1, x_2}$, and  $\gamma^{x_i  }$ for the $1$-path corresponding to line    $  [0,x_i]\times \{0\}$.

A smooth $\mathbf{Gray}$-functor $F:\mathcal{P}_3(X)\rightarrow \mathcal{G}^{\mathscr L}$ is given  by smooth mappings $F_0 :\mathcal{S}_0(X)\rightarrow  \{\bullet\}$, $ F_1:\mathcal{S}_1(X)\rightarrow  {G}$, $F_2 :\mathcal{S}^l_2(X)\rightarrow G\times {H}$ and $F_3 :\mathcal{S}_3(X)\rightarrow G\times {H}\times L$. Denote by $\pi_H:G\times {H}\rightarrow H$   the projection.
Then
\begin{equation}\label{eq:f-F}
     F_1(\gamma^{x_2 })F_1( \gamma^{x_1;x_2 }  ) =\alpha( \pi_H \circ F_2(\Sigma_{x_1, x_2}))^{-1} F_1( \gamma^{x_1 })F_1( \gamma^{x_2;x_1 } ).
\end{equation}
 Define
 \begin{equation}\label{eq:A-A-B}
     A_x(v_j)=\left.\frac {\partial F_1(\gamma^{x_j })}{\partial x_j}\right|_{x_j=0},
  \quad
    B_x(v_1,v_2)=  \left.\frac {\partial^2 \pi_H \circ F_2( {\Sigma}_{x_1, x_2})}{\partial x_1 \partial x_2 }\right|_{x_1=0,x_2=0},
 \end{equation}$j=1,2$. We claim that $B_x$ is a $2$-form, i.e., $B_x(v_1,v_2)=-B_x(v_2 ,v_1 )$ (cf. Lemma 3.7 in \cite{SW11}).
 Set $\overline{\Gamma}(s, t) := \Gamma(t, s)$.   Note that $\overline{{\Sigma}}_{x_1, x_2}=\overline{\Gamma}\circ \dot {\Sigma}_{x_1, x_2} = {\Gamma}\circ  \dot{{\Sigma}}_{x_1, x_2}^{-1} ={{\Sigma}}_{x_1, x_2}^{-1}$, where
$ {\Sigma}_{x_1, x_2}^{-1} $ is the $2$-arrow inverse to the $2$-arrow $\Sigma_{x_1, x_2}  $ under vertical composition.
Since the $2$-functor $F$ sends the inverse $2$-arrow  to the inverse group element, we have
 $\pi_H \circ F_2(\overline{\Sigma}_{x_1, x_2})=\pi_H \circ F_2( {\Sigma}_{x_1, x_2})^{-1} $. Hence, by taking derivatives, we get $B_x(v_2 ,v_1 )=-B_x(v_1,v_2)$. Moreover, $B$ is independent of the choice of the mapping $\Gamma$ in (\ref{eq:Gamma}) (cf. Lemma 3.6 in \cite{SW11}).

 Take derivatives $\frac {\partial^2}{\partial x_1 \partial x_2} $ at $(0,0)$ on both sides of (\ref{eq:f-F}) to get
 \begin{equation*}A_x(v_2)A_x(v_1)+v_2 A_x(v_1)=-\alpha (B_x(v_1,v_2))+A_x(v_1)A_x(v_2)+
     v_1 A_x(v_2),
 \end{equation*}
by using
\begin{equation}\label{eq:der-F=0}
   \left. \frac {\partial \pi_H \circ F_2( {\Sigma}_{x_1, x_2})}{\partial x_j}\right |_{(0,0)}=0,
\end{equation}
which follows from  $ \pi_H \circ F_2( {\Sigma}_{0, x_2})   =\pi_H \circ F_2( {\Sigma}_{x_1, 0})=1_H$. Thus,
  $
     \mathcal{ F}_1=   dA+A\wedge A-\alpha(B)=0.
   $

 Now for fixed $(x_1,x_2,x_3)\in \mathbb{R}^3$,   choose
a  good $3$-path $\dot{{\Omega}}_{x_1,x_2,x_3}  $ in $ \mathcal{{S}}_3(\mathbb{R}^3)$  to be
  \begin{equation}\label{eq:Omega}\xy
  (0,0)*+{  }="1";
(10,0)*+{ }="2";(0,20)*+{[0,1]^3  }="3";(7,20)*+{ }="5";
(25,20)*+{ }="4";
{\ar@{->} "5";"4"^{\dot{{\Omega}}_{x_1,x_2,x_3}   } };
   \endxy
   \xy 0;/r.17pc/:
   (0,0 )*+{\scriptscriptstyle(0,x_2,x_3) }="1";
(40, 0)*+{\scriptscriptstyle(x_1,x_2,x_3) }="2";
( 15,15)*+{\scriptscriptstyle(0,0,x_3)}="3";
(55,15)*+{\scriptscriptstyle(x_1,0,x_3) }="4";
(0,40 )*+{\scriptscriptstyle(0,x_2,0) }="5";
(40, 40)*+{\scriptscriptstyle(x_1,x_2,0) }="6";
(  15,55)*+{\scriptscriptstyle(0,0,0) }="7";
(55,55)*+{\scriptscriptstyle (x_1,0,0)}="8";
 {\ar@{-->}|-{ }  "1";"2" };
{\ar@{-->}|-{ } "3";"1" };
{\ar@{->}|-{ } "5";"1" };
{\ar@{->}|-{ } "4";"2" };
{\ar@{~>}|-{ } "6";"2" };
{\ar@{-->}|-{ } "7";"3" };
{\ar@{->}|-{ } "3";"4" };
{\ar@{->}|-{ } "8";"4" };
{\ar@{->}|-{ } "5";"6" };
{\ar@{->}|-{ } "7";"5" };
{\ar@{~>}|-{ } "8";"6" };
{\ar@{~>}|-{ } "7";"8" };
 \endxy
\xy
  (0,0)*+{  }="1";
(10,0)*+{ }="2";(0,20)*+{   }="3";
(10,20)*+{ }="4";
{\ar@{->} "3";"4"^{\Gamma } };
   \endxy
     \xy 0;/r.17pc/:
(0,0 )*+{\bullet}="1";
(40, 0)*+{\bullet}="2";
( 15,15)*+{\bullet}="3";
(55,15)*+{\bullet}="4";
(0,40 )*+{\bullet}="5";
(40, 40)*+{\bullet}="6";
(  15,55)*+{\bullet}="7";
(55,55)*+{\bullet}="8";
{\ar@{-->}|-{\gamma^{x_1;x_2,x_3}} "1";"2" };
{\ar@{-->}|-{\gamma^{x_2 ;x_3}}"3";"1"};
{\ar@{->}|-{\gamma^{x_3; x_2}} "5";"1" };
{\ar@{->}|-{\gamma^{x_2;x_1,x_3}} "4";"2" };
{\ar@{~>}|-{\gamma^{x_3;x_1,x_2}} "6";"2" };
{\ar@{-->}|-{\gamma^{x_3}} "7";"3" };
{\ar@{->}|-{\gamma^{x_1;x_3}} "3";"4" };
{\ar@{->}|-{\gamma^{x_3;x_1}} "8";"4" };
{\ar@{->}|-{\gamma^{x_1;x_2 }} "5";"6" };
{\ar@{->}|-{\gamma^{x_2}} "7";"5" };
{\ar@{~>}|-{\gamma^{x_2;x_1 }} "8";"6" };
{\ar@{~>}|-{\gamma^{x_1}} "7";"8" };
 \endxy
\end{equation}In the $3$-path $\dot{\Omega}_{x_1,x_2,x_3}$,
  the wavy line is its $1$-source and  the dotted  line is its $1$-target.
      $\dot{\Omega}_{x_1,x_2,x_3}$ can be constructed by dilation from one fixed  good $3$-path $\dot{\Omega}_{ 1,1,1}$. So it is a smooth family of  good $3$-paths $\dot{\Omega}$ in $ \mathcal{S}_3(\mathbb{R}^3)$.

For fixed $x\in X$ and tangential vectors $  v_1, v_2, v_3\in T_xX$, choose a smooth mapping $\Gamma:\mathbb{R}^3\rightarrow   X$ such that $\Gamma(0,0,0)=x$ and
\begin{equation*}
     v_j= \frac {\partial\Gamma}{\partial x_j} (0,0,0) ,\qquad j=1,2,3.
\end{equation*}Then $\Omega_{x_1,x_2,x_3}:=\Gamma\circ\dot{\Omega}_{x_1,x_2,x_3}$ is a good $3$-path in $ \mathcal{S}_3(X)$.

We use notations $\gamma^{x_i;x_j,x_k}$   for the $1$-path of $ \Omega_{x_1,x_2,x_3}$ corresponding to line    $  [0,x_i]\times \{x_j\} \times \{x_k\}$ in   $\dot{\Omega}_{x_1,x_2,x_3}$,   $\gamma^{x_i  }$ for the $1$-path corresponding to line    $  [0,x_i]\times \{0\}\times \{0\}$, etc.. Similarly, we denote by $\Sigma_{x_i, x_j;x_k}$ the $2$-path of $ \Omega_{x_1,x_2,x_3}$ corresponding  to the $2$-cell $[0,x_i]\times [0,x_j]\times \{x_k\}$ in   $\dot{\Omega}_{x_1,x_2,x_3}$, and by $\Sigma_{x_i, x_j }$  the $2$-path of $ \Omega_{x_1,x_2,x_3}$ corresponding  to the $2$-cell $[0,x_i]\times [0,x_j]\times \{0\}$, etc..

$ \Omega_{x_1,x_2,x_3}$ is a good $3$-path in $ \mathcal{S}_3(X)$ with the $2$-source $\Sigma_{- }$ and     the $2$-target $\Sigma_{+}$ as follows:
\begin{equation}\label{eq:Omega-boundary}
  \xy 0;/r.17pc/:
(0,0 )*+{\bullet}="1";
(40, 0)*+{\bullet}="2";
( 15,15)*+{\bullet}="3";
(55,15)*+{\bullet}="4";
(40, 40)*+{\bullet}="6";
(  15,55)*+{\bullet}="7";
(55,55)*+{\bullet}="8";
{\ar@{-->}|-{ } "1";"2" };
{\ar@{-->}|-{ } "3";"1" };
{\ar@{->} "4";"2"^{\gamma^{x_2;x_1,x_3} } };
{\ar@{~>}|-{ } "6";"2" };
{\ar@{-->}|-{\gamma^{x_3} } "7";"3" };
{\ar@{->}|-{ } "3";"4" };
{\ar@{->}|-{ } "8";"4" };
{\ar@{~>}|-{ } "8";"6" };
{\ar@{~>}|-{\gamma^{x_1} } "7";"8" };
{\ar@{=>}|-{\scriptscriptstyle \Sigma_{x_2,x_3 ; x_1} } (42,38)*{};(53,17)*{}} ;
{\ar@{=>}|-{\scriptscriptstyle \Sigma_{x_1,x_3 } } (40,53)*{};(18,18)*{}} ;
{\ar@{=>}|-{\scriptscriptstyle \Sigma_{x_1,x_2;x_3 } } (38,13)*{};(4,2)*{}} ;
 \endxy \xy
(0,0)*+{}="A";
(15,0)*+{}="B";
{\ar@3{->}"A" +(0,20); "B"+(0,20)^{ \Omega_{x_1,x_2,x_3} }};
\endxy
  \xy 0;/r.17pc/:
(0,0 )*+{\bullet}="1";
(40, 0)*+{\bullet}="2";
( 15,15)*+{\bullet}="3";
(0,40 )*+{\bullet}="5";
(40, 40)*+{\bullet}="6";
(  15,55)*+{\bullet}="7";
(55,55)*+{\bullet}="8";
{\ar@{-->}|-{ } "1";"2"_{\gamma^{x_1;x_2,x_3} } };
{\ar@{-->}|-{ } "3";"1" };
{\ar@{->}|-{ } "5";"1" };
{\ar@{~>}|-{ } "6";"2"^{\gamma^{x_3;x_1,x_2} } };
{\ar@{-->}|-{ } "7";"3" };
{\ar@{->}|-{ } "5";"6" };
{\ar@{->}|-{\gamma^{ x_2} } "7";"5" };
{\ar@{~>}|-{ } "8";"6" };
{\ar@{~>}|-{ } "7";"8" };
{\ar@{=>}|-{\scriptscriptstyle \Sigma_{x_1,x_2 } } (51,53)*{};(17,42)*{}} ;
{\ar@{=>}|-{\scriptscriptstyle \Sigma_{x_1,x_3;x_2 } } (38,38)*{};(12,3)*{}} ;
{\ar@{=>}|-{\scriptscriptstyle \Sigma_{x_2,x_3  } } (2,38)*{};(13,17)*{}} ;
 \endxy
    \end{equation}
More precisely, we have a $3$-arrow  $(*,\Sigma_-,\Omega_{x_1,x_2,x_3})$ in the $3$-groupoid
$\mathcal{P}_3(X)$ with
\begin{equation}\label{eq:Sigma+-}\begin{split}  &\Sigma_-:=\left[\gamma^{x_1 }\#_0\Sigma_{x_2,x_3 ; x_1} \right]\#_1[\Sigma_{x_1,x_3 }\#_0\gamma^{x_2;x_1,x_3 }]\#_1[\gamma^{ x_3 }\#_0 \Sigma_{x_1,x_2;x_3 }]
\\
 \xy
(0,0)*+{}="A";
(10,0)*+{}="B";
{\ar@3{->}"A"+(0,1)  ; "B"+(0,1) };
\endxy  &\Sigma_+:=
[ \Sigma_{x_1,x_2 }\#_0\gamma^{x_3;x_1,x_2}]\#_1
  [\gamma^{x_2}\#_0\Sigma_{x_1,x_3; x_2}]\#_1[\Sigma_{x_2,x_3  } \#_0 \gamma^{x_1;x_2,x_3}].\end{split}\end{equation}
  The $2$-source $\Sigma_{- }$ and     the $2$-target $\Sigma_{+}$ are compositions of 3 whiskered $2$-arrows, respectively,  as follows:
  \begin{equation}\begin{split}  &
  \xy 0;/r.17pc/:
(0,0 )*+{\bullet}="1";
(40, 0)*+{\bullet}="2";
( 20,20)*+{\bullet}="3";
(20,-20)*+{\bullet}="4";
(-20,  0)*+{\bullet}="5";(20,0)*+{}="A";
 {\ar@{~>}|-{ } "1";"3" };
{\ar@{->}|-{ } "1";"4" };
{\ar@{~>}|-{  } "3";"2" };
{\ar@{->}|-{  } "4";"2"_{\gamma^{x_2;x_1,x_3} }  };
{\ar@{~>}|-{\gamma^{x_1 } } "5";"1" };
{\ar@{=>}|-{\scriptscriptstyle \Sigma_{x_2,x_3 ; x_1}  } (20,16)*{};(20,-16)*{}}
 \endxy,\qquad\qquad  \xy 0;/r.17pc/:
(0,0 )*+{\bullet}="1";
(40, 0)*+{\bullet}="2";
( 20,20)*+{\bullet}="3";
(20,-20)*+{\bullet}="4";
(60,  0)*+{\bullet}="5";{\ar@{=>}|-{\scriptscriptstyle \Sigma_{x_1,x_2 }  } (20,16)*{}; (20,-16)*{}} ;
 {\ar@{~>}|-{ } "1";"3" };
{\ar@{->}|-{ \gamma^{x_2 }} "1";"4" };
{\ar@{~>}|-{ } "3";"2" };
{\ar@{->}|-{ } "4";"2" };
{\ar@{~>}|-{}"2" ;"5"^{\gamma^{x_3;x_1,x_2} }};
 \endxy, \\&
  \xy 0;/r.17pc/:
(0,0 )*+{\bullet}="1";
(40, 0)*+{\bullet}="2";
( 20,20)*+{\bullet}="3";
(20,-20)*+{\bullet}="4";
(60,  0)*+{\bullet}="5"; {\ar@{=>}|-{\scriptscriptstyle   \Sigma_{x_1,x_3  } } (20,16)*{}; (20,-16)*{}};
 {\ar@{~>}|-{\gamma^{x_1 }  } "1";"3" };
{\ar@{-->}|-{\gamma^{x_3 }   } "1";"4" };
{\ar@{->}|-{ } "3";"2" };
{\ar@{->}|-{ } "4";"2" };
{\ar@{->}|-{  } "2";"5"^{\gamma^{x_2;x_1,x_3} } };
 \endxy,\qquad\qquad  \xy 0;/r.17pc/:
(0,0 )*+{\bullet}="1";
(40, 0)*+{\bullet}="2";
( 20,20)*+{\bullet}="3";
(20,-20)*+{\bullet}="4";
(-20,  0)*+{\bullet}="5"; {\ar@{=>}|-{\scriptscriptstyle   \Sigma_{x_1,x_3 ; x_2} } (20,16)*{};(20,-16)*{}};
 {\ar@{->}|-{ } "1";"3" };
{\ar@{->}|-{ } "1";"4" };
{\ar@{~>}|-{ } "3";"2"^{\gamma^{x_3;x_1,x_2} } };
{\ar@{-->}|-{ } "4";"2"_{ \gamma^{x_1;x_2,x_3}} };
{\ar@{->}|-{\gamma^{x_2 } } "5";"1" };
 \endxy, \\&
 \xy 0;/r.17pc/:
(0,0 )*+{\bullet}="1";
(40, 0)*+{\bullet}="2";
( 20,20)*+{\bullet}="3";
(20,-20)*+{\bullet}="4";
(-20,  0)*+{\bullet}="5"; {\ar@{=>}|-{\scriptscriptstyle \Sigma_{x_1 ,x_2;x_3 }  } (20,16)*{}; (20,-16)*{}};
 {\ar@{->}|-{ } "1";"3" };
{\ar@{-->}|-{ } "1";"4" };
{\ar@{->}|-{ } "3";"2"^{\gamma^{x_2;x_1,x_3} } };
{\ar@{-->}|-{ } "4";"2" };
{\ar@{-->}|-{  } "5";"1"^{\gamma^{x_3 } } };
 \endxy, \qquad\qquad  \xy 0;/r.17pc/:
(0,0 )*+{\bullet}="1";
(40, 0)*+{\bullet}="2";
( 20,20)*+{\bullet}="3";
(20,-20)*+{\bullet}="4";
(60,  0)*+{\bullet}="5"; {\ar@{=>}|-{\scriptscriptstyle   \Sigma_{ x_2 , x_3} } (20,16)*{}; (20,-16)*{}};
 {\ar@{->}|-{\gamma^{x_2 } } "1";"3" };
{\ar@{-->}|-{ } "1";"4" };
{\ar@{->}|-{ } "3";"2" };
{\ar@{-->}|-{ } "4";"2" };
{\ar@{-->}|-{ } "2" ;"5"^{\gamma^{x_1;x_2,x_3}} };
 \endxy.
 \end{split}   \end{equation}

Then $ \hat F_2( \Sigma_+)=\delta(\hat F_3(\Omega_{x_1,x_2,x_3}))^{-1} \hat F_2(\Sigma_-)$, i.e.,
 \begin{equation}\label{eq:F+-}\begin{split}
 &  \hat F_2( \Sigma_{x_1,x_2 })\cdot F_1(\gamma^{x_2 })\rhd
 \hat F_2(  \Sigma_{x_1,x_3; x_2})\cdot  \hat F_2(\Sigma_{x_2,x_3  })\\=& \delta(\hat F_3(\Omega_{x_1,x_2,x_3}))^{-1} \cdot F_1(\gamma^{x_1 })\rhd  \hat F_2( \Sigma_{x_2,x_3 ; x_1 })\cdot
  \hat F_2(  \Sigma_{x_1,x_3 })\cdot F_1(\gamma^{x_3 })\rhd \hat F_2(\Sigma_{x_1,x_2;x_3   }),
\end{split} \end{equation} where $ \hat F_2=\pi_H \circ F_2$, $ \hat F_3=\pi_L \circ F_3$ and $\pi_L:G\times {H}\times L\rightarrow L$ is the projection. Here we use Remark \ref{rem:whiskering} that in the Gray $3$-groupoid $ \mathcal{ G}^{\mathscr L}$,    whiskering from right by a $ 1$-arrow   is trivial.
Set
\begin{equation}\label{eq:connect-C}
     C(v_1,v_2,v_3):= \left.\frac {\partial^3 \hat F_3(\Omega_{x_1,x_2,x_3})}{\partial x_1\partial x_2\partial x_3}\right|_{x_1=x_2=x_3=0}.
\end{equation}
$C$ is a $3$-form, in the same way as $B$ is a $2$-form.
 Take derivatives $\frac {\partial^3}{\partial x_1\partial x_2\partial x_3}$ at $(0,0,0)$ on both sides of (\ref{eq:F+-}), noting (\ref{eq:der-F=0}), to get
\begin{equation*}\begin{split}&A(v_2)\rhd B(v_1,v_3) +v_2  B(v_1,v_3)  \\
     =& -\delta(C(v_1,v_2,v_3))+
     A(v_1)\rhd B(v_2,v_3)+v_1  B(v_2,v_3) +A(v_3)\rhd B(v_1,v_2)+v_3 B(v_1,v_2).
\end{split}\end{equation*}
Thus,
 $
       \mathcal{F}_2=  dB+A \wedge^\rhd B-\delta(C)=0.
  $

\subsection{The  gauge transformations}
 \begin{prop} Suppose $(A,B,C)$ and $(\widetilde{A},\widetilde{B},\widetilde{C})$ are $3$-connections constructed   from smooth $\mathbf{Gray}$-functors $ {F}$ and $\widetilde{F} :\mathcal{P}_3(X)\rightarrow \mathcal{G}^{\mathscr L}$, respectively, in the above subsection, and there exists a lax-natural transformations $\Psi:{F} \rightarrow\widetilde{F}$. Then there exist some $g\in \Lambda^0(X,G)$, $\varphi\in   \Lambda^1(X,\mathfrak h), \psi\in\Lambda^1(X,\mathfrak l)$ such that \begin{equation}\label{eq:C}\begin{split}
       &  \widetilde{A} =Ad_g A+gdg^{-1}+\alpha(\varphi );
    \\&
         \widetilde{B} = g \rhd B+d\varphi+\widetilde{A} \wedge^{\rhd}\varphi-\varphi\wedge\varphi-\delta(\psi);
     \\&
         \widetilde{C }= g \rhd C - d\psi-\widetilde{A }\wedge^{\rhd }\psi +\varphi\wedge^{\rhd'}\psi-\widetilde{B}\wedge^{\{,\}}\varphi-\varphi\wedge^{\{,\}}( g \rhd{B}).
         \end{split} \end{equation}
\end{prop}
This is exactly the $3$-gauge transformation in (\ref{eq:gauge-transformations}) with $g$ replaced by $g^{-1}$.
The transformation  formula for the $A$ field  is easy.
Let $\Omega_{x_1,x_2,x_3}$ be the  $3$-path   in $ \mathcal{S}_3(X)$ in (\ref{eq:Omega}) and (\ref{eq:Omega-boundary}).
Set
\begin{equation*}\begin{split}&
     f^{x_j;*}:= F_1( \gamma^{x_j;*}),\qquad\qquad\qquad\quad
    \widetilde{ f }^{x_j;*}:= \widetilde{F}_1 (  \gamma^{x_j;*}) ,\\&
     F^{x_i,x_j;*}:=\pi_H\circ F_2\left(\Sigma_{x_i,x_j;*}\right),\qquad \widetilde{F}^{x_i,x_j;*}:=\pi_H\circ \widetilde{F}_2\left(\Sigma_{x_i,x_j;*}\right),
     \\&
     h^{x_j;*}:=\pi_H\circ \Psi_{ \gamma^{x_j;*}},\qquad\qquad\qquad
    k^{x_i,x_j;*}:=\pi_L\circ \Psi_{\Sigma_{x_i,x_j;*}},
 \end{split}\end{equation*}
 for $*$= empty or $x_k$, etc.. For a $G$-valued function $g$ on $X$, we denote
 \begin{equation*}
  g
^{x_j}=g(\Gamma(x_j,0,0)) ,\qquad  g^{x_j, x_k}=g(\Gamma(x_j,x_k,0))\qquad{\rm  and }\qquad g^{x_j, x_k,x_l}=g(\Gamma(x_j,x_k,x_l)),\end{equation*}  up to the order of coordinates. Note that
 \begin{equation}\label{eq:value-0}
   f^{x_j ;*}|_{x_j=0}=1_G,\qquad  h^{x_j ;*} |_{x_j=0}=1_L,\qquad\left. \frac {\partial k^{x_i,x_j;*}}{\partial x_s}\right|_{x_i=x_j=0} =\left.\frac {\partial F^{x_i,x_j;*}}{\partial x_s}\right|_{x_i=x_j=0} =0,
 \end{equation} for $s=i$ or $j$. The last identity comes from $F^{x_i,x_j;*}=1_H$ if $x_i=0$  or $x_j=0$.

 By definition, the lax-natural transformation  $\Psi$ defines a $3$-arrow $\Psi_{\Sigma_{x_1,x_2}}$,
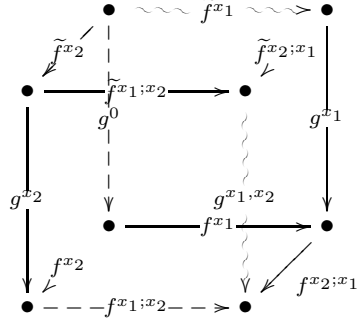
\begin{equation*}
     \xy 0;/r.17pc/:
(0,0 )*+{\bullet}="1";
(40, 0)*+{\bullet}="2";
( 15,15)*+{\bullet}="3";
(55,15)*+{\bullet}="4";
(0,40 )*+{\bullet}="5";
(40, 40)*+{\bullet}="6";
(  15,55)*+{\bullet}="7";
(55,55)*+{\bullet}="8";
{\ar@{-->}|-{f^{x_1;x_2 }} "1";"2" };
{\ar@{-->} |-{f^{x_2}}"3";"1" };
{\ar@{->}|-{g^{x_2}} "5";"1" };
{\ar@{->} "4";"2"^{f^{x_2;x_1 }} };
{\ar@{~>}|-{g^{ x_1,x_2}} "6";"2" };
{\ar@{-->}|-{g^{ 0}} "7";"3" };
{\ar@{->}|-{f^{x_1 }} "3";"4" };
{\ar@{->}|-{g^{x_1}} "8";"4" };
{\ar@{->}|-{\widetilde{f}^{x_1;x_2 }} "5";"6" };
{\ar@{->}|-{\widetilde{f}^{x_2 }} "7";"5" };
{\ar@{~>}|-{\widetilde{f}^{x_2;x_1}} "8";"6" };
{\ar@{~>}|-{\widetilde{f}^{x_1 }} "7";"8" };
 \endxy
\end{equation*}
such that
 \begin{equation}\label{eq:f-g-delta}\delta(k^{x_1,x_2})^{-1} \cdot
   \widetilde{f}^{x_1}\rhd h^{x_2 ;x_1}\cdot h^{x_1 }\cdot   g^{0}\rhd F^{x_1,x_2}
            =\widetilde{F}^{x_1,x_2} \cdot\widetilde{f}^{x_2}\rhd h^{x_1;x_2}\cdot   h^{x_2 } .
 \end{equation}
Let $A(v)$ and $B(v_1,v_2)$ as before and let
\begin{equation*}\begin{split}&
     \varphi(v_1):=\left.\frac {\partial h^{x_1}}{\partial x_1  }\right|_{x_1=0},\qquad\varphi(v_2):=\left.\frac {\partial h^{x_2}}{\partial x_2  }\right|_{x_2=0},\qquad \psi(v_1,v_2):=\left.\frac {\partial^2 k^{x_1,x_2}}{\partial x_1\partial x_2 }\right|_{x_1=x_2=0} .
\end{split}\end{equation*} We can
take derivatives $\frac {\partial^2}{\partial x_1  \partial x_2  }$ at $(0,0)$ on both sides of (\ref{eq:f-g-delta}) and use (\ref{eq:value-0}) to get
 \begin{equation*}\begin{split}&
   -  \delta(\psi(v_1,v_2 ))+ \widetilde{A}(v_1)\rhd \varphi(v_2 )+v_1  \varphi(v_2 ) +\varphi(v_2)\varphi(v_1)+g(x)\rhd B (v_1,v_2) \\
     =& \widetilde{B}(v_1,v_2)+\widetilde{A}(v_2)\rhd \varphi(v_1 )+v_2  \varphi(v_1 ) +  \varphi(v_1)     \varphi(v_2).
\end{split}\end{equation*}This is exactly the transformations  formula for the $B$ field in (\ref{eq:C}).
 \subsection{The  gauge transformations of the $C$ field: the $\Psi_{\Sigma_-}$ part}

 $F(\Omega_{x_1,x_2,x_3})$ is a $3$-arrow in $\mathcal{G}^{\mathscr L}$, whose $2$-source $F_2(\Sigma_-)$ and $2$-target $F_2(\Sigma_+)$ (cf. (\ref{eq:Omega-boundary}) for $ \Sigma_- $ and  $ \Sigma_+ $ ) are as follows:
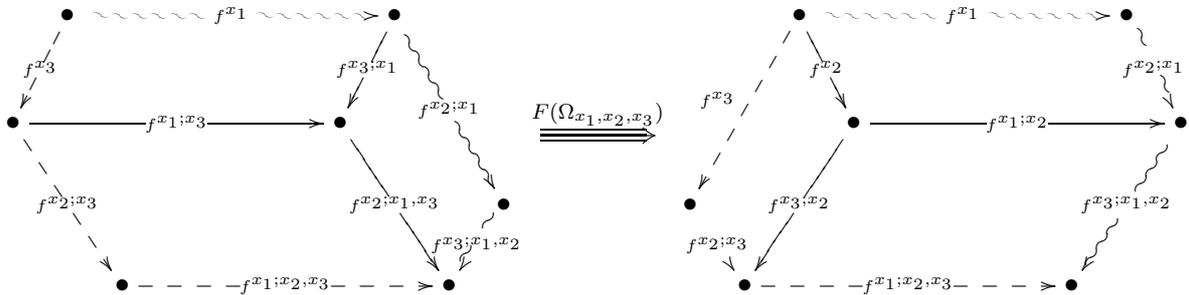
\begin{equation*}
  \xy 0;/r.17pc/:
(0,0 )*+{\bullet}="1";
(60, 0)*+{\bullet}="2";
( 70,15)*+{\bullet}="3";
(-20,30)*+{\bullet}="4";
(40,30 )*+{\bullet}="5";
(-10, 50)*+{\bullet}="6";
(  50,50)*+{\bullet}="7";
 {\ar@{-->}|-{\scriptscriptstyle f^{x_1;x_2,x_3 }} "1";"2" };
{\ar@{-->}|-{\scriptscriptstyle f^{ x_2; x_3}} "4";"1" };
{\ar@{~>}|- {\scriptscriptstyle f^{ x_3;x_1,x_2}} "3";"2" };
{\ar@{->}|-{\scriptscriptstyle f^{x_2; x_1 ,x_3}} "5";"2" };
{\ar@{->}|-{\scriptscriptstyle f^{ x_1;x_3}} "4";"5" };
{\ar@{-->}|-{\scriptscriptstyle f^{x_3 }} "6";"4" };
{\ar@{~>}|-{\scriptscriptstyle f^{x_1}} "6";"7" };
{\ar@{~>}|-{\scriptscriptstyle f^{x_2;x_1  }} "7";"3" };
{\ar@{->}|-{\scriptscriptstyle f^{x_3;x_1 }} "7";"5" };
  \endxy \xy
(0,0)*+{}="A";
(15,0)*+{}="B";
{\ar@3{->}"A" +(0,20); "B"+(0,20)^{F(\Omega_{x_1,x_2,x_3})}};
\endxy   \xy 0;/r.17pc/:
(0,0 )*+{\bullet}="1";
(60, 0)*+{\bullet}="2";
(-10,15)*+{\bullet}="3";
(20,30)*+{\bullet}="4";
(80,30 )*+{\bullet}="5";
( 10, 50)*+{\bullet}="6";
(  70,50)*+{\bullet}="7";
 {\ar@{-->}|-{\scriptscriptstyle f^{x_1;x_2, x_3 }} "1";"2" };
{\ar@{-->}|-{\scriptscriptstyle f^{ x_2; x_3}} "3";"1" };
{\ar@{<--} "3";"6"^{\scriptscriptstyle f^{ x_3 }} };
{\ar@{~>}|-{\scriptscriptstyle f^{x_3; x_1 ,x_2}}"5";"2"};
{\ar@{->}|-{\scriptscriptstyle f^{ x_1;x_2}} "4";"5" };
{\ar@{->}|-{\scriptscriptstyle f^{x_2 }} "6";"4" };
{\ar@{->}|-{\scriptscriptstyle f^{x_3;x_2}} "4";"1" };
{\ar@{~>}|-{\scriptscriptstyle f^{ x_1  }}"6"; "7"};
{\ar@{~>}|-{\scriptscriptstyle f^{x_2;x_1 }}"7";"5"};  \endxy
\end{equation*} Similarly for $3$-arrow $\widetilde{F}(\Omega_{x_1,x_2,x_3})$ in $\mathcal{G}^{\mathscr L}$.
The naturality (\ref{eq:naturality0})-(\ref{eq:naturality1}) of the lax-natural transformation $\Psi:F\rightarrow\widetilde{F}$ implies
\begin{equation}\label{eq:naturality}F^*_3(\Omega_{x_1,x_2,x_3})\cdot \Psi_{\Sigma_+}  =\Psi_{\Sigma_-}\cdot  \widetilde{F}_3(\Omega_{x_1,x_2,x_3})
.
\end{equation} where $ F^* (\Omega_{x_1,x_2,x_3})$ is a suitable whiskering of $ F (\Omega_{x_1,x_2,x_3})$.
$\Psi_{\Sigma_-}$ is the following $3$-arrow.
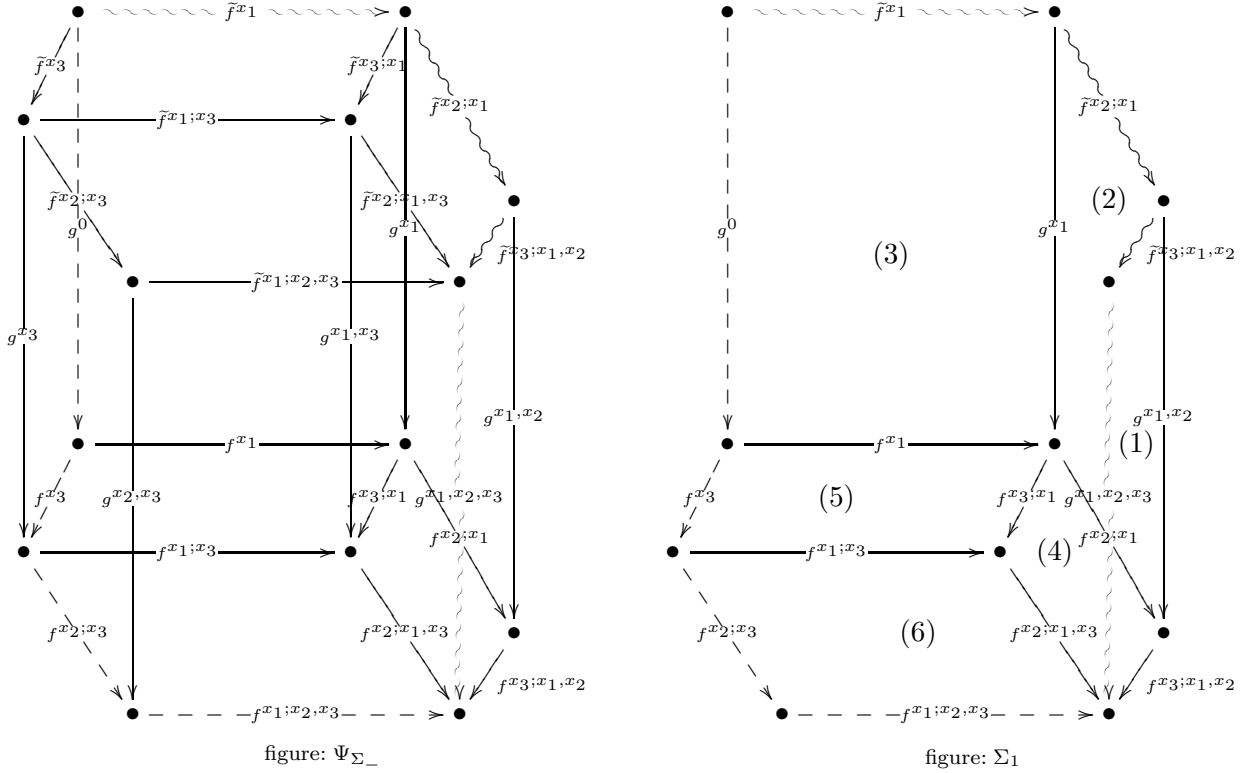
\begin{equation*}
     \xy 0;/r.17pc/:
(0,0 )*+{\bullet}="1";
(60, 0)*+{\bullet}="2";
( 70,15)*+{\bullet}="3";
(-20,30)*+{\bullet}="4";
(40,30 )*+{\bullet}="5";
(-10, 50)*+{\bullet}="6";
(  50,50)*+{\bullet}="7";
(0,80 )*+{\bullet}="8";
(60,80)*+{\bullet}="9";
( 70,95)*+{\bullet}="10";
(-20,110)*+{\bullet}="11";
(40,110 )*+{\bullet}="12";
(-10, 130)*+{\bullet}="13";
(  50,130)*+{\bullet}="14";
 {\ar@{-->}|-{\scriptscriptstyle f^{x_1;x_2, x_3}} "1";"2" };
{\ar@{-->}|-{\scriptscriptstyle f^{ x_2; x_3}} "4";"1" };
{\ar@{->} "3";"2"^{\scriptscriptstyle f^{ x_3;x_1,x_2}} };
{\ar@{->}|-{\scriptscriptstyle f^{x_2; x_1 ,x_3}} "5";"2" };
{\ar@{->}|-{\scriptscriptstyle f^{ x_1;x_3}} "4";"5" };
{\ar@{-->}|-{\scriptscriptstyle f^{x_3 }} "6";"4" };
{\ar@{->}|-{\scriptscriptstyle f^{x_1}} "6";"7" };
{\ar@{->}|-{\scriptscriptstyle f^{x_2;x_1  }} "7";"3" };
{\ar@{->}|-{\scriptscriptstyle f^{x_3;x_1 }} "7";"5" };
{\ar@{->}|-{\scriptscriptstyle\widetilde{f}^{x_1;x_2, x_3}} "8";"9" };
{\ar@{->}|-{\scriptscriptstyle\widetilde{f}^{x_2;x_3}} "11";"8" };
{\ar@{~>} "10";"9"^{\scriptscriptstyle\widetilde{f}^{ x_3;x_1,x_2}} };
{\ar@{->}|-{\scriptscriptstyle\widetilde{f}^{x_2; x_1 ,x_3}} "12";"9" };
{\ar@{->}|-{\scriptscriptstyle\widetilde{f}^{x_1;x_3}} "11";"12" };
{\ar@{->}|-{\scriptscriptstyle\widetilde{f}^{x_3 }} "13";"11" };
{\ar@{~>}|-{\scriptscriptstyle\widetilde{f}^{x_1}} "13";"14" };
{\ar@{~>}|-{\scriptscriptstyle\widetilde{f}^{x_2;x_1  }} "14";"10" };
{\ar@{->}|-{\scriptscriptstyle\widetilde{f}^{x_3;x_1 }} "14";"12" };
 {\ar@{<- }|-{\scriptscriptstyle g^{x_2, x_3}} "1";"8" };
 {\ar@{<~ }|-{\scriptscriptstyle g^{ x_1,x_2,x_3 }} "2";"9" };
 {\ar@{<- }|-{\scriptscriptstyle g^{ x_1,x_2 }} "3";"10" };
 {\ar@{<-}|-{\scriptscriptstyle g^{ x_3  }} "4";"11" };{\ar@{<- }|-{\scriptscriptstyle g^{ x_1,x_3  }} "5";"12" };
 {\ar@{<--}|-{\scriptscriptstyle g^{0  }} "6";"13" };{\ar@{<- }|-{\scriptscriptstyle g^{ x_1 }} "7";"14" };( 35,-8)*+{ \scriptstyle  {\rm figure}:\;\Psi_{\Sigma_-} }="30";
  \endxy \qquad  \xy 0;/r.17pc/:
(0,0 )*+{\bullet}="1";
(60, 0)*+{\bullet}="2";
( 70,15)*+{\bullet}="3";
(-20,30)*+{\bullet}="4";
(40,30 )*+{\bullet}="5";
(-10, 50)*+{\bullet}="6";
(  50,50)*+{\bullet}="7";
(60,80)*+{\bullet}="9";
( 70,95)*+{\bullet}="10";
(-10, 130)*+{\bullet}="13";
(  50,130)*+{\bullet}="14";
( 65,50)*+{(1)}="21";
( 60,95)*+{(2)}="22";
( 20,85)*+{(3)}="23";
( 50,30)*+{(4)}="24";
( 10,40)*+{(5)}="24";
( 25,15)*+{(6)}="26";
  {\ar@{-->}|-{\scriptscriptstyle f^{x_1;x_2, x_3}} "1";"2" };
{\ar@{-->}|-{\scriptscriptstyle f^{ x_2; x_3}} "4";"1" };
{\ar@{->} "3";"2"^{\scriptscriptstyle f^{ x_3;x_1,x_2}} };
{\ar@{->}|-{\scriptscriptstyle f^{x_2; x_1 ,x_3}} "5";"2" };
{\ar@{->}|-{\scriptscriptstyle f^{ x_1;x_3}} "4";"5" };
{\ar@{-->}|-{\scriptscriptstyle f^{x_3 }} "6";"4" };
{\ar@{->}|-{\scriptscriptstyle f^{x_1}} "6";"7" };
{\ar@{->}|-{\scriptscriptstyle f^{x_2;x_1  }} "7";"3" };
{\ar@{->}|-{\scriptscriptstyle f^{x_3;x_1 }} "7";"5" };
{\ar@{~>} "10";"9"^{\scriptscriptstyle\widetilde{f}^{ x_3;x_1,x_2}} };
{\ar@{~>}|-{\scriptscriptstyle\widetilde{f}^{x_1}} "13";"14" };
{\ar@{~>}|-{\scriptscriptstyle\widetilde{f}^{x_2;x_1  }} "14";"10" };
 {\ar@{<~ }|-{\scriptscriptstyle g^{ x_1,x_2,x_3 }} "2";"9" };
 {\ar@{<- }|-{\scriptscriptstyle g^{ x_1,x_2 }} "3";"10" };
 {\ar@{<--}|-{\scriptscriptstyle g^{0  }} "6";"13" };{\ar@{<- }|-{\scriptscriptstyle g^{ x_1 }} "7";"14" };( 35,-8)*+{ \scriptstyle  {\rm figure}:\; \Sigma_1  }="30";
  \endxy
\end{equation*} See also the figure $\Psi_{\Sigma_+}$ later.
Write $3$-arrow
$\Psi_{\Sigma_{x_i,x_j;x_k}}=
 (*,*  , k^{x_i,x_j;x_k} )
$
  in $\mathcal{G}^{\mathscr L}$.

Since the $2$-arrow $\Sigma_-$ is the composition  of  three $2$-arrows in (\ref{eq:Sigma+-}),
by the functoriality (4) of    the lax-natural transformation, the $3$-arrow $\Psi_{\Sigma_-}$ is the whiskered
composition of $\Psi_{\gamma^{x_1}\#_0\Sigma_{x_2,x_3 ; x_1}}$, $
  \Psi_{\Sigma_{x_1,x_3  } \#_0 \gamma^{x_2;x_1,x_3}}$ and $
  \Psi_{\gamma^{x_3}\#_0\Sigma_{x_1,x_2 ; x_3 }  }
$. Let us write it down explicitly as composition of $3$-arrows $(g_*,\Sigma_1,K_1),\ldots, (g_*,\Sigma_6,K_6)$ in $\mathcal{G}^{\mathscr L}$, where
\begin{equation}\label{eq:g}
     g_*=\widetilde{f}^{x_1} \widetilde{f}^{x_2;x_1  } \widetilde{f}^{ x_3;x_1,x_2} g^{ x_1,x_2 ,x_3},
\end{equation}
corresponding to the wavy path in the above figures.

The first $2$-arrow is $(g_*,\Sigma_1)$ is the composition of the whiskered $2$-arrows  (1), (2), (3), (4), (5) and (6) in the figure $\Sigma_1$, where $(1)=( *,h^{ x_3; x_1 ,x_2} ),(2)= ( *,h^{x_2; x_1 }), (3)=( *,  h^{ x_1 }),(4)= (g_*,F^{x_2 ,x_3;x_1 }) , (5)=( *, F^{x_1 ,x_3 }), (6)=( *, F^{x_1 , x_2;x_3 }) $:
  \begin{equation*}
     \Sigma_1:=\widetilde{f}^{x_1}\rhd\left[  \widetilde{f}^{x_2;x_1  } \rhd h^{x_3; x_1 ,x_2 } \cdot   h^{x_2; x_1 }\right]\cdot   h^{ x_1 }\cdot g^0   \rhd \left[  {f}^{x_1}  \rhd F^{ x_2,x_3;x_1 } \cdot F^{x_1 ,x_3 }\cdot {f}^{x_3}  \rhd F^{x_1 , x_2;x_3 } \right ].\end{equation*}

     The first   $3$-arrow is $(g_*,\Sigma_1,K_1)$,  interchanging $2$-arrows $(3)$ and $ (4)$ in the figure $\Sigma_1$, with
     \begin{equation}\label{eq:K-1}
K_1=K_1^0\rhd'\left[ h^{ x_1}\rhd' \left\{ (h^{ x_1})^{-1},\left(\widetilde{f}^{x_1} g^{x_1}\right )\rhd F^{x_2 ,x_3;x_1 }\right\}^{-1} \right] ,\end{equation} whose $2$-target is the $2$-arrow $(g_*,\Sigma_2)$, where    $K_1^0=\widetilde{f}^{x_1}\rhd \left[  \widetilde{f}^{x_2;x_1  } \rhd h^{x_3; x_1 ,x_2 } \cdot   h^{x_2; x_1 }\right]$ is the whiskering corresponding to  the composition of the whiskered (1) and (2) by definition (\ref{eq:whiskering}) of the $2$-whiskering.
       The   interchanging   $3$-arrow has this form by the definition (\ref{eq:interchanging}). We will not write $ \Sigma_j$'s down explicitly, since we will not use them.

The   second  $3$-arrow is $(g_*,\Sigma_2,K_2)$ with  $K_2 =\widetilde{f}^{x_1}\rhd k^{x_2 ,x_3 ;x_1}$, whose $2$-target is $(g_*,\Sigma_3)$.

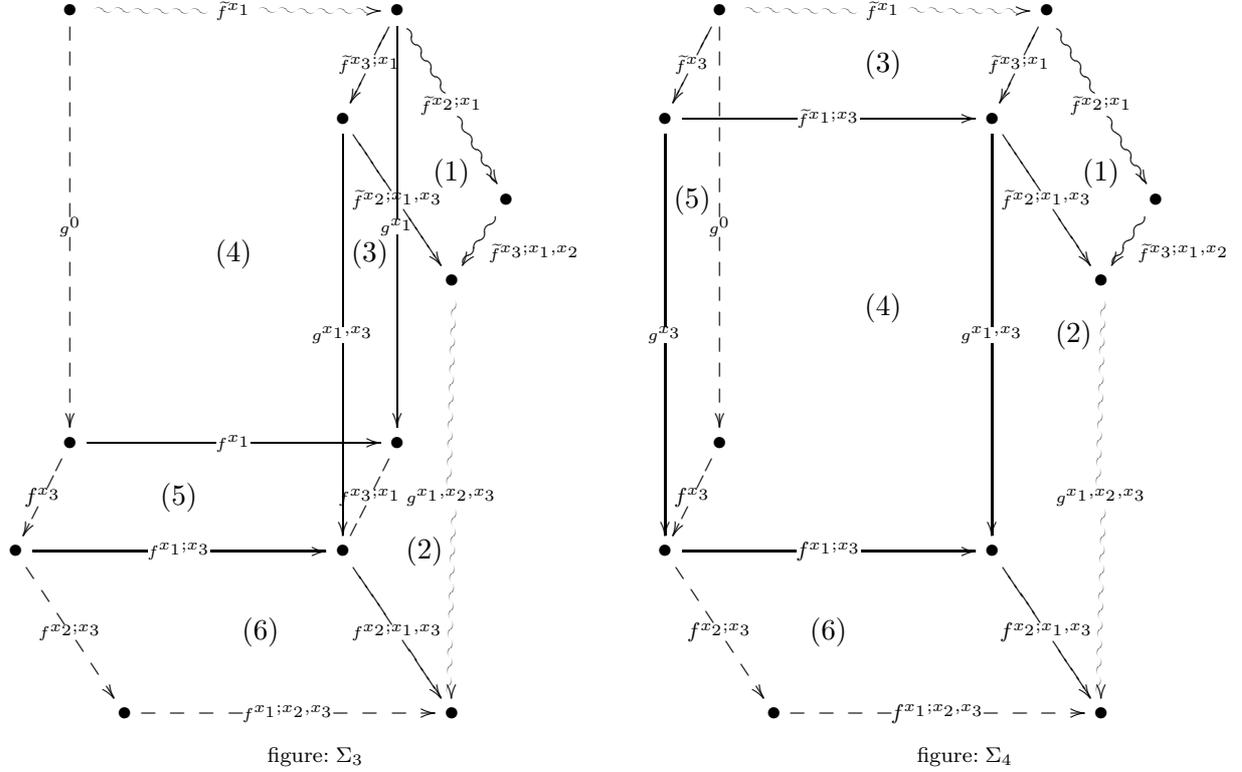
\begin{equation*}
      \xy 0;/r.17pc/:
(0,0 )*+{\bullet}="1";
(60, 0)*+{\bullet}="2";
(-20,30)*+{\bullet}="4";
(40,30 )*+{\bullet}="5";
(-10, 50)*+{\bullet}="6";
(  50,50)*+{\bullet}="7";
(60,80)*+{\bullet}="9";
( 70,95)*+{\bullet}="10";
(40,110 )*+{\bullet}="12";
(-10, 130)*+{\bullet}="13";
(  50,130)*+{\bullet}="14";
( 45,85)*+{(3)}="21";
( 60,100)*+{(1)}="22";
( 20,85)*+{(4)}="23";
( 55,30)*+{(2)}="24";
( 10,40)*+{(5)}="24";
( 25,15)*+{(6)}="26";{\ar@{-->}|-{f^{x_3 }} "6";"4" };{\ar@{->}|-{\scriptscriptstyle f^{x_1}} "6";"7" };
  {\ar@{-->}|-{\scriptscriptstyle f^{x_1;x_2,  x_3}} "1";"2" };
{\ar@{-->}|-{\scriptscriptstyle f^{ x_2; x_3}} "4";"1" };
{\ar@{->}|-{\scriptscriptstyle f^{x_2; x_1 ,x_3}} "5";"2" };
{\ar@{->}|-{\scriptscriptstyle f^{ x_1;x_3}} "4";"5" };
{\ar@{--}|-{\scriptscriptstyle f^{x_3;x_1 }} "7";"5" };
{\ar@{~>} "10";"9"^{\scriptscriptstyle\widetilde{f}^{ x_3;x_1,x_2}} };
{\ar@{->}|-{\scriptscriptstyle\widetilde{f}^{x_2; x_1 ,x_3}} "12";"9" };
{\ar@{~>}|-{\scriptscriptstyle\widetilde{f}^{x_1}} "13";"14" };
{\ar@{~>}|-{\scriptscriptstyle\widetilde{f}^{x_2;x_1  }} "14";"10" };
{\ar@{->}|-{\scriptscriptstyle\widetilde{f}^{x_3;x_1 }} "14";"12" };
 {\ar@{<~ }|-{\scriptscriptstyle g^{ x_1,x_2,x_3 }} "2";"9" };
{\ar@{<- }|-{\scriptscriptstyle g^{ x_1,x_3  }} "5";"12" };
 {\ar@{<--}|-{\scriptscriptstyle g^{0  }} "6";"13" };
 {\ar@{<- }|-{\scriptscriptstyle g^{ x_1 }} "7";"14" };( 35,-8)*+{ \scriptstyle  {\rm figure}:\; \Sigma_3  }="30";
  \endxy
\qquad
      \xy 0;/r.17pc/:
(0,0 )*+{\bullet}="1";
(60, 0)*+{\bullet}="2";
(-20,30)*+{\bullet}="4";
(40,30 )*+{\bullet}="5";
(-10, 50)*+{\bullet}="6";
 (60,80)*+{\bullet}="9";
( 70,95)*+{\bullet}="10";
(-20,110)*+{\bullet}="11";
(40,110 )*+{\bullet}="12";
(-10, 130)*+{\bullet}="13";
(  50,130)*+{\bullet}="14";
( 20,120)*+{(3)}="21";
( 60,100)*+{(1)}="22";
( 20,75)*+{(4)}="23";
( 55,70)*+{(2)}="24";
( 10,15)*+{(6)}="24";
( -15,95)*+{(5)}="26";
 {\ar@{-->}|-{f^{x_1;x_2, x_3}} "1";"2" };
{\ar@{-->}|-{f^{ x_2; x_3}} "4";"1" };
{\ar@{->}|-{f^{x_2; x_1 ,x_3}} "5";"2" };
{\ar@{->}|-{f^{ x_1;x_3}} "4";"5" };
{\ar@{-->}|-{f^{x_3 }} "6";"4" };
 {\ar@{~>} "10";"9"^{\scriptscriptstyle\widetilde{f}^{ x_3;x_1,x_2}} };
{\ar@{->}|-{\scriptscriptstyle\widetilde{f}^{x_2; x_1 ,x_3}} "12";"9" };
{\ar@{->}|-{\scriptscriptstyle\widetilde{f}^{x_1;x_3}} "11";"12" };
{\ar@{->}|-{\scriptscriptstyle\widetilde{f}^{x_3 }} "13";"11" };
{\ar@{~>}|-{\scriptscriptstyle\widetilde{f}^{x_1}} "13";"14" };
{\ar@{~>}|-{\scriptscriptstyle\widetilde{f}^{x_2;x_1  }} "14";"10" };
{\ar@{->}|-{\scriptscriptstyle\widetilde{f}^{x_3;x_1 }} "14";"12" };
 {\ar@{<~}|-{\scriptscriptstyle g^{ x_1,x_2,x_3 }} "2";"9" };
 {\ar@{<-}|-{\scriptscriptstyle g^{ x_3  }} "4";"11" };{\ar@{<- }|-{\scriptscriptstyle g^{ x_1,x_3  }} "5";"12" };
 {\ar@{<--}|-{\scriptscriptstyle g^{0  }} "6";"13" };( 35,-8)*+{ \scriptstyle  {\rm figure}:\; \Sigma_4  }="30";
  \endxy
\end{equation*} The  third  $3$-arrow is $(g_*,\Sigma_3,K_3)$ with
\begin{equation*}
    K_3= \left[ \widetilde{f}^{x_1} \rhd \left(\widetilde{F}^{x_2 ,x_3;x_1 } \cdot  \widetilde{f}^{x_3;x_1 }  \rhd h^{x_2; x_1 ,x_3}\right)\right]\rhd' k^{x_1,x_3}.
\end{equation*}
 whose $2$-target is $(g_*,\Sigma_4)$, where the part before $ \rhd'$ is the whiskering corresponding to  the composition of $2$-arrows (1) and (2) in the figure $\Sigma_3$.
The  fourth  $3$-arrow is $(g_*,\Sigma_4,K_4)$, interchanging $2$-arrows $(3)$ and $ (2)$ in the figure $\Sigma_4$, with
 \begin{equation*}
    K_4=\left( \widetilde{f}^{x_1} \rhd \widetilde{F}^{x_2 ,x_3;x_1 }\right) \rhd' \left[\widetilde{F}^{  x_1,x_3 }\rhd'\left\{\left(\widetilde{F}^{  x_1,x_3 }\right)^{-1} , \left(\widetilde{f}^{x_1}\widetilde{f}^{x_3;x_1 }\right) \rhd h^{x_2; x_1 ,x_3}\right\} \right],
\end{equation*}  whose $2$-target is $(g_*,\Sigma_5)$.

The  fifth  $3$-arrow is $(g_*,\Sigma_5,K_5)$, interchanging $2$-arrows $(5)$ and $ (6)$ in the figure $\Sigma_4$, with
 \begin{equation*}
    K_5=K_5^0
    \rhd'\left[h^{  x_3 }\rhd' \left\{(h^{  x_3 }  )^{-1} ,\left(\widetilde{f}^{x_3}  g^{x_3}  \right)\rhd F^{x_1 , x_2;x_3 }  \right\}^{-1}\right] .
  \end{equation*} whose $2$-target is $(g_*,\Sigma_6)$, where  $ K_5^0=\widetilde{f}^{x_1} \rhd \widetilde{F}^{x_2 ,x_3;x_1 }\cdot \widetilde{F}^{x_1 , x_3 } \cdot \widetilde{f}^{x_3}\rhd \left[\widetilde{f}^{x_1;x_3 }  \rhd  h^{x_2; x_1 ,x_3}\cdot   h^{x_1;  x_3} \right]$ is the whiskering corresponding to   the composition of $2$-arrows (1), (3), (2) and (4) in the figure $\Sigma_4$.

The sixth    $3$-arrow is $(g_*,\Sigma_6,K_6)$ with
\begin{equation*}
    K_6=\left [\widetilde{f}^{x_1} \rhd \widetilde{F}^{x_2 ,x_3;x_1 } \cdot \widetilde{F}^{x_1 , x_3 }\right ]\rhd' \left(\widetilde{f}^{x_3}  \rhd k^{x_1,x_2;x_3}\right).
\end{equation*}
whose $2$-target is $(g_*,\Sigma_7)$.
The last $3$-arrow is $(g_*,\Sigma_7,K_7)=\widetilde{F}_3(\Omega_{x_1,x_2,x_3})$ with
$
  K_7=\widetilde{k}^{x_1,x_2,x_3}
$, whose $2$-target is $(g_*,\Sigma_8)$.
Now $\Psi_{\Sigma_-}$ in (\ref{eq:naturality}) is $(*,*,K_1K_2\cdots K_6)$ by the functoriality of $\Psi$ in the definition in \S 4.3, while the RHS of (\ref{eq:naturality}) is $(*,*,K_1K_2\cdots K_7)$.

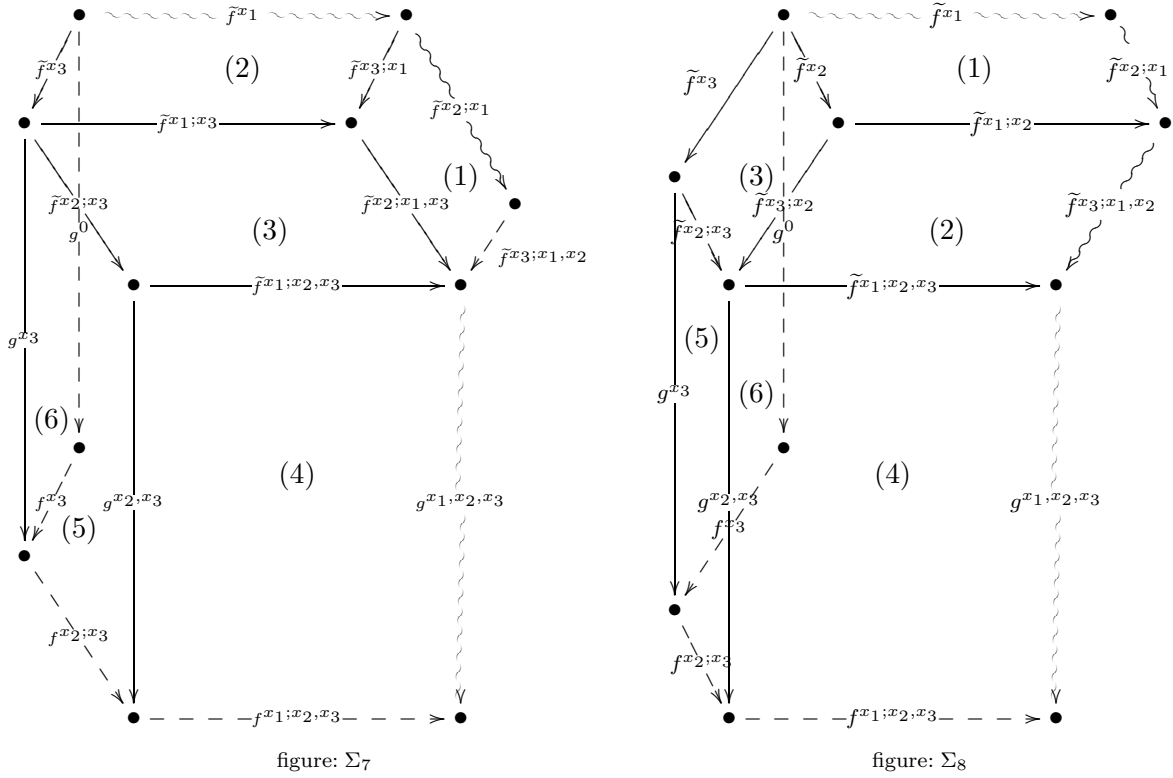
\begin{equation*}
     \xy 0;/r.17pc/:
(0,0 )*+{\bullet}="1";
(60, 0)*+{\bullet}="2";
(-20,30)*+{\bullet}="4";
(-10, 50)*+{\bullet}="6";
(0,80 )*+{\bullet}="8";
(60,80)*+{\bullet}="9";
( 70,95)*+{\bullet}="10";
(-20,110)*+{\bullet}="11";
(40,110 )*+{\bullet}="12";
(-10, 130)*+{\bullet}="13";
(  50,130)*+{\bullet}="14";
( 20,120)*+{(2)}="21";
( 60,100)*+{(1)}="22";
( -10,35)*+{(5)}="23";
( 25,90)*+{(3)}="24";
( 30,45)*+{(4)}="24";
( -15,55)*+{(6)}="26";
 {\ar@{-->}|-{\scriptscriptstyle f^{x_1; x_2 ,  x_3}} "1";"2" };
{\ar@{-->}|-{\scriptscriptstyle f^{ x_2; x_3}} "4";"1" };
{\ar@{-->}|-{\scriptscriptstyle f^{x_3 }} "6";"4" };
{\ar@{->}|-{\scriptscriptstyle \widetilde{f}^{x_1;x_2,  x_3}} "8";"9" };
{\ar@{->}|-{\scriptscriptstyle \widetilde{f}^{x_2;x_3}} "11";"8" };
{\ar@{-->} "10";"9"^{\scriptscriptstyle \widetilde{f}^{ x_3;x_1,x_2}} };
{\ar@{->}|-{\scriptscriptstyle \widetilde{f}^{x_2; x_1 ,x_3}} "12";"9" };
{\ar@{->}|-{\scriptscriptstyle \widetilde{f}^{x_1;x_3}} "11";"12" };
{\ar@{->}|-{\scriptscriptstyle \widetilde{f}^{x_3 }} "13";"11" };
{\ar@{~>}|-{\scriptscriptstyle \widetilde{f}^{x_1}} "13";"14" };
{\ar@{~>}|-{\scriptscriptstyle \widetilde{f}^{x_2;x_1  }} "14";"10" };
{\ar@{->}|-{\scriptscriptstyle \widetilde{f}^{x_3;x_1 }} "14";"12" };
 {\ar@{<- }|-{\scriptscriptstyle g^{x_2 ,x_3}} "1";"8" };
 {\ar@{<~ }|-{\scriptscriptstyle g^{ x_1,x_2,x_3 }} "2";"9" };
 {\ar@{<-}|-{\scriptscriptstyle g^{ x_3  }} "4";"11" };
 {\ar@{<--}|-{\scriptscriptstyle g^{0  }} "6";"13" };( 35,-8)*+{ \scriptstyle  {\rm figure}:\; \Sigma_7  }="30";
  \endxy
\qquad
     \xy 0;/r.17pc/:
(0,0 )*+{\bullet}="1";
(60, 0)*+{\bullet}="2";
(-10,20)*+{\bullet}="3";
( 10, 50)*+{\bullet}="6";
(0,80 )*+{\bullet}="8";
(60, 80)*+{\bullet}="9";
(-10,100)*+{\bullet}="10";
(20,110)*+{\bullet}="11";
(80,110 )*+{\bullet}="12";
( 10, 130)*+{\bullet}="13";
(  70,130)*+{\bullet}="14";
( 40,90)*+{(2)}="21";( 45,120)*+{(1)}="22";
(  5,60)*+{(6)}="23";
(  5,100)*+{(3)}="24";
( 30,45)*+{(4)}="24";
(  -5,70)*+{(5)}="26";
 {\ar@{-->}|-{ {f}^{x_1;x_2,  x_3}} "1";"2" };
{\ar@{-->}|-{ {f}^{ x_2; x_3}} "3";"1" };
{\ar@{<--}|-{ {f}^{ x_3 }} "3";"6" };
 {\ar@{->}|-{\widetilde f^{x_1;x_2, x_3}} "8";"9" };
{\ar@{->}|-{\widetilde f^{ x_2; x_3}} "10";"8" };
{\ar@{<-} "10";"13"^{\widetilde f^{ x_3 }} };
{\ar@{~>}|-{\widetilde f^{x_3; x_1 ,x_2}}"12";"9"};
{\ar@{->}|-{\widetilde f^{ x_1;x_2}} "11";"12" };
{\ar@{->}|-{\widetilde f^{x_2 }} "13";"11" };
{\ar@{->}|-{\widetilde f^{x_3;x_2}} "11";"8" };
{\ar@{~>}|-{\widetilde f^{ x_1  }}"13"; "14"};
{\ar@{~>}|-{\widetilde f^{x_2;x_1 }}"14";"12"};
{\ar@{<- }|-{g^{x_2, x_3 }} "1";"8" };
 {\ar@{<~ }|-{g^{ x_1,x_2,x_3 }} "2";"9" };
 {\ar@{<- }|-{g^{ x_3  }} "3";"10" };
 {\ar@{<--}|-{g^{0  }} "6";"13" };( 35,-8)*+{ \scriptstyle  {\rm figure}:\;  \Sigma_8 }="30";
\endxy
\end{equation*}

To calculate $\frac {\partial^3K_1K_2\cdots K_7}{\partial x_1\partial x_2\partial x_3}$ at $(0,0,0)$, note that if $\frac {\partial^2 m }{\partial x_i\partial x_j }
 (0,0,0) =\frac {\partial  m }{\partial x_i  }
 (0,0,0)=0,m(0,0,0)=1_L$, we have
\begin{equation}\label{eq:cancel}
    \left. \frac {\partial^3 }{\partial x_1\partial x_2\partial x_3}\left\{ n(x_1, x_2, x_3)\rhd'
     m(x_1, x_2, x_3)\right\}\right|_{(0,0,0)}= n(0,0,0)\rhd'\frac {\partial^3 m }{\partial x_1\partial x_2\partial x_3}
 (0,0,0),
\end{equation} by using Taylor's expansion for  $m(x_1, x_2, x_3)$,     since $n \rhd' 1_L=1_L$ for any $n\in H$ by  (\ref{eq:brackt-1}), where $1_L$ is the identity of $L$. A similar identity holds for $\rhd$. It is direct to check that  $\frac {\partial    K_j }{\partial x_i  }(0,0,0)=0$  and $     K_j  (0,0,0)=1_L$ for $j=1,\ldots,7$, $i=1,2,3$. So we have
 \begin{equation*}
     \frac {\partial^3K_1K_2\cdots K_7}{\partial x_1\partial x_2\partial x_3}(0,0,0)=\sum_{j=1}^7\frac {\partial^3K_j}{\partial x_1\partial x_2\partial x_3}(0,0,0).
 \end{equation*}
 By  (\ref{eq:brackt-1}), (\ref{eq:value-0}) and (\ref{eq:cancel}), we have
 \begin{equation*}
      \frac {\partial^3K_1 }{\partial x_1\partial x_2\partial x_3}(0,0,0)= \left\{\frac {\partial   h^{ x_1} }{\partial x_1 } , g^0\rhd \frac {\partial^2  F^{x_2 ,x_3;x_1 }}{ \partial x_2\partial x_3}\right\} (0,0,0)= \{\varphi(v_1),g^0\rhd B(v_2,v_3)\},
 \end{equation*}
 for $K_1$ in  (\ref{eq:K-1}), and similarly for the other $K_j$.
 Consequently, the derivative  $\frac {\partial^3K_1K_2\cdots K_7}{\partial x_1\partial x_2\partial x_3}$ at $(0,0,0)$ gives
\begin{equation}\label{eq:derivative-R}\begin{split}&
 \{\varphi(v_1),g^0\rhd B(v_2,v_3)\}+    v_1\psi(v_2,v_3)+\widetilde{A}(v_1)\rhd \psi(v_2,v_3)+\varphi(v_2)\rhd'\psi(v_1,v_3)\\&-\{ \widetilde{B}(v_1,v_3), \varphi(v_2)\}+\{\varphi(v_3), g^0\rhd B(v_1,v_2)\}+ v_3\psi(v_1,v_2)+\widetilde{A}(v_3)\rhd \psi(v_1,v_2)+ \widetilde{ C}(v_1,v_2,v_3).
\end{split}\end{equation}
 \subsection{The  gauge transformations of the $C$ field: the $\Psi_{\Sigma_+}$ part}
Now consider $\Psi_{\Sigma_+}$ in (\ref{eq:naturality}).
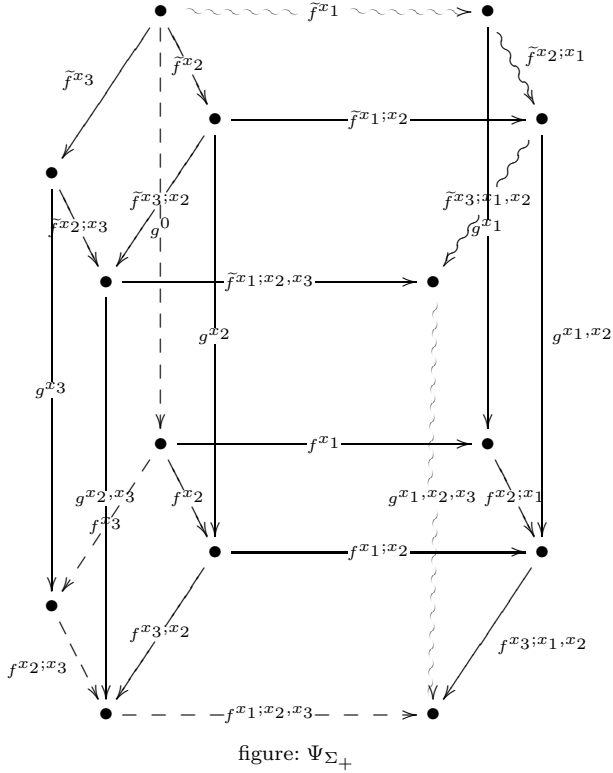
\begin{equation*}
     \xy 0;/r.17pc/:
(0,0 )*+{\bullet}="1";
(60, 0)*+{\bullet}="2";
(-10,20)*+{\bullet}="3";
(20,30)*+{\bullet}="4";
(80,30 )*+{\bullet}="5";
( 10, 50)*+{\bullet}="6";
(  70,50)*+{\bullet}="7";
(0,80 )*+{\bullet}="8";
(60, 80)*+{\bullet}="9";
(-10,100)*+{\bullet}="10";
(20,110)*+{\bullet}="11";
(80,110 )*+{\bullet}="12";
( 10, 130)*+{\bullet}="13";
(  70,130)*+{\bullet}="14";
 {\ar@{-->}|-{\scriptscriptstyle  {f}^{x_1;x_2 , x_3 }} "1";"2" };
{\ar@{-->} "3";"1"_{\scriptscriptstyle  {f}^{ x_2; x_3}} };
{\ar@{<--}|-{\scriptscriptstyle  {f}^{ x_3 }} "3";"6" };
{\ar@{->}"5";"2"^{\scriptscriptstyle f^{x_3; x_1 ,x_2}}};
{\ar@{->}|-{\scriptscriptstyle f^{ x_1;x_2}} "4";"5" };
{\ar@{->}|-{\scriptscriptstyle f^{x_2 }} "6";"4" };
{\ar@{->}|-{\scriptscriptstyle f^{x_3;x_2}} "4";"1" };
{\ar@{->}|-{\scriptscriptstyle f^{ x_1  }}"6"; "7"};
{\ar@{->}|-{\scriptscriptstyle f^{x_2;x_1 }}"7";"5"};
 {\ar@{->}|-{\scriptscriptstyle \widetilde f^{x_1;x_2, x_3 }} "8";"9" };
{\ar@{->}|-{\scriptscriptstyle \widetilde f^{ x_2; x_3}} "10";"8" };
{\ar@{<-} "10";"13"^{\scriptscriptstyle \widetilde f^{ x_3 }} };
{\ar@{~>}|-{\scriptscriptstyle \widetilde f^{x_3; x_1 ,x_2}}"12";"9"};
{\ar@{->}|-{\scriptscriptstyle \widetilde f^{ x_1;x_2}} "11";"12" };
{\ar@{->}|-{\scriptscriptstyle \widetilde f^{x_2 }} "13";"11" };
{\ar@{->}|-{\scriptscriptstyle \widetilde f^{x_3;x_2}} "11";"8" };
{\ar@{~>}|-{\scriptscriptstyle \widetilde f^{ x_1  }}"13"; "14"};
{\ar@{~>}"14";"12"^{\scriptscriptstyle \widetilde f^{x_2;x_1 }}};
{\ar@{<- }|-{\scriptscriptstyle g^{x_2, x_3 }} "1";"8" };
 {\ar@{<~ }|-{\scriptscriptstyle g^{ x_1,x_2,x_3 }} "2";"9" };
 {\ar@{<- }|-{\scriptscriptstyle g^{ x_3  }} "3";"10" };
 {\ar@{<-}|-{\scriptscriptstyle g^{ x_2  }} "4";"11" };{\ar@{<- }"5";"12"_{\scriptscriptstyle g^{ x_1,x_2  }} };
 {\ar@{<--}|-{\scriptscriptstyle g^{0  }} "6";"13" };{\ar@{<- }|-{\scriptscriptstyle g^{ x_1 }} "7";"14" };( 35,-8)*+{ \scriptstyle  {\rm figure}:\; \Psi_{\Sigma_+} }="30";
\endxy \qquad   
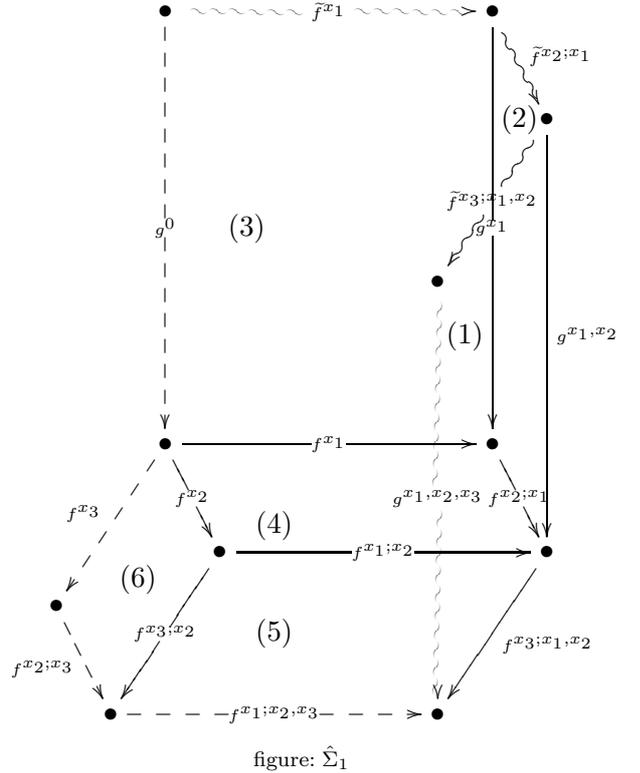
\xy 0;/r.17pc/:
(0,0 )*+{\bullet}="1";
(60, 0)*+{\bullet}="2";
(-10,20)*+{\bullet}="3";
(20,30)*+{\bullet}="4";
(80,30 )*+{\bullet}="5";
( 10, 50)*+{\bullet}="6";
(  70,50)*+{\bullet}="7";
(60, 80)*+{\bullet}="9";
(80,110 )*+{\bullet}="12";
( 10, 130)*+{\bullet}="13";
(  70,130)*+{\bullet}="14";
( 75,110)*+{(2)}="21";
( 65,70)*+{(1)}="22";
( 30,15)*+{(5)}="23";
( 25,90)*+{(3)}="24";
( 30,35)*+{(4)}="24";
(  5,25)*+{(6)}="26";
 {\ar@{-->}|-{\scriptscriptstyle  {f}^{x_1;x_2, x_3 }} "1";"2" };
{\ar@{-->} "3";"1"_{\scriptscriptstyle  {f}^{ x_2; x_3}} };
{\ar@{<--} "3";"6"^{\scriptscriptstyle  {f}^{ x_3 }} };
{\ar@{->}"5";"2"^{\scriptscriptstyle f^{x_3; x_1 ,x_2}}};
{\ar@{->}|-{\scriptscriptstyle f^{ x_1;x_2}} "4";"5" };
{\ar@{->}|-{\scriptscriptstyle f^{x_2 }} "6";"4" };
{\ar@{->}|-{\scriptscriptstyle f^{x_3;x_2}} "4";"1" };
{\ar@{->}|-{\scriptscriptstyle f^{ x_1  }}"6"; "7"};
{\ar@{->}|-{\scriptscriptstyle f^{x_2;x_1 }}"7";"5"};
 {\ar@{~>}|-{\scriptscriptstyle \widetilde f^{x_3; x_1 ,x_2}}"12";"9"};
{\ar@{~>}|-{\scriptscriptstyle \widetilde f^{ x_1  }}"13"; "14"};
{\ar@{~>}"14";"12"^{\scriptscriptstyle \widetilde f^{x_2;x_1 }}};
 {\ar@{<~ }|-{\scriptscriptstyle g^{ x_1,x_2,x_3 }} "2";"9" };
 {\ar@{<- }"5";"12"_{\scriptscriptstyle g^{ x_1,x_2  }} };
 {\ar@{<--}|-{\scriptscriptstyle g^{0  }} "6";"13" };{\ar@{<- }|-{\scriptscriptstyle g^{ x_1 }} "7";"14" };( 35,-8)*+{ \scriptstyle  {\rm figure}:\;\hat \Sigma_1 }="30";
\endxy
\end{equation*}

The first $3$-arrow is $(g_*,\hat\Sigma_1,\hat K_1)$ with
 \begin{equation*}\hat
K_1=\left[\left(\widetilde{f}^{x_1} \widetilde{f}^{x_2;x_1  }\right)\rhd h^{ x_3; x_1 ,x_2}\right]\rhd' k^{x_1,x_2}.\end{equation*}
whose   $2$-target is $(g_*,\hat \Sigma_2)$.  The second $3$-arrow is $(g_*,\hat\Sigma_2,\hat K_2)$, interchanging $2$-arrows  $(2)$ and $ (1)$ in the figure $ \hat\Sigma_2$, with
  \begin{equation*}\hat
K_2=  \widetilde{F}^{x_1,x_2 }\rhd'\left\{ \left(\widetilde{F}^{x_1,x_2 }\right)^{-1},(\widetilde{f}^{x_1} \widetilde{f}^{x_2;x_1  })\rhd h^{ x_3; x_1 ,x_2} \right\}  ,\end{equation*}
 whose   $2$-target is $(g_*,\hat \Sigma_3)$.
     The third $3$-arrow is $(g_*,\hat\Sigma_3,\hat K_3)$, interchanging $2$-arrows  $(4)$ and $ (5)$ in the figure $ \hat\Sigma_2$, with
  \begin{equation*}\begin{split}\hat
K_3=&\left[ \widetilde{F}^{x_1,x_2 }\cdot   \widetilde{f}^{x_2}\rhd\left( \widetilde{f}^{ x_1  ;x_2} \rhd h^{ x_3; x_1 ,x_2} \cdot     h^{  x_1; x_2}\right)\right]  \rhd'\left[h^{  x_2 }\rhd'  \left\{ (h^{  x_2 })^{-1}, \left ( \widetilde{{f}}^{x_2} g^{x_2}\right) \rhd F^{x_1 ,x_3;x_2 }\right \}^{-1}\right] ,\end{split}\end{equation*}
 whose   $2$-target is $(g_*,\hat \Sigma_4)$, where the first bracket is the whiskering corresponding to the composition of $2$-arrows (2), (1) and (3) in the figure $\hat\Sigma_2$.

\begin{equation*}
    \xy 0;/r.17pc/:
(0,0 )*+{\bullet}="1";
(60, 0)*+{\bullet}="2";
(-10,20)*+{\bullet}="3";
(20,30)*+{\bullet}="4";
(80,30 )*+{\bullet}="5";
( 10, 50)*+{\bullet}="6";
(60, 80)*+{\bullet}="9";
(20,100)*+{\bullet}="11";
(80,100 )*+{\bullet}="12";
( 10, 120)*+{\bullet}="13";
(  70,120)*+{\bullet}="14";
( 45,110)*+{(2)}="21";
( 69,75)*+{(1)}="22";
( 30,15)*+{(5)}="23";
( 15,80)*+{(4)}="24";
( 50,65)*+{(3)}="24";
(  5,25)*+{(6)}="26";
 {\ar@{-->}|-{\scriptscriptstyle  {f}^{x_1;x_2, x_3 }} "1";"2" };
{\ar@{-->} "3";"1"_{\scriptscriptstyle  {f}^{ x_2; x_3}} };
{\ar@{<--} "3";"6"^{\scriptscriptstyle  {f}^{ x_3 }} };
{\ar@{->}"5";"2"^{\scriptscriptstyle f^{x_3; x_1 ,x_2}}};
{\ar@{->}|-{\scriptscriptstyle f^{ x_1;x_2}} "4";"5" };
{\ar@{->}|-{\scriptscriptstyle f^{x_2 }} "6";"4" };
{\ar@{->}|-{\scriptscriptstyle f^{x_3;x_2}} "4";"1" };
{\ar@{~>}|-{\scriptscriptstyle \widetilde f^{x_3; x_1 ,x_2}}"12";"9"};
{\ar@{->}|-{\scriptscriptstyle \widetilde f^{ x_1;x_2}} "11";"12" };
{\ar@{->}|-{\scriptscriptstyle \widetilde f^{x_2 }} "13";"11" };
{\ar@{~>}|-{\scriptscriptstyle \widetilde f^{ x_1  }}"13"; "14"};
{\ar@{~>}"14";"12"^{\scriptscriptstyle \widetilde f^{x_2;x_1 }}};
 {\ar@{<~ }|-{\scriptscriptstyle g^{ x_1,x_2,x_3 }} "2";"9" };
 {\ar@{<-}|-{\scriptscriptstyle g^{ x_2  }} "4";"11" };{\ar@{<- }"5";"12"_{\scriptscriptstyle g^{ x_1,x_2  }} };
 {\ar@{<--}|-{\scriptscriptstyle g^{0  }} "6";"13" };( 35,-8)*+{ \scriptstyle  {\rm figure}:\;\hat \Sigma_2 }="30";
\endxy \qquad \xy 0;/r.17pc/:
(0,0 )*+{\bullet}="1";
(60, 0)*+{\bullet}="2";
(-10,20)*+{\bullet}="3";
(20,30)*+{\bullet}="4";
( 10, 50)*+{\bullet}="6";
(0,80 )*+{\bullet}="8";
(60, 80)*+{\bullet}="9";
(20,100)*+{\bullet}="11";
(80,100 )*+{\bullet}="12";
( 10, 120)*+{\bullet}="13";
(  70,120)*+{\bullet}="14";
( 40,90)*+{(2)}="21";( 45,110)*+{(1)}="22";
( 15,100)*+{(5)}="23";
( 15,85)*+{(4)}="24";
( 50,65)*+{(3)}="24";
(  -5,20)*+{(6)}="26";
 {\ar@{-->}|-{\scriptscriptstyle  {f}^{x_1;x_2, x_3}} "1";"2" };
{\ar@{-->} "3";"1"_{\scriptscriptstyle  {f}^{ x_2; x_3}} };
{\ar@{<--}|-{\scriptscriptstyle  {f}^{ x_3 }} "3";"6" };
{\ar@{->}|-{\scriptscriptstyle f^{x_2 }} "6";"4" };
{\ar@{->}|-{\scriptscriptstyle f^{x_3;x_2}} "4";"1" };
 {\ar@{->}|-{\scriptscriptstyle \widetilde f^{x_1;x_2, x_3}} "8";"9" };
{\ar@{~>}|-{\scriptscriptstyle \widetilde f^{x_3; x_1 ,x_2}}"12";"9"};
{\ar@{->}|-{\scriptscriptstyle \widetilde f^{ x_1;x_2}} "11";"12" };
{\ar@{->}|-{\scriptscriptstyle \widetilde f^{x_2 }} "13";"11" };
{\ar@{->}|-{\scriptscriptstyle \widetilde f^{x_3;x_2}} "11";"8" };
{\ar@{~>}|-{\scriptscriptstyle \widetilde f^{ x_1  }}"13"; "14"};
{\ar@{~>}"14";"12"^{\scriptscriptstyle \widetilde f^{x_2;x_1 }}};
{\ar@{<- }|-{\scriptscriptstyle g^{x_2, x_3 }} "1";"8" };
 {\ar@{<~ }|-{\scriptscriptstyle g^{ x_1,x_2,x_3 }} "2";"9" };
  {\ar@{<-}|-{\scriptscriptstyle g^{ x_2  }} "4";"11" };
 {\ar@{<--}|-{\scriptscriptstyle g^{0  }} "6";"13" };( 35,-8)*+{ \scriptstyle  {\rm figure}:\;\hat \Sigma_5 }="30";
\endxy
\end{equation*}

\begin{equation*}
    \xy 0;/r.17pc/:
(0,0 )*+{\bullet}="1";
(60, 0)*+{\bullet}="2";
(-10,20)*+{\bullet}="3";
( 10, 50)*+{\bullet}="6";
(0,80 )*+{\bullet}="8";
(60, 80)*+{\bullet}="9";
(-10,100)*+{\bullet}="10";
(20,110)*+{\bullet}="11";
(80,110 )*+{\bullet}="12";
( 10, 130)*+{\bullet}="13";
(  70,130)*+{\bullet}="14";
( 40,90)*+{(2)}="21";( 45,120)*+{(1)}="22";
(  5,60)*+{(6)}="23";
(  5,100)*+{(4)}="24";
( 30,45)*+{(3)}="24";
(  -5,70)*+{(5)}="26";
 {\ar@{-->}|-{ {f}^{x_1;x_2,x_3 }} "1";"2" };
{\ar@{-->}|-{ {f}^{ x_2; x_3}} "3";"1" };
{\ar@{<--}|-{ {f}^{ x_3 }} "3";"6"};
 {\ar@{->}|-{\widetilde f^{x_1;x_2, x_3}} "8";"9" };
{\ar@{->}|-{\widetilde f^{ x_2; x_3}} "10";"8" };
{\ar@{<-} "10";"13"^{\widetilde f^{ x_3 }} };
{\ar@{~>}|-{\widetilde f^{x_3; x_1 ,x_2}}"12";"9"};
{\ar@{->}|-{\widetilde f^{ x_1;x_2}} "11";"12" };
{\ar@{->}|-{\widetilde f^{x_2 }} "13";"11" };
{\ar@{->}|-{\widetilde f^{x_3;x_2}} "11";"8" };
{\ar@{~>}|-{\widetilde f^{ x_1  }}"13"; "14"};
{\ar@{~>}|-{\widetilde f^{x_2;x_1 }}"14";"12"};
{\ar@{<- }|-{g^{x_2, x_3 }} "1";"8" };
 {\ar@{<~ }|-{g^{ x_1,x_2,x_3 }} "2";"9" };
 {\ar@{<- }|-{g^{ x_3  }} "3";"10" };
 {\ar@{<--}|-{g^{0  }} "6";"13" };( 35,-8)*+{ \scriptstyle  {\rm figure}:\; \hat \Sigma_6 }="30";
\endxy
\end{equation*}
 The fourth $3$-arrow is $(g_*,\hat\Sigma_4,\hat K_4)$   with
$\hat
K_4= \widetilde{F}^{x_1,x_2 } \rhd' \left[\widetilde{f}^{x_2}\rhd k^{x_1,x_3;x_2}\right] ,$
 whose   $2$-target is $(g_*,\hat \Sigma_5)$.
 The fifth $3$-arrow is $(g_*,\hat\Sigma_5,\hat K_5)$  with
\begin{equation*}
\hat K_5= \left[\widetilde{F}^{x_1,x_2 }\cdot    \widetilde{f}^{x_2}\rhd\left( \widetilde{F}^{ x_1, x_3 ;x_2}\cdot\widetilde{f}^{ x_3  ;x_2} \rhd h^{x_1;x_2, x_3} \right)\right] \rhd'  k^{  x_2,x_3 }  ,\end{equation*}
 whose   $2$-target is $(g_*,\hat \Sigma_6)$, where the part before $ \rhd'$ is the whiskering corresponding to  the composition of $2$-arrows (1), (2) and (3) in the figure $\hat\Sigma_5$.
 The last $3$-arrow    is $(g_*,\hat\Sigma_6,\hat K_6)$,  interchanging $2$-arrows  $(4)$ and $ (3)$ in the figure $ \hat\Sigma_6$, with
    \begin{equation*}\begin{split}\hat
K_6=  &\left[\widetilde{F}^{x_1,x_2 }\cdot   \widetilde{f}^{x_2} \rhd \widetilde{F}^{ x_1, x_3 ;x_2} \right] \rhd'\left[ \widetilde{F}^{x_2 ,x_3  } \rhd'\left\{ \left(\widetilde{F}^{x_2 ,x_3  }\right)^{-1},  (\widetilde{f}^{x_2}\widetilde{f}^{  x_3; x_2})\rhd  h^{  x_1,; x_2, x_3}\right \}\right ] ,\end{split}\end{equation*}
  whose $2$-target is   $ (g_*,\hat \Sigma_7)= (g_*,  \Sigma_8)$ in the above subsection.

The $0$-th $3$-arrow is $(g_*,\hat\Sigma_0,\hat K_0)$ with
\begin{equation*}
 \hat K_0= {F}_3^*(\Omega_{x_1,x_2,x_3})=\left[ \widetilde{f}^{x_1} \rhd\left(\widetilde{f}^{x_2;x_1  } \rhd h^{ x_3; x_1 ,x_2}\cdot h^{ x_2;  x_1  }\right)\cdot h^{   x_1  }\right]\rhd' [g^0\rhd {k}^{x_1,x_2,x_3}],
\end{equation*}where the part before $ \rhd'$ is the whiskering corresponding to   the composition of $2$-arrows (1), (2) and (3) in the figure $\hat\Sigma_1$.
Now $\Psi_{\Sigma_+}$ in (\ref{eq:naturality}) is $ \hat K_1 \hat K_2\cdots \hat  K_6$ and the LHS of (\ref{eq:naturality}) gives $ \hat K_0 \hat K_1\cdots \hat  K_6$.
The derivative  $\frac {\partial^3 \hat K_0\hat K_1\cdots \hat  K_6}{\partial x_1\partial x_2\partial x_3}$ at $(0,0,0)$ gives
\begin{equation}\label{eq:derivative-L}\begin{split}& g^0\rhd C(v_1,v_2,v_3)+\varphi(v_3)\rhd'\psi(v_1,v_2)-
 \{ \widetilde{B}(v_1,v_2),\varphi(v_3) \}+\{\varphi(v_2),g^0\rhd B(v_1,v_3)\}\\+& v_2 \psi(v_1,v_3)+ \widetilde{A}(v_2)\rhd \psi(v_1,v_3)+    \varphi(v_1)\rhd'\psi(v_2,v_3)-
 \{ \widetilde{B}(v_2,v_3),\varphi(v_1) \} .
\end{split}\end{equation}
The derivatives of both sides of (\ref{eq:naturality}) give   (\ref{eq:derivative-R})= (\ref{eq:derivative-L}), which is exactly the gauge transformation formula (\ref{eq:C}) for the $C$-field.
\section{The $3$-holonomies  and the $3$-curvatures }
For  $(x_1,x_2,x_3,x_4)\in \mathbb{R}^4$, consider a smooth family of  $4$-paths $\dot{\Theta}_{x_1,x_2,x_3,x_4}: [0,1]^4\rightarrow [0,x_1]\times\cdots\times[0,x_4]\subset\mathbb{R}^4$,  and $ \Theta_{x_1,x_2,x_3,x_4}:=\Gamma \circ\dot{\Theta}_{x_1,x_2,x_3,x_4}$, where $\Gamma$ is a mapping $\mathbb{R}^4\rightarrow U\subset \mathbb{R}^n$ such that
$\frac {\partial\Gamma}{\partial x_j}(0,0,0,0)=v_j$ for fixed $v_j\in T_xU$, $j=1,\ldots,4$.

As before, we use notations $\gamma^{x_i;*}$ for $1$-paths,    $\Sigma_{x_i, x_j;*}$ for $2$-paths and  $\Omega_{x_i, x_j,x_k;*}$ for $3 $-paths. Under the action of a $\mathbf{Gray}$-functor $F$, we get $1$-arrows $f^{x_i;*}$, $2$-arrows $(*,F^{x_i, x_j;*})$  and  $3$-arrows $(*,*,k^{x_i, x_j,x_k;*})$.

The boundary $ \partial{\Theta}_{x_1,x_2,x_3,x_4}$ of the $4$-path  $  {\Theta}_{x_1,x_2,x_3,x_4}$ is a closed $3$-path, which is the composition of two $3$-paths corresponding to $\Sigma_-\times[0,x_4]\cup\Omega_{x_1, x_2,x_3;x_4}$ and  $\Omega_{x_1, x_2,x_3}\cup\Sigma_+\times[0,x_4]$,   respectively, where $\Sigma_\mp$ are the $2$-source and $2$-target in (\ref{eq:Omega-boundary}) of the $3$-path $\Omega_{x_1, x_2,x_3 }$. Each of these two $3$-paths is a composition  of several $3$-paths.   Then $F(\partial\Theta_{x_1,x_2,x_3,x_4})$ is the $3$-dimensional holonomy, the lattice version of $3$-curvature. Let us write it down explicitly.

\subsection{$3$-arrows corresponding to $\Sigma_-\times[0,x_4]\cup\Omega_{x_1, x_2,x_3;x_4}$ }
The first $3$-arrow is $(g_*, \Sigma_1,  K_1)$, interchanging $2$-arrows  $(3)$ and $ (4)$ in the figure $  \Sigma_1$, with
  \begin{equation*}\begin{split}
K_1=&  K_1^0
\rhd'\left [F^{ x_1 ,x_4 }\rhd'\left \{ (F^{ x_1 ,x_4 })^{-1}, (  {f}^{x_1} {f}^{x_4;x_1}  )\rhd F^{x_2,x_3;x_1 ,x_4  }\right\}^{-1}\right ] ,\end{split}\end{equation*}
whose   $2$-target is $(g_*,\Sigma_2)$, where  $ K_1^0={f}^{x_1}\rhd\left(  {f}^{x_2;x_1  } \rhd F^{x_3,x_4;x_1 ,x_2} \cdot    F^{ x_2,x_4;x_1 }\right)$ is the whiskering corresponding to   the composition of $2$-arrows (1) and (2) in the figure $\Sigma_1$.
     The second $3$-arrow is $(g_*, \Sigma_2,  K_2)$ with    $K_2 ={f}^{x_1}\rhd k^{x_2,x_3,x_4;x_1}$, whose $2$-target is $(g_*,\Sigma_3)$.
\begin{equation*}
     \xy 0;/r.17pc/:
(0,0 )*+{\bullet}="1";
(60, 0)*+{\bullet}="2";
( 70,15)*+{\bullet}="3";
(-20,30)*+{\bullet}="4";
(40,30 )*+{\bullet}="5";
(-10, 50)*+{\bullet}="6";
(  50,50)*+{\bullet}="7";
(0,80 )*+{\bullet}="8";
(60,80)*+{\bullet}="9";
( 70,95)*+{\bullet}="10";
(-20,110)*+{\bullet}="11";
(40,110 )*+{\bullet}="12";
(-10, 130)*+{\bullet}="13";
(  50,130)*+{\bullet}="14";
 {\ar@{-->}|-{\scriptscriptstyle f^{x_1;x_2 ,x_3,x_4}} "1";"2" };
{\ar@{-->}|-{\scriptscriptstyle f^{ x_2; x_3,x_4}} "4";"1" };
{\ar@{->} "3";"2"^{\scriptscriptstyle f^{ x_3;x_1,x_2,x_4}} };
{\ar@{->}|-{\scriptscriptstyle f^{x_2; x_1 ,x_3,x_4}} "5";"2" };
{\ar@{->}|-{\scriptscriptstyle f^{ x_1;x_3,x_4}} "4";"5" };
{\ar@{-->}|-{\scriptscriptstyle f^{x_3 ;x_4}} "6";"4" };
{\ar@{->}|-{\scriptscriptstyle f^{x_1;x_4}} "6";"7" };
{\ar@{->}|-{\scriptscriptstyle f^{x_2;x_1,x_4  }} "7";"3" };
{\ar@{->} "7";"5"_{\scriptscriptstyle f^{x_3;x_1,x_4 }} };
{\ar@{->}|-{\scriptscriptstyle { f}^{x_1;x_2,x_3}} "8";"9" };
{\ar@{->}|-{\scriptscriptstyle { f}^{x_2;x_3}} "11";"8" };
{\ar@{~>} "10";"9"^{\scriptscriptstyle {f}^{ x_3;x_1,x_2}} };
{\ar@{->}|-{\scriptscriptstyle { f}^{x_2; x_1 ,x_3}} "12";"9" };
{\ar@{->}|-{ \scriptscriptstyle{ f}^{x_1;x_3}} "11";"12" };
{\ar@{->}|-{\scriptscriptstyle { f}^{x_3 }} "13";"11" };
{\ar@{~>}|-{\scriptscriptstyle { f}^{x_1}} "13";"14" };
{\ar@{~>}|-{\scriptscriptstyle {f}^{x_2;x_1  }} "14";"10" };
{\ar@{->}|-{\scriptscriptstyle {f}^{x_3;x_1 }} "14";"12" };
 {\ar@{<- }|-{\scriptscriptstyle f^{x_4;x_2, x_3 }} "1";"8" };
 {\ar@{<~ }|-{\scriptscriptstyle f^{x_4; x_1,x_2,x_3 }} "2";"9" };
 {\ar@{<- }|-{\scriptscriptstyle f^{x_4; x_1,x_2 }} "3";"10" };
 {\ar@{<-}|-{\scriptscriptstyle f^{x_4; x_3  }} "4";"11" };{\ar@{<- }|-{\scriptscriptstyle f^{x_4; x_1,x_3  }} "5";"12" };
 {\ar@{<--}|-{\scriptscriptstyle f^{x_4   }} "6";"13" };{\ar@{<- }|-{\scriptscriptstyle f^{x_4; x_1 }} "7";"14" };( 27,-8)*+{ \scriptstyle  {\rm figure}:\;  \Sigma_-\times[0,x_4] }="30";
  \endxy
\qquad
     \xy 0;/r.17pc/:
(0,0 )*+{\bullet}="1";
(60, 0)*+{\bullet}="2";
( 70,15)*+{\bullet}="3";
(-20,30)*+{\bullet}="4";
(40,30 )*+{\bullet}="5";
(-10, 50)*+{\bullet}="6";
(  50,50)*+{\bullet}="7";
(60,80)*+{\bullet}="9";
( 70,95)*+{\bullet}="10";
(-10, 130)*+{\bullet}="13";
(  50,130)*+{\bullet}="14";
( 65,50)*+{(1)}="21";
( 60,95)*+{(2)}="22";
( 20,85)*+{(3)}="23";
( 50,30)*+{(4)}="24";
( 10,40)*+{(5)}="24";
( 25,15)*+{(6)}="26";
 {\ar@{-->}|-{\scriptscriptstyle f^{x_1;x_2 ,x_3,x_4}} "1";"2" };
{\ar@{-->}|-{\scriptscriptstyle f^{ x_2; x_3,x_4}} "4";"1" };
{\ar@{->} "3";"2"^{\scriptscriptstyle f^{ x_3;x_1,x_2,x_4}} };
{\ar@{->}|-{\scriptscriptstyle f^{x_2; x_1 ,x_3,x_4}} "5";"2" };
{\ar@{->}|-{\scriptscriptstyle f^{ x_1;x_3,x_4}} "4";"5" };
{\ar@{-->}|-{\scriptscriptstyle f^{x_3 ;x_4}} "6";"4" };
{\ar@{->}|-{\scriptscriptstyle f^{x_1;x_4}} "6";"7" };
{\ar@{->}|-{\scriptscriptstyle f^{x_2;x_1,x_4  }} "7";"3" };
{\ar@{->} "7";"5"_{\scriptscriptstyle f^{x_3;x_1,x_4 }} };
{\ar@{~>} "10";"9"^{\scriptscriptstyle {f}^{ x_3;x_1,x_2}} };
{\ar@{~>}|-{\scriptscriptstyle { f}^{x_1}} "13";"14" };
{\ar@{~>}|-{\scriptscriptstyle {f}^{x_2;x_1  }} "14";"10" };
 {\ar@{<~ }|-{\scriptscriptstyle f^{x_4; x_1,x_2,x_3 }} "2";"9" };
 {\ar@{<- }|-{\scriptscriptstyle f^{x_4; x_1,x_2 }} "3";"10" };
  {\ar@{<--}|-{\scriptscriptstyle f^{x_4   }} "6";"13" };{\ar@{<- }|-{\scriptscriptstyle f^{x_4; x_1 }} "7";"14" };( 35,-8)*+{ \scriptstyle  {\rm figure}:\;  \Sigma_1 }="30";
  \endxy
\end{equation*}

The third $3$-arrow is $(g_*, \Sigma_3,  K_3)$ with
\begin{equation*}
    K_3=\left [ {f}^{x_1}\rhd\left(  F^{x_2 ,x_3 ;x_1 } \cdot   {f}^{ x_3;x_1}   \rhd F^{ x_2,x_4;x_1 ,x_3} \right) \right]\rhd' k^{x_1,x_3,x_4},
\end{equation*}
 whose  $2$-target is $(g_*,\Sigma_4)$, where the part before $ \rhd'$ is the whiskering corresponding to   the composition of $2$-arrows (1) and (2) in the figure $\Sigma_3$.
 The fourth  $3$-arrow     is  $(g_*, \Sigma_4,  K_4)$, interchanging $2$-arrows    $(3)$ and $ (2)$ in the figure $\Sigma_4$, with
 \begin{equation*}
    K_4=(  {f}^{x_1} \rhd  {F}^{x_2 ,x_3;x_1 }) \rhd'\left[{F}^{x_1 ,x_3  } \rhd' \left\{({F}^{x_1 ,x_3  })^{-1} , ( {f}^{x_1} {f}^{ x_3;x_1})   \rhd F^{ x_2,x_4;x_1 ,x_3}\right\} \right],
\end{equation*}whose $2$-target is  $(g_*,\Sigma_5)$.
 The fifth  $3$-arrow  is $(g_*, \Sigma_5,  K_5)$,  interchanging $2$-arrows    $(5)$ and $ (6)$ in the figure $\Sigma_4$,  with
 \begin{equation*}\begin{split}
    K_5=&K_5^0\rhd'\left[ F^{x_3 , x_4 }\rhd'\left \{(F^{x_3 , x_4 } )^{-1} ,({f}^{x_3}{f}^{x_4;x_3}  )\rhd F^{x_1 , x_2;x_3,x_4  }\right \}^{-1}\right],
\end{split}\end{equation*}whose $2$-target is $(g_*,\Sigma_6)$, where   $ K_5^0= {f}^{x_1}\rhd   F^{x_2 ,x_3 ;x_1 } \cdot   F^{x_1 ,x_3  } \cdot {f}^{x_3} \rhd \left( {f}^{ x_1;x_3} \rhd F^{ x_2,x_4;x_1 ,x_3} \cdot    F^{ x_1 ,x_4; x_3}\right )$ is the whiskering corresponding to   the composition of $2$-arrows (1), (3), (2) and (4) in the figure $\Sigma_4$.

The sixth $3$-arrow is    $(g_*, \Sigma_6,  K_6)$   with
\begin{equation*}
    K_6=\left [ {f}^{x_1}\rhd   F^{x_2 ,x_3 ;x_1 } \cdot   F^{x_1 ,x_3  }\right ]\rhd' \left[ {f}^{x_3}  \rhd k^{x_1,x_2,x_4;x_3}\right].
\end{equation*}
whose   $2$-target is $(g_*,\Sigma_7)$.
\begin{equation*}
     \xy 0;/r.17pc/:
(0,0 )*+{\bullet}="1";
(60, 0)*+{\bullet}="2";
(-20,30)*+{\bullet}="4";
(40,30 )*+{\bullet}="5";
(-10, 50)*+{\bullet}="6";
(  50,50)*+{\bullet}="7";
(60,80)*+{\bullet}="9";
( 70,95)*+{\bullet}="10";
(40,110 )*+{\bullet}="12";
(-10, 130)*+{\bullet}="13";
(  50,130)*+{\bullet}="14";
( 45,80)*+{(3)}="21";
( 60,100)*+{(1)}="22";
( 20,85)*+{(4)}="23";
( 55,30)*+{(2)}="24";
( 10,40)*+{(5)}="24";
( 25,15)*+{(6)}="26";
 {\ar@{-->}|-{\scriptscriptstyle  f^{x_1;x_2 ,x_3,x_4}} "1";"2" };
{\ar@{-->}|-{\scriptscriptstyle f^{ x_2; x_3,x_4}} "4";"1" };
{\ar@{->}|-{\scriptscriptstyle f^{x_2; x_1 ,x_3,x_4}} "5";"2" };
{\ar@{->}|-{\scriptscriptstyle f^{ x_1;x_3,x_4}} "4";"5" };
{\ar@{-->}|-{\scriptscriptstyle f^{x_3 ;x_4}} "6";"4" };
{\ar@{->}|-{\scriptscriptstyle f^{x_1;x_4}} "6";"7" };
{\ar@{->}|-{\scriptscriptstyle f^{x_3;x_1,x_4 }} "7";"5" };
{\ar@{~>}|-{\scriptscriptstyle  {f}^{ x_3;x_1,x_2}} "10";"9" };
{\ar@{->}|-{\scriptscriptstyle  {f}^{x_2; x_1 ,x_3}} "12";"9" };
{\ar@{~>}|-{\scriptscriptstyle  {f}^{x_1}} "13";"14" };
{\ar@{~>}|-{ \scriptscriptstyle {f}^{x_2;x_1  }} "14";"10" };
{\ar@{->}|-{ \scriptscriptstyle {f}^{x_3;x_1 }} "14";"12" };
  {\ar@{<~ } "2";"9"_{\scriptscriptstyle f^{x_4; x_1,x_2,x_3 }} };
{\ar@{<- }"5";"12"^{\scriptscriptstyle f^{x_4; x_1,x_3  }} };
 {\ar@{<--}|-{\scriptscriptstyle f^{x_4   }} "6";"13" };{\ar@{<- }|-{\scriptscriptstyle f^{x_4; x_1 }} "7";"14" };( 35,-8)*+{ \scriptstyle  {\rm figure}:\,  \Sigma_3}="30";
  \endxy
\qquad
     \xy 0;/r.17pc/:
(0,0 )*+{\bullet}="1";
(60, 0)*+{\bullet}="2";
(-20,30)*+{\bullet}="4";
(40,30 )*+{\bullet}="5";
(-10, 50)*+{\bullet}="6";
(60,80)*+{\bullet}="9";
( 70,95)*+{\bullet}="10";
(-20,110)*+{\bullet}="11";
(40,110 )*+{\bullet}="12";
(-10, 130)*+{\bullet}="13";
(  50,130)*+{\bullet}="14";
( 20,120)*+{(3)}="21";
( 60,100)*+{(1)}="22";
( 20,75)*+{(4)}="23";
( 55,70)*+{(2)}="24";
( 10,15)*+{(6)}="24";
( -15,55)*+{(5)}="26";
 {\ar@{-->}|-{\scriptscriptstyle f^{x_1;x_2 ,x_3,x_4}} "1";"2" };
{\ar@{-->}|-{\scriptscriptstyle f^{ x_2; x_3,x_4}} "4";"1" };
{\ar@{->}|-{\scriptscriptstyle f^{x_2; x_1 ,x_3,x_4}} "5";"2" };
{\ar@{->}|-{\scriptscriptstyle f^{ x_1;x_3,x_4}} "4";"5" };
{\ar@{-->}|-{\scriptscriptstyle f^{x_3 ;x_4}} "6";"4" };
{\ar@{~>} "10";"9"^{ {f}^{ x_3;x_1,x_2}} };
{\ar@{->}|-{\scriptscriptstyle  {f}^{x_2; x_1 ,x_3}} "12";"9" };
{\ar@{->}|-{ \scriptscriptstyle  {f}^{x_1;x_3}} "11";"12" };
{\ar@{->}|-{\scriptscriptstyle  {f}^{x_3 }} "13";"11" };
{\ar@{~>}|-{\scriptscriptstyle  {f}^{x_1}} "13";"14" };
{\ar@{~>}|-{\scriptscriptstyle  {f}^{x_2;x_1  }} "14";"10" };
{\ar@{->}|-{\scriptscriptstyle  {f}^{x_3;x_1 }} "14";"12" };
 {\ar@{<~ }|-{\scriptscriptstyle f^{x_4; x_1,x_2,x_3 }} "2";"9" };
  {\ar@{<-}|-{\scriptscriptstyle f^{x_4; x_3  }} "4";"11" };{\ar@{<- }|-{\scriptscriptstyle f^{x_4; x_1,x_3  }} "5";"12" };
 {\ar@{<--}|-{\scriptscriptstyle f^{x_4   }} "6";"13" };( 35,-8)*+{ \scriptstyle  {\rm figure}:\,  \Sigma_4 }="30";
  \endxy
\end{equation*}

  \begin{equation*}
     \xy 0;/r.17pc/:
(0,0 )*+{\bullet}="1";
(60, 0)*+{\bullet}="2";
(-20,30)*+{\bullet}="4";
(-10, 50)*+{\bullet}="6";
(0,80 )*+{\bullet}="8";
(60,80)*+{\bullet}="9";
( 70,95)*+{\bullet}="10";
(-20,110)*+{\bullet}="11";
(40,110 )*+{\bullet}="12";
(-10, 130)*+{\bullet}="13";
(  50,130)*+{\bullet}="14";
( 20,120)*+{(2)}="21";
( 60,100)*+{(1)}="22";
( -10,30)*+{(5)}="23";
( 25,90)*+{(3)}="24";
( 30,45)*+{(4)}="24";
( -15,55)*+{(6)}="26";
 {\ar@{-->}|-{\scriptscriptstyle f^{x_1;x_2 ,x_3,x_4}} "1";"2" };
{\ar@{-->}|-{\scriptscriptstyle f^{ x_2; x_3,x_4}} "4";"1" };
{\ar@{-->}|-{\scriptscriptstyle f^{x_3 ;x_4}} "6";"4" };
{\ar@{->}|-{\scriptscriptstyle  {f}^{x_1;x_2,x_3}} "8";"9" };
{\ar@{->}|-{ \scriptscriptstyle {f}^{x_2;x_3}} "11";"8" };
{\ar@{~>} "10";"9"^{ {f}^{ x_3;x_1,x_2}} };
{\ar@{->}|-{ \scriptscriptstyle {f}^{x_2; x_1 ,x_3}} "12";"9" };
{\ar@{->}|-{\scriptscriptstyle  {f}^{x_1;x_3}} "11";"12" };
{\ar@{->}|-{ \scriptscriptstyle {f}^{x_3 }} "13";"11" };
{\ar@{~>}|-{\scriptscriptstyle  {f}^{x_1}} "13";"14" };
{\ar@{~>}|-{\scriptscriptstyle  {f}^{x_2;x_1  }} "14";"10" };
{\ar@{->}|-{ \scriptscriptstyle  {f}^{x_3;x_1 }} "14";"12" };
 {\ar@{<- }|-{\scriptscriptstyle  f^{x_4;x_2, x_3 }} "1";"8" };
 {\ar@{<~ }|-{\scriptscriptstyle f^{x_4; x_1,x_2,x_3 }} "2";"9" };
 {\ar@{<-}|-{\scriptscriptstyle f^{x_4; x_3  }} "4";"11" };
 {\ar@{<--}|-{\scriptscriptstyle f^{x_4   }} "6";"13" };( 35,-8)*+{ \scriptstyle  {\rm figure}:\,  \Sigma_7 }="30";
   \endxy \qquad
     \xy 0;/r.17pc/:
(0,0 )*+{\bullet}="1";
(60, 0)*+{\bullet}="2";
(-10,20)*+{\bullet}="3";
( 10, 50)*+{\bullet}="6";
(0,80 )*+{\bullet}="8";
(60, 80)*+{\bullet}="9";
(-10,100)*+{\bullet}="10";
(20,110)*+{\bullet}="11";
(80,110 )*+{\bullet}="12";
( 10, 130)*+{\bullet}="13";
(  70,130)*+{\bullet}="14";
( 40,90)*+{(2)}="21";( 45,120)*+{(1)}="22";
(  5,60)*+{(6)}="23";
(  5,100)*+{(3)}="24";
( 30,45)*+{(4)}="24";
(  -5,70)*+{(5)}="26";
 {\ar@{-->}|-{\scriptscriptstyle  {f}^{x_1;x_2,x_3,x_4 }} "1";"2" };
{\ar@{-->} "3";"1"_{\scriptscriptstyle  {f}^{ x_2; x_3,x_4}} };
{\ar@{<--} "3";"6"^{\scriptscriptstyle  {f}^{ x_3;x_4 }} };
 {\ar@{->}|-{\scriptscriptstyle   f^{x_1;x_2, x_3}} "8";"9" };
{\ar@{->}|-{\scriptscriptstyle  f^{ x_2; x_3}} "10";"8" };
{\ar@{<-} "10";"13"^{\scriptscriptstyle   f^{ x_3 }} };
{\ar@{~>}|-{\scriptscriptstyle  f^{x_3; x_1 ,x_2}}"12";"9"};
{\ar@{->}|-{\scriptscriptstyle   f^{ x_1;x_2}} "11";"12" };
{\ar@{->}|-{\scriptscriptstyle  f^{x_2 }} "13";"11" };
{\ar@{->}|-{\scriptscriptstyle    f^{x_3;x_2}} "11";"8" };
{\ar@{~>}|-{\scriptscriptstyle   f^{ x_1  }}"13"; "14"};
{\ar@{~>}"14";"12"^{\scriptscriptstyle   f^{x_2;x_1 }}};
{\ar@{<- }|-{\scriptscriptstyle f^{x_4 ;x_2, x_3 }} "1";"8" };
 {\ar@{<~ }|-{\scriptscriptstyle f^{x_4 ;x_1,x_2,x_3 }} "2";"9" };
 {\ar@{<- }|-{\scriptscriptstyle f^{x_4 ; x_3  }} "3";"10" };
  {\ar@{<--}|-{\scriptscriptstyle f^{x_4  }} "6";"13" };( 35,-8)*+{ \scriptstyle  {\rm figure}:\,  \Sigma_8 }="30";
\endxy
\end{equation*}

The last $3$-arrow is  $(g_*, \Sigma_7,  K_7)$   with
$
  K_7= \pi_L \circ{F}_3(\Omega_{x_1,x_2,x_3 })= {k}^{x_1,x_2,x_3 },
$ whose   $2$-target is $(g_*,\Sigma_8)$.

The derivative  $\frac {\partial^4K_1K_2\cdots K_7}{\partial x_1\partial x_2\partial x_3\partial x_4}$ at $(0,0,0,0)$ gives
\begin{equation}\label{eq:C-}\begin{split}&
  \{B(v_1,v_4), B(v_2,v_3)\}+    v_1C(v_2,v_3,v_4)+A( v_1 )\rhd C(v_2,v_3,v_4)-\{B(v_1,v_3), B(v_2,v_4)\}\\&+\{B(v_3,v_4), B(v_1,v_2)\}+ v_3C(v_1,v_2,v_4) + A(v_3)\rhd  C(v_1,v_2,v_4).
\end{split}\end{equation}
\subsection{$3$-arrows corresponding to $\Omega_{x_1, x_2,x_3;x_4}\cup\Sigma_+\times[0,x_4]$ }
 
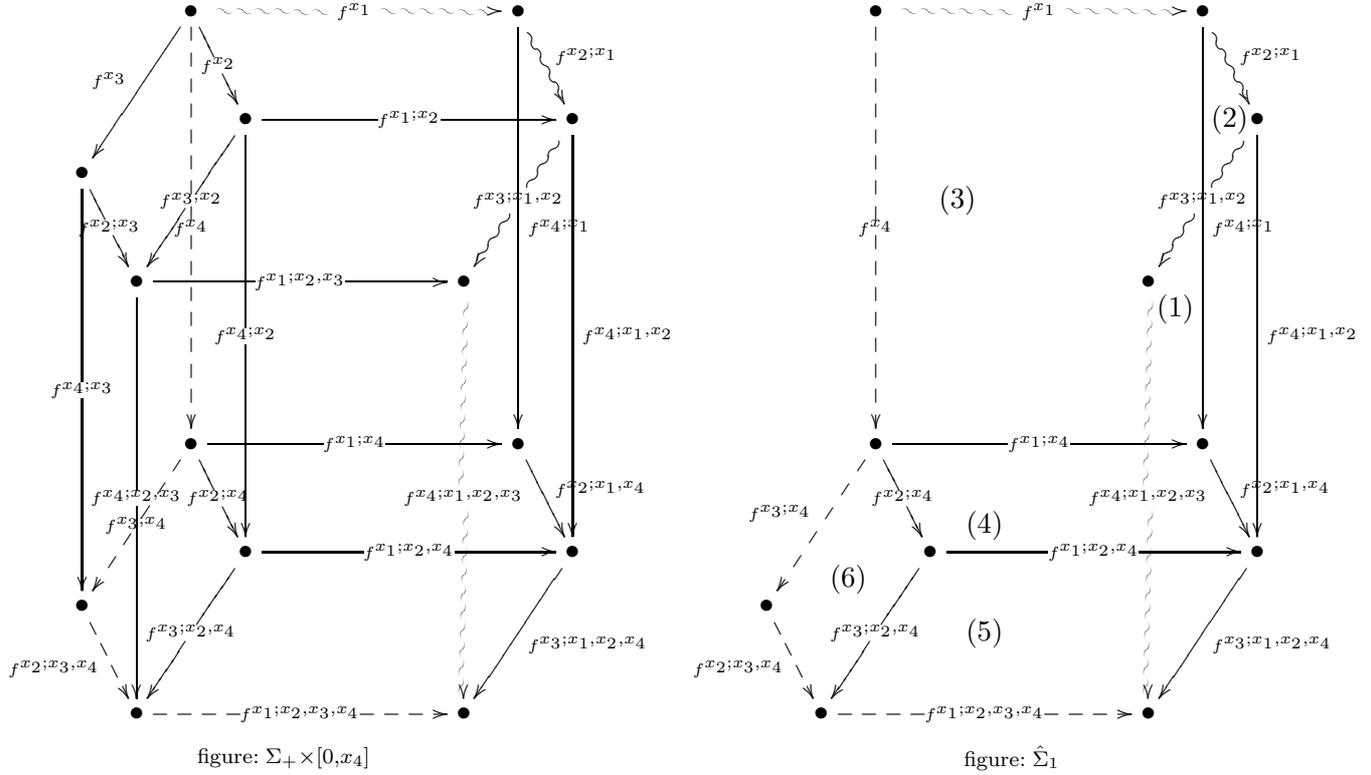
\begin{equation*}
     \xy 0;/r.17pc/:
(0,0 )*+{\bullet}="1";
(60, 0)*+{\bullet}="2";
(-10,20)*+{\bullet}="3";
(20,30)*+{\bullet}="4";
(80,30 )*+{\bullet}="5";
( 10, 50)*+{\bullet}="6";
(  70,50)*+{\bullet}="7";
(0,80 )*+{\bullet}="8";
(60, 80)*+{\bullet}="9";
(-10,100)*+{\bullet}="10";
(20,110)*+{\bullet}="11";
(80,110 )*+{\bullet}="12";
( 10, 130)*+{\bullet}="13";
(  70,130)*+{\bullet}="14";
 {\ar@{-->}|-{\scriptscriptstyle  {f}^{x_1;x_2,x_3,x_4 }} "1";"2" };
{\ar@{-->} "3";"1"_{\scriptscriptstyle  {f}^{ x_2; x_3,x_4}} };
{\ar@{<--}|-{\scriptscriptstyle  {f}^{ x_3;x_4 }} "3";"6" };
{\ar@{->}"5";"2"^{\scriptscriptstyle f^{x_3; x_1 ,x_2,x_4}}};
{\ar@{->}|-{\scriptscriptstyle f^{ x_1;x_2,x_4}} "4";"5" };
{\ar@{->}|-{\scriptscriptstyle f^{x_2 ; x_4}} "6";"4" };
{\ar@{->}|-{\scriptscriptstyle f^{x_3;x_2,x_4}} "4";"1" };
{\ar@{->}|-{\scriptscriptstyle f^{ x_1;x_4  }}"6"; "7"};
{\ar@{->}"7";"5"^{\scriptscriptstyle f^{x_2;x_1,x_4 }}};
 {\ar@{->}|-{\scriptscriptstyle   f^{x_1;x_2, x_3}} "8";"9" };
{\ar@{->}|-{\scriptscriptstyle  f^{ x_2; x_3}} "10";"8" };
{\ar@{<-} "10";"13"^{\scriptscriptstyle   f^{ x_3 }} };
{\ar@{~>}|-{\scriptscriptstyle  f^{x_3; x_1 ,x_2}}"12";"9"};
{\ar@{->}|-{\scriptscriptstyle   f^{ x_1;x_2}} "11";"12" };
{\ar@{->}|-{\scriptscriptstyle  f^{x_2 }} "13";"11" };
{\ar@{->}|-{\scriptscriptstyle    f^{x_3;x_2}} "11";"8" };
{\ar@{~>}|-{\scriptscriptstyle   f^{ x_1  }}"13"; "14"};
{\ar@{~>}"14";"12"^{\scriptscriptstyle   f^{x_2;x_1 }}};
{\ar@{<- }|-{\scriptscriptstyle f^{x_4 ;x_2, x_3 }} "1";"8" };
 {\ar@{<~ }|-{\scriptscriptstyle f^{x_4 ;x_1,x_2,x_3 }} "2";"9" };
 {\ar@{<- }|-{\scriptscriptstyle f^{x_4 ; x_3  }} "3";"10" };
 {\ar@{<-}|-{\scriptscriptstyle f^{x_4 ;x_2  }} "4";"11" };{\ar@{<- }"5";"12"_{\scriptscriptstyle f^{x_4 ; x_1,x_2  }} };
 {\ar@{<--}|-{\scriptscriptstyle f^{x_4  }} "6";"13" };{\ar@{<- }"7";"14"_{\scriptscriptstyle f^{x_4 ; x_1 }}  };( 27,-8)*+{ \scriptstyle  {\rm figure}:\;  \Sigma_+\times[0,x_4] }="30";
\endxy     \xy 0;/r.17pc/:
(0,0 )*+{\bullet}="1";
(60, 0)*+{\bullet}="2";
(-10,20)*+{\bullet}="3";
(20,30)*+{\bullet}="4";
(80,30 )*+{\bullet}="5";
( 10, 50)*+{\bullet}="6";
(  70,50)*+{\bullet}="7";
(60, 80)*+{\bullet}="9";
(80,110 )*+{\bullet}="12";
( 10, 130)*+{\bullet}="13";
(  70,130)*+{\bullet}="14";
( 75,110)*+{(2)}="21";
( 65,75)*+{(1)}="22";
( 30,15)*+{(5)}="23";
( 25,95)*+{(3)}="24";
( 30,35)*+{(4)}="24";
(  5,25)*+{(6)}="26";
 {\ar@{-->}|-{\scriptscriptstyle  {f}^{x_1;x_2,x_3,x_4 }} "1";"2" };
{\ar@{-->} "3";"1"_{\scriptscriptstyle  {f}^{ x_2; x_3,x_4}} };
{\ar@{<--} "3";"6"^{\scriptscriptstyle  {f}^{ x_3;x_4 }} };
{\ar@{->}"5";"2"^{\scriptscriptstyle f^{x_3; x_1 ,x_2,x_4}}};
{\ar@{->}|-{\scriptscriptstyle f^{ x_1;x_2,x_4}} "4";"5" };
{\ar@{->}|-{\scriptscriptstyle f^{x_2 ; x_4}} "6";"4" };
{\ar@{->}|-{\scriptscriptstyle f^{x_3;x_2,x_4}} "4";"1" };
{\ar@{->}|-{\scriptscriptstyle f^{ x_1;x_4  }}"6"; "7"};
{\ar@{->}"7";"5"^{\scriptscriptstyle f^{x_2;x_1,x_4 }}};
{\ar@{~>}|-{\scriptscriptstyle  f^{x_3; x_1 ,x_2}}"12";"9"};
{\ar@{~>}|-{\scriptscriptstyle   f^{ x_1  }}"13"; "14"};
{\ar@{~>}"14";"12"^{\scriptscriptstyle   f^{x_2;x_1 }}};
 {\ar@{<~ }|-{\scriptscriptstyle f^{x_4 ;x_1,x_2,x_3 }} "2";"9" };
{\ar@{<- }"5";"12"_{\scriptscriptstyle f^{x_4 ; x_1,x_2  }} };
 {\ar@{<--}|-{\scriptscriptstyle f^{x_4  }} "6";"13" };
 {\ar@{<- } "7";"14"_{\scriptscriptstyle f^{x_4 ; x_1 }} };( 35,-8)*+{ \scriptstyle  {\rm figure}:\; \hat \Sigma_1 }="30";
\endxy
\end{equation*}The first $3$-arrow is  $(g_*, \hat \Sigma_1, \hat K_1)$   with
 \begin{equation*}\hat
K_1=[( {f}^{x_1}  {f}^{x_2;x_1  })\rhd F^{ x_3, x_4; x_1 ,x_2} ]\rhd' k^{x_1,x_2, x_4},\end{equation*}
whose  $2$-target is $(g_*,\hat \Sigma_2)$.
 The second $3$-arrow is    $(g_*, \hat \Sigma_2, \hat K_2)$, interchanging $2$-arrows      $(2)$ and $ (1)$ in the figure $\hat\Sigma_2$, with
  \begin{equation*}\hat
K_2= {F}^{x_1,x_2 }\rhd'\left \{ ( {F}^{x_1,x_2 })^{-1},( {f}^{x_1}  {f}^{x_2;x_1  })\rhd F^{ x_3 ,x_4 ; x_1 ,x_2} \right\}  ,\end{equation*}
whose $2$-target is $(g_*,\hat\Sigma_3)$.
    The third $3$-arrow is   $(g_*, \hat \Sigma_3, \hat K_3)$, interchanging $2$-arrows      $(4)$ and $ (5)$ in the figure $\hat\Sigma_2$, with
  \begin{equation*}\begin{split}\hat
K_3=& \hat
K_3^0
\rhd' \left[F^{x_2 ,x_4  }\rhd'\left\{ (F^{x_2 ,x_4  })^{-1},   ({f}^{x_2} {f}^{x_4;x_2}  ) \rhd F^{x_1 ,x_3;x_2, x_4 }\right \}^{-1}\right] , \end{split} \end{equation*}
whose  $2$-target is $(g_*,\hat\Sigma_4)$, where  $ \hat
K_3^0={F}^{x_1,x_2 }\cdot  ( {f}^{x_2}  {f}^{ x_1  ;x_2})\rhd F^{x_3 ,x_4  ; x_1 ,x_2} \cdot  {f}^{x_2  } \rhd F^{ x_1 ,x_4  ; x_2}$ is the whiskering corresponding to   the composition of $2$-arrows (2), (1) and (3)  in the figure $\hat\Sigma_2$.
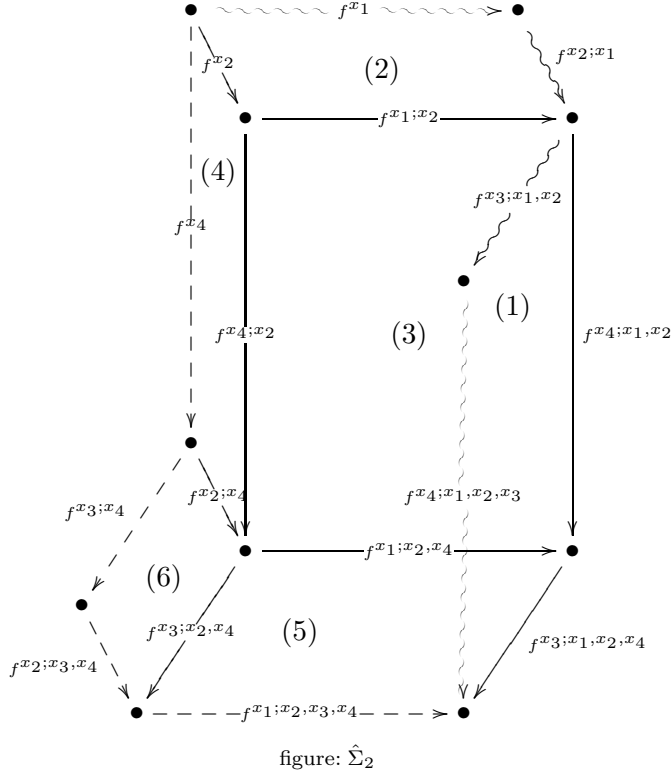
\begin{equation*}
    \xy 0;/r.17pc/:
(0,0 )*+{\bullet}="1";
(60, 0)*+{\bullet}="2";
(-10,20)*+{\bullet}="3";
(20,30)*+{\bullet}="4";
(80,30 )*+{\bullet}="5";
( 10, 50)*+{\bullet}="6";
(60, 80)*+{\bullet}="9";
(20,110)*+{\bullet}="11";
(80,110 )*+{\bullet}="12";
( 10, 130)*+{\bullet}="13";
(  70,130)*+{\bullet}="14";
( 45,119)*+{(2)}="21";
( 69,75)*+{(1)}="22";
( 30,15)*+{(5)}="23";
( 15,100)*+{(4)}="24";
( 50,70)*+{(3)}="24";
(  5,25)*+{(6)}="26";
{\ar@{-->}|-{\scriptscriptstyle  {f}^{x_1;x_2,x_3,x_4 }} "1";"2" };
{\ar@{-->} "3";"1"_{\scriptscriptstyle  {f}^{ x_2; x_3,x_4}} };
{\ar@{<--} "3";"6"^{\scriptscriptstyle  {f}^{ x_3;x_4 }} };
{\ar@{->}"5";"2"^{\scriptscriptstyle f^{x_3; x_1 ,x_2,x_4}}};
{\ar@{->}|-{\scriptscriptstyle f^{ x_1;x_2,x_4}} "4";"5" };
{\ar@{->}|-{\scriptscriptstyle f^{x_2 ; x_4}} "6";"4" };
{\ar@{->}|-{\scriptscriptstyle f^{x_3;x_2,x_4}} "4";"1" };
{\ar@{~>}|-{\scriptscriptstyle  f^{x_3; x_1 ,x_2}}"12";"9"};
{\ar@{->}|-{\scriptscriptstyle   f^{ x_1;x_2}} "11";"12" };
{\ar@{->}|-{\scriptscriptstyle  f^{x_2 }} "13";"11" };
{\ar@{~>}|-{\scriptscriptstyle   f^{ x_1  }}"13"; "14"};
{\ar@{~>}"14";"12"^{\scriptscriptstyle   f^{x_2;x_1 }}};
 {\ar@{<~ }|-{\scriptscriptstyle f^{x_4 ;x_1,x_2,x_3 }} "2";"9" };
 {\ar@{<-}|-{\scriptscriptstyle f^{x_4 ;x_2  }} "4";"11" };{\ar@{<- }"5";"12"_{\scriptscriptstyle f^{x_4 ; x_1,x_2  }} };
 {\ar@{<--}|-{\scriptscriptstyle f^{x_4  }} "6";"13" };( 35,-8)*+{ \scriptstyle  {\rm figure}:\; \hat \Sigma_2 }="30";
\endxy \qquad 
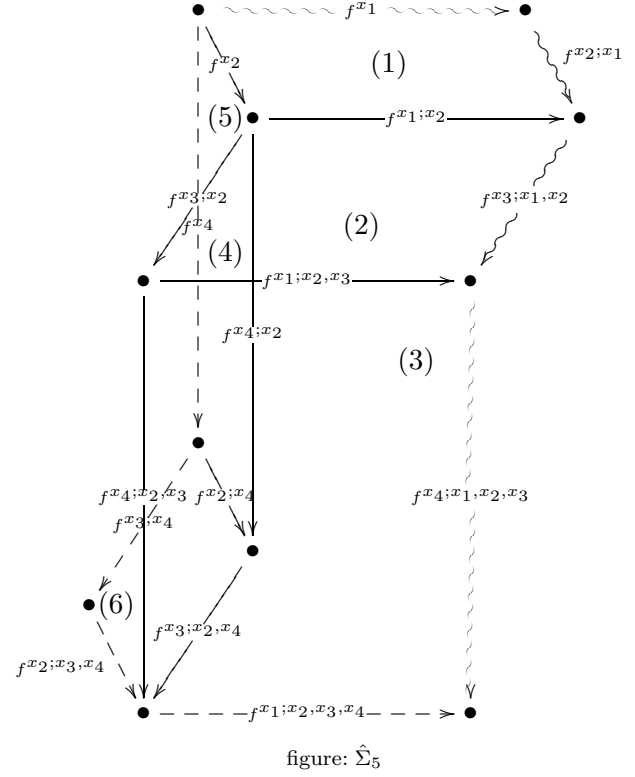
\xy 0;/r.17pc/:
(0,0 )*+{\bullet}="1";
(60, 0)*+{\bullet}="2";
(-10,20)*+{\bullet}="3";
(20,30)*+{\bullet}="4";
( 10, 50)*+{\bullet}="6";
(0,80 )*+{\bullet}="8";
(60, 80)*+{\bullet}="9";
(20,110)*+{\bullet}="11";
(80,110 )*+{\bullet}="12";
( 10, 130)*+{\bullet}="13";
(  70,130)*+{\bullet}="14";
( 40,90)*+{(2)}="21";( 45,120)*+{(1)}="22";
( 15,110)*+{(5)}="23";
( 15,85)*+{(4)}="24";
( 50,65)*+{(3)}="24";
(  -5,20)*+{(6)}="26";
 {\ar@{-->}|-{\scriptscriptstyle  {f}^{x_1;x_2,x_3,x_4 }} "1";"2" };
{\ar@{-->} "3";"1"_{\scriptscriptstyle  {f}^{ x_2; x_3,x_4}} };
{\ar@{<--}|-{\scriptscriptstyle  {f}^{ x_3;x_4 }} "3";"6"};
{\ar@{->}|-{\scriptscriptstyle f^{x_3;x_2,x_4}} "4";"1" };
 {\ar@{->}|-{\scriptscriptstyle   f^{x_1;x_2,x_3 }} "8";"9" };{\ar@{->}|-{\scriptscriptstyle f^{x_2 ; x_4}} "6";"4" };
{\ar@{~>}|-{\scriptscriptstyle  f^{x_3; x_1 ,x_2}}"12";"9"};
{\ar@{->}|-{\scriptscriptstyle   f^{ x_1;x_2}} "11";"12" };
{\ar@{->}|-{\scriptscriptstyle  f^{x_2 }} "13";"11" };
{\ar@{->}|-{\scriptscriptstyle    f^{x_3;x_2}} "11";"8" };
{\ar@{~>}|-{\scriptscriptstyle   f^{ x_1  }}"13"; "14"};
{\ar@{~>}"14";"12"^{\scriptscriptstyle   f^{x_2;x_1 }}};
{\ar@{<- }|-{\scriptscriptstyle f^{x_4 ;x_2, x_3 }} "1";"8" };
 {\ar@{<~ }|-{\scriptscriptstyle f^{x_4 ;x_1,x_2,x_3 }} "2";"9" };
  {\ar@{<-}|-{\scriptscriptstyle f^{x_4 ;x_2  }} "4";"11" };
 {\ar@{<--}|-{\scriptscriptstyle f^{x_4  }} "6";"13" };( 35,-8)*+{ \scriptstyle  {\rm figure}:\; \hat \Sigma_5 }="30";
\endxy
\end{equation*}

The fourth $3$-arrow is $(g_*, \hat \Sigma_4, \hat K_4)$ with
\begin{equation*}\hat
K_4=  {F}^{x_1,x_2 } \rhd' [ {f}^{x_2}\rhd k^{x_1,x_3,x_4  ;x_2}] ,\end{equation*}
 whose   $2$-target is $(g_*,\hat\Sigma_5)$.
  The fifth  $3$-arrow   is $(g_*, \hat \Sigma_5, \hat K_5 )$ with
\begin{equation*}\hat
K_5= [ F^{  x_1, x_2}\cdot   {f}^{x_2} \rhd(  F^{x_1,x_3;x_2  }   \cdot    {f}^{x_3;x_2 }  \rhd  F^{ x_1 ,x_4;x_2,x_3  } )] \rhd'  k^{  x_2,x_3 ,x_4}  ,\end{equation*}
whose  $2$-target  is $(g_*,\hat\Sigma_6)$, where the part before $ \rhd'$ is the whiskering corresponding to  the composition of $2$-arrows (1), (2) and (3)  in the figure $\hat\Sigma_5$.

The last $3$-arrow is $(g_*, \hat \Sigma_6, \hat K_6 )$, interchanging $2$-arrows         $(4)$ and $ (3)$ in the figure $\hat\Sigma_6$, with
  \begin{equation*}\hat
K_6=  [  F^{  x_1, x_2}\cdot   {f}^{x_2} \rhd F^{x_1,x_3;x_2  } ] \rhd'\left[{F}^{x_2 ,x_3  }  \rhd'\left\{ ( {F}^{x_2 ,x_3  })^{-1}, (  {f}^{x_2} {f}^{  x_3; x_2})\rhd  F^{x_1 ,x_4 ; x_2, x_3}\right \}\right ] , ,\end{equation*}
 whose   $2$-target  is $(g_*,\Sigma_8)$ in the last subsection.
 The $0$-th $3$-arrow is
\begin{equation*}
 \hat K_0=\pi_L \circ {F}_3(\Omega_{x_1,x_2,x_3;x_4})=\left[{f}^{x_1}\rhd\left(  {f}^{x_2;x_1  } \rhd F^{x_3,x_4;x_1 ,x_2}\cdot    F^{ x_2,x_4;x_1 }\right )\cdot F^{ x_1 ,x_4 }\right]\rhd' ({f}^{x_4} \rhd{k}^{x_1,x_2,x_3;x_4}),
\end{equation*}
 where the part before $ \rhd'$ is the whiskering corresponding to  the composition of $2$-arrows (1), (2) and (3) in the figure $\hat \Sigma_1$.

The derivative  $\frac {\partial^3\hat K_0 \hat  K_1 \cdots\hat
 K_6}{\partial x_1\partial x_2\partial x_3\partial x_4}$ at $(0,0,0,0)$ gives
\begin{equation}\label{eq:C+}\begin{split}&v_4C(v_1,v_2,v_3)+A( v_4 ) \rhd C(v_1,v_2,v_3) -
  \{B(v_1,v_2), B(v_3,v_4 )\}+\{B(v_2,v_4), B(v_1,v_3)\}\\&+    v_2C(v_1,v_3,v_4)+A( v_2 ) \rhd C(v_1,v_3,v_4)-\{B(v_2,v_3), B(v_1,v_4)\}.
\end{split}\end{equation}
 \begin{equation*}
\xy 0;/r.17pc/:
(0,0 )*+{\bullet}="1";
(60, 0)*+{\bullet}="2";
(-10,20)*+{\bullet}="3";
( 10, 50)*+{\bullet}="6";
(0,80 )*+{\bullet}="8";
(60, 80)*+{\bullet}="9";
(-10,100)*+{\bullet}="10";
(20,110)*+{\bullet}="11";
(80,110 )*+{\bullet}="12";
( 10, 130)*+{\bullet}="13";
(  70,130)*+{\bullet}="14";
( 40,90)*+{(2)}="21";( 45,120)*+{(1)}="22";
(  5,60)*+{(6)}="23";
(  5,105)*+{(4)}="24";
( 30,40)*+{(3)}="24";
(  -5,65)*+{(5)}="26";
 {\ar@{-->}|-{\scriptscriptstyle  {f}^{x_1;x_2,x_3,x_4 }} "1";"2" };
{\ar@{-->} "3";"1"_{\scriptscriptstyle  {f}^{ x_2; x_3,x_4}} };
{\ar@{<--} "3";"6"^{\scriptscriptstyle  {f}^{ x_3;x_4 }} };
 {\ar@{->}|-{\scriptscriptstyle   f^{x_1;x_2, x_3}} "8";"9" };
{\ar@{->}|-{\scriptscriptstyle  f^{ x_2; x_3}} "10";"8" };
{\ar@{<-} "10";"13"^{\scriptscriptstyle   f^{ x_3 }} };
{\ar@{~>}|-{\scriptscriptstyle  f^{x_3; x_1 ,x_2}}"12";"9"};
{\ar@{->}|-{\scriptscriptstyle   f^{ x_1;x_2}} "11";"12" };
{\ar@{->}|-{\scriptscriptstyle  f^{x_2 }} "13";"11" };
{\ar@{->}|-{\scriptscriptstyle    f^{x_3;x_2}} "11";"8" };
{\ar@{~>}|-{\scriptscriptstyle   f^{ x_1  }}"13"; "14"};
{\ar@{~>}"14";"12"^{\scriptscriptstyle   f^{x_2;x_1 }}};
{\ar@{<- }|-{\scriptscriptstyle f^{x_4 ;x_2, x_3 }} "1";"8" };
 {\ar@{<~ }|-{\scriptscriptstyle f^{x_4 ;x_1,x_2,x_3 }} "2";"9" };
 {\ar@{<- }|-{\scriptscriptstyle f^{x_4 ; x_3  }} "3";"10" };
  {\ar@{<--}|-{\scriptscriptstyle f^{x_4  }} "6";"13" };( 35,-8)*+{ \scriptstyle  {\rm figure}:\; \hat \Sigma_6 }="30";
\endxy
\end{equation*}
\subsection{The covariance of the $3$-dimensional  holonomies }
 Denote $K_8=\hat K_6^{-1},\ldots, K_{14}=\hat K_0^{-1}$.
The {\it  $3$-dimensional  holonomy} is defined as $\mathscr H_F=
K_1K_2\cdots K_{14}$. Then
\begin{equation*}
  \frac {\partial^4\mathscr H_F}{\partial x_1\partial x_2\partial x_3\partial x_4} (0,0,0,0) =(\ref{eq:C-})  -  (\ref{eq:C+}) =\Omega_3.
\end{equation*}

In our construction above, $\partial\Theta$ is the composition of $14$  $3$-arrows, say $\vartheta_1,\ldots,\vartheta_{14}$, with the $2$-target of $\vartheta_j$  coinciding with the $2$-source of $\vartheta_{j+1}$'s , i.e.,
\begin{equation}\label{eq:s=t}
     t_2(\vartheta_1)=s_2(\vartheta_2),\ldots,t_2(\vartheta_{14})=s_2(\vartheta_1).
\end{equation} Denote
 $\widetilde{F}(\vartheta_j)=(*,*,\widetilde{K}_j)$ and $F(\vartheta_j)=(*,*,K_j)$.
The $0$- and $1$-source of $\vartheta_j$ are independent of $j$. Denote
$f:=s_1(\vartheta_j)$, $c:=s_0(\vartheta_j)$.
The naturality (\ref{eq:naturality0})-(\ref{eq:naturality1}) of the lax-natural transformation $\Psi:F\rightarrow\widetilde{F}$ implies
\begin{equation}\label{eq:naturality*}\Psi_{s_2(\vartheta_j)}\cdot  \widetilde{F}_3(\vartheta_j )
   = \Psi_f \rhd'\left[\Psi_c\rhd F_3( \vartheta_j)\right]\cdot \Psi_{t_2(\vartheta_j)}.
\end{equation}
Thus, $\Psi_f \rhd'\left[\Psi_c \rhd K_j\right]=\Psi_{s_2(\vartheta_j)}\cdot \widetilde{K}_j\cdot \Psi_{t_2(\vartheta_j)}^{-1}$, which implies that
\begin{equation*}
 \Psi_f \rhd'\left[\Psi_c\rhd (K_1K_2\cdots K_{14}) \right] =\Psi_{s_2(\vartheta_1)}\widetilde{K}_1\widetilde{K}_2\cdots \widetilde{K}_{14}\Psi_{s_2(\vartheta_1)}^{-1}.
\end{equation*} by (\ref{eq:s=t})  and both $\rhd$ and $\rhd'$ being automorphisms.
This is the covariance   of the $3$-dimensional  holonomy under lattice $3$-gauge transformations.

\section{Discussion}
The $3$-dimensional  holonomy is $3$-gauge invariant. We can use the construction in section 5 and 6 to give the construction of a
non-Abelian $3$-form lattice gauge theory.
 It is interesting to give a
lattice $3$-BF theory (cf. \cite{P03} \cite{GPP} for the $2$-gauge   case),
a combinatorial construction of  the  topological higher
gauge theory as a state sum model. These models are expected to be trivially renormalizable, i.e.,
 independent of the chosen triangulation. Then they will give topological invariants of manifolds.

In the standard lattice gauge theory, the most general  gauge invariant expressions  are
{\it spin networks}, the generalization  of Wilson loops that includes branchings of the
lines with intertwiners of the gauge group at the branching points. The most general $2$-gauge invariant expressions will be given by coloured branched surfaces, i.e., by some sort of {\it spin foams}. It is quite interesting to consider the most general $3$-gauge invariant expressions.

Barrett-Crane-Yetter state sum model of quantum gravity \cite{BC}  is equivalent to a lattice $2$-gauge theory. How about  the $3$-version counterpart of  quantum gravity?

\end{document}